%% file: sparse_mimo_dof.tex
\documentclass[journal,a4paper,onecolumn,draftcls]{IEEEtran}
\usepackage{amsmath,amssymb,mathtools}
\usepackage{winsnotation}
\usepackage{booktabs,array,enumitem}
\usepackage{xcolor}
\usepackage{pgfplots}
\usepackage{needspace}

\pgfplotsset{compat=1.18}

\usepackage{cite}
\usepackage[colorlinks=true,linkcolor=blue,citecolor=blue,urlcolor=blue]{hyperref}
\usepackage[nameinlink,capitalise,noabbrev]{cleveref}
\usepackage{mathrsfs}
\usepackage{acronym}
\usepackage{tikz}
\usetikzlibrary{arrows.meta,calc,decorations.markings,shapes.geometric}
\crefname{theorem}{Theorem}{Theorems}
\crefname{proposition}{Proposition}{Propositions}
\crefname{lemma}{Lemma}{Lemmas}
\crefname{corollary}{Corollary}{Corollaries}
\crefname{definition}{Definition}{Definitions}
\crefname{assumption}{Assumption}{Assumptions}

\newtheorem{theorem}{Theorem}
\newtheorem{proposition}{Proposition}
\newtheorem{lemma}{Lemma}
\newtheorem{corollary}{Corollary}
\newtheorem{definition}{Definition}
\newtheorem{assumption}{Assumption}
\newtheorem{remark}{Remark}
\newtheorem{sublemma}{Sublemma}
\crefname{sublemma}{Sublemma}{Sublemmas}

\newcommand{\C}{\mathbb C}
\newcommand{\R}{\mathbb R}
\newcommand{\CN}{\mathcal{CN}}
\newcommand{\rank}{\operatorname{rank}}
\newcommand{\diag}{\operatorname{diag}}

\newcommand{\dist}{\operatorname{dist}}

\newcommand{\tr}{\operatorname{tr}}
\newcommand{\realify}{\operatorname{realify}}
\newcommand{\Real}{\operatorname{Re}}
\newcommand{\Imag}{\operatorname{Im}}
\newcommand{\jmathu}{\mathrm j}
\newcommand{\quant}[2]{\langle #1 \rangle_{#2}}

\acrodef{aoa}[AoA]{angle of arrival}
\acrodef{aod}[AoD]{angle of departure}
\acrodef{ula}[ULA]{uniform linear array}
\acrodef{upa}[UPA]{uniform planar array}
\acrodef{nula}[NULA]{nonuniform linear array}
\acrodefplural{nula}[NULAs]{nonuniform linear arrays}
\acrodef{dof}[DoF]{degree of freedom}
\acrodef{snr}[SNR]{signal-to-noise ratio}
\acrodef{sv}[SV]{Saleh-Valenzuela}
\acrodef{isac}[ISAC]{integrated sensing and communication}
\acrodef{mimo}[MIMO]{multiple-input-multiple-output}
\acrodef{6g}[6G]{sixth-generation}
\acrodef{csi}[CSI]{channel state information}
\acrodef{mse}[MSE]{mean-squared error}
\acrodef{lmmse}[LMMSE]{linear minimum mean-squared error}

\acrodefplural{dof}[DoFs]{degrees of freedom}
\acrodefplural{aod}[AoDs]{angles of departure}
\acrodefplural{aoa}[AoAs]{angles of arrival}

\title{Exact Degrees of Freedom of Spatially Sparse MIMO Channels Without Prior CSI}
\author{Yifeng Xiong, \IEEEmembership{Member, IEEE}, Weijiang Zhao, Fan Liu, \IEEEmembership{Senior Member, IEEE}, \\Shi Jin, \IEEEmembership{Fellow, IEEE}, and Jianhua Zhang, \IEEEmembership{Fellow, IEEE}}

\begin{document}
\maketitle

\begin{abstract}
We characterize the \ac{dof} of a point-to-point blockwise memoryless channel without prior \ac{csi}, with a fixed number \(K\) of propagation paths, where the transmitter (Tx) and the receiver (Rx) are equipped with \acp{nula} of \(N_t\) and \(N_r\) antennas, respectively. The positions of array elements are fixed, known, pairwise distinct, and need not be equally spaced. The \ac{ula} is a special case. In each block of length \(T\), the continuous \acp{aoa}, \acp{aod}, and independent complex Gaussian path gains are redrawn. Both Tx and Rx know the state distributions but are not given the current realizations before transmission. The receiver may estimate the channel from reference signals or decode without explicit channel estimation, with reference symbols counted in $T$ and their energy counted against the power constraint. Under the aforementioned model, we show that the \ac{dof} is 
$$
{\rm DoF}=\begin{cases}
            1-\frac{1}{T}, & \mbox{$K=1$}, \\
            K\Big(1-\frac{3}{2T}\Big), & \mbox{$K\geq 2$},
          \end{cases}
$$
when \(N_r\ge K+1\), \(N_t\ge\max\{K,2\}\), and \(T\ge K\). The analytical results are further demonstrated by their applications to the \ac{dof} tradeoff analysis in \ac{isac}. For more general array structures, an achievability result is established, while the converse remains open in general. 
\end{abstract}


\section{Introduction}
\subsection{Background and Related Works}
The increasing antenna array size and transmission bandwidth in wireless systems make the acquisition of channel knowledge an important part of communication system design~\cite{Larsson2014,Alkhateeb2014,Venugopal2017}. In propagation environments dominated by a limited number of resolvable paths, the channel coefficients are highly correlated through common physical parameters. In particular, each path contributes jointly across antennas, frequencies, and observation times through its gain, directions, and delay. This structure motivates sparse channel models that describe a large collection of coefficients using a small number of propagation components \cite{Sayeed2002,Bajwa2010,Heath2016}. Exploiting the structure can reduce the cost of channel estimation, and also raises the question of how it affects the fundamental communication limit when the channel realization is unknown.

Sparse channel models take diverse forms according to the propagation scenario and signal representation. Far-field spatial domain models use the angular responses of a few paths, while sparse delay-domain models describe frequency-selective propagation~\cite{Sayeed2002,Bajwa2010}. For electrically large apertures in the near field, the response of a path can depend on both angle and distance, leading to a joint angle-distance representation \cite{CuiDai2022}. Different array geometries further produce different response functions~\cite[Ch.~2]{VanTrees2002}. When a smooth parameterization has fewer real coordinates than the ambient channel space, the possible channel realizations live on lower-dimensional manifolds. 

To make the problem explicit, consider the narrowband array observation model
\[
 \RM{Y}_\rho=\sqrt\rho\,\RM{H}\RM{X}+\RM{N},
\]
where \(\RM{X}\), \(\RM{H}\), \(\RM{Y}_\rho\), and \(\RM{N}\) denote the transmitted signal, channel matrix, received signal, and independent Gaussian noise, respectively, and \(\rho\) is the \ac{snr} parameter. The Tx and the Rx know the channel distribution but not the current realization of \(\RM{H}\) before transmission. The Tx encodes information into \(\RM{X}\) independently of \(\RM{H}\), and the Rx decodes the message from \(\RM{Y}_\rho\) under this channel uncertainty. For a given blockwise memoryless channel and power constraint, capacity $C(\rho)$ is the supremum of the mutual information $I(\RM{Y}_\rho;\RM{X})$ per channel use over all admissible input distributions at each \ac{snr} $\rho$. The channel \ac{dof} is the pre-log of this optimized rate~\cite{Ngo2021,StotzBolcskei2016}, in the sense that
$$
{\rm DoF}=\mathop{\lim\sup}_{\rho\rightarrow\infty} \frac{C(\rho)}{\log \rho}.
$$
A specified input distribution would instead yield an achievable mutual-information pre-log. Establishing that an input is pre-log optimal, or \ac{dof}-achieving, requires a matching converse over all admissible input distributions. We would like to highlight that known reference signals can be included as deterministic components of $\RM{X}$. The capacity optimization therefore also covers reference-assisted channel estimation and decoding, with all reference symbols included in $T$ and their energy included in the power constraint.

Considerable effort has been devoted to acquiring sparse channel knowledge from training signals. Compressed channel sensing exploits the dependence among the coefficients to reduce the resources needed for reconstruction \cite{Bajwa2010,Heath2016}. More relevant to the continuous spatial model studied here, Shrestha \emph{et al.} established recovery guarantees for a \ac{mimo} channel with off-grid angles at both the Tx and the Rx from compressed and quantized pilot observations \cite{Shrestha2025}. In these channel estimation problems, the probe matrix is known, hence the noiseless observation varies only with the channel. In communication without prior \ac{csi}, the data-bearing part of the transmitted block is unknown to the Rx, while other entries may be assigned to known reference signals. If the reference signals are insufficient, certain changes in the channel state can then be compensated by changes in the data, leaving the noiseless observation unchanged~\cite{ZhengTse2002}. Consequently, recovering the geometric parameters from known probes and conveying unknown data through the same channel involve different identifiability questions. The communication analysis must account for the joint roles of the data and channel state. 

The information-theoretic study of channel capacity without \acf{csi} constitutes another line of related work. For independent Rayleigh block fading, the seminal work of Zheng and Tse characterized a signaling rank that attains the capacity pre-log, revealing a balance between spatial multiplexing and channel uncertainty \cite{ZhengTse2002}. Subsequent studies established the influence of temporal correlation and more general channel distributions \cite{Koliander2014,Ngo2021}. In particular, Ngo, Yang, and Guillaud derived the optimal \ac{dof} for generic block-fading channel matrices with finite power and finite differential entropy in the ambient matrix space \cite{Ngo2021}. However, a continuous spatially sparse channel may instead have a singular distribution in that space~\cite{Koliander2016}. Its capacity pre-log therefore requires an analysis that respects the geometric constraints on the channel realizations.

Several results further characterize the benefits of particular forms of sparsity structure. For sparse frequency-selective block fading with a fixed finite collection of possible delay supports, support uncertainty contributes no pre-log penalty by itself, while the unknown complex coefficients determine the leading learning cost \cite{KannuSchniter2011}. Angular sparsity has been exploited in blind multiuser massive-\ac{mimo} achievability \cite{ZhangYuanZhang2018}, and known transmit correlation subspaces can reduce noncoherent overhead \cite{ZhangNgoYangNosratinia2022}. Random specular components and blockwise constant phase noise provide additional examples in which the uncertainty structure affects capacity or its high-\ac{snr} expansion \cite{Godavarti2003,Durisi2012}. Nevertheless, their different assumptions about support, subspaces, and continuous parameters do not give a common rule for determining the \ac{dof} of continuous geometric channels.

Against the aforementioned background, the central remaining issue is thus how to characterize the capacity pre-log when both the geometric state and the transmitted data are unknown. We would like to highlight that counting all channel parameters only gives an incomplete description of the achievable pre-log of a candidate input distribution, because some state variations remain distinguishable in the joint output, whereas others can be absorbed into the data or an effective gain~\cite{KarzandZheng2014,ZhengTse2002}. To elaborate, projecting a steering vector onto a transmit subspace can retain some directional variations and suppress others~\cite[Chs.~2 and~8]{VanTrees2002}. Input design thus affects the achievable pre-log by determining the number of transmitted streams and the channel uncertainty visible to those streams. The remaining task is to match an achievable pre-log with an upper bound uniform over all admissible input distributions at each SNR.

\subsection{Contributions of this Paper}
To investigate the above problem, we consider a point-to-point \ac{nula}-\ac{sv} channel with a fixed number of paths~\cite{SalehValenzuela1987,Alkhateeb2014}. The model retains the continuous arrival and departure directions at both ends of the link, together with independent complex Gaussian path gains. The complete state is redrawn independently from block to block, yielding a blockwise memoryless channel whose law is known to both terminals. The known array positions are arbitrary and pairwise distinct, including equally spaced ULA geometries. The block length and array sizes remain fixed as \ac{snr} grows. The \ac{nula}-\ac{sv} model thus provides a representative setting for an explicit output-geometry calculation and a matching capacity-pre-log converse.

Our analysis connects the geometry of the received signal to the communication objective in two steps. For an input distribution fixed as \ac{snr} grows and satisfying the stated regularity conditions, the achievable pre-log per channel use is the joint-minus-conditional real information-dimension difference divided by twice the block length. Under explicit regularity conditions, these dimensions can be evaluated through Jacobian ranks \cite{WuVerdu2010,StotzBolcskei2016,Koliander2016}. This approach builds on geometric analyses of structured communication models \cite{KarzandZheng2014,ChengYu2024}. To establish the capacity pre-log, we then derive an upper bound uniform over all admissible input distributions at each \ac{snr}. In particular, distributions approaching the supremum may depend on \ac{snr} and concentrate near degenerate inputs~\cite{Ngo2021,Morgenshtern2013}. This distinction allows us to identify pre-log-optimal signaling without assuming that one fixed distribution achieves capacity at every finite \ac{snr}.

Under this setting, our main contributions are summarized as follows.
\begin{enumerate}
 \item We determine the exact capacity pre-log under the stated conditions. It equals \(1-1/T\) for \(K=1\) and \(K(1-3/(2T))\) for \(K\ge2\). We show that these values are attained by Gaussian signaling through a single antenna and through a fixed \(K\)-dimensional subarray, respectively. 
 \item We reveal the geometric origin of the noncoherent loss. For the multistream input, the \ac{aoa} contribution cancels between the joint and fixed-data output dimensions, while the complex gains and visible departure geometry determine the remaining loss.
 \item We prove a matching converse covering \ac{snr}-dependent input distributions, including mixtures approaching rank-1 matrices at multiple scales. It also bounds the \ac{dof} of reference-signal-assisted schemes. For the NULA-SV model, the second singular value controls both \ac{aod} resolvability and proximity to the rank-1 set. The resulting conditional output entropy lower bound and output-entropy upper bound are then combined to yield the complete \ac{dof} characterization.
 \item We extend the achievability analysis to general separable finite-ray channels with fixed precoders. We derive an achievable pre-log formula for Gaussian signaling and give sufficient conditions based on path response matrices and array response derivatives. We illustrate these results through \ac{upa} antenna selection and Fourier-mode selection on rectangular continuous apertures, showing how the transmit subspace controls the directional uncertainty that remains distinguishable after accounting for unknown path gains.
\end{enumerate}

The rest of this paper is organized as follows. Section~\ref{sec:model-main-result} presents the model and main result. Sections~\ref{sec:achievability} and~\ref{sec:converse} prove achievability and the converse. Section~\ref{sec:projected-geometry} discusses extensions to other array structures and some practical implications of our results, and we conclude the paper in Section~\ref{sec:conclusion}.

\subsection*{Notations}
Throughout this paper, $\rv{a}$, $\RV{a}$, $\RM{A}$, and $\RS{A}$ denote random scalars, vectors, matrices, and sets, respectively. Their deterministic counterparts are denoted by $a$, $\V{a}$, $\M{A}$, and $\Set{A}$. The same convention applies to finite multisets. The imaginary unit is denoted as $\jmathu=\sqrt{-1}$. The zero matrix and $n$-dimensional identity matrix are denoted by $\M{0}$, and $\M{I}_n$, with dimensions omitted when clear from the context.

The superscripts $T$, $*$, $H$, and $\dagger$ denote transpose, complex conjugation, Hermitian transpose, and the Moore-Penrose pseudoinverse, respectively. The notation $\diag(\V{a})$ denotes the diagonal matrix with diagonal entries given by $\V{a}$, and $\otimes$ denotes the Kronecker product. The operators $\operatorname{vec}$, $\Real$, and $\Imag$ denote column stacking, the real part, and the imaginary part, respectively. The realification $\realify(\M{Z})=(\Real\operatorname{vec}\M{Z},  \Imag\operatorname{vec}\M{Z})$ collects the real and imaginary parts of $\operatorname{vec}\M{Z}$ into a real vector. The notation $\|\V{a}\|_2$ denotes the Euclidean norm, $\|\M{A}\|_2$ the spectral norm, and $\|\M{A}\|_F$ the Frobenius norm. The singular values $\sigma_i(\M{A})$ are ordered nonincreasingly. The notation $\dist_F(\M{A},\Set{A})$ denotes distance to a set in the Frobenius norm. The operators $\rank$, $\ker$, and $\operatorname{col}$ denote rank, kernel, and column space, respectively. For Hermitian matrices, $\M{A}\succeq\M{B}$ means that $\M{A}-\M{B}$ is positive semidefinite, with $\succ$ denoting positive definiteness.

The operators $\E{\cdot}$, $\Prob{\cdot}$, and $\Var{\cdot}$ denote expectation, probability, and variance, respectively. The notation $P_{\RV{x}}$ denotes the probability distribution of $\RV{x}$, and $\perp$ between random quantities denotes independence. The symbols $H(\cdot)$, $h(\cdot)$, and $I(\cdot;\cdot)$ denote discrete entropy, differential entropy, and mutual information, respectively. $H_2(p)$ denotes the entropy of a Bernoulli random variable $\rv{x}$ with $\mathbb{P}\{\rv{x}=0\}=p$. The floor operation $\lfloor\cdot\rfloor$ is applied componentwise to vectors. All logarithms are natural, $\log^+x=\max\{0,\log x\}$ for $x>0$, and $\log^-x=\log^+x-\log x$. The symbols $O(\cdot)$, $o(\cdot)$, and $\Theta(\cdot)$ have their usual asymptotic meanings, with the limiting variable specified by the context. Jacobian ranks, manifold dimensions, and information dimensions are taken over $\R$ after realification. The notation $D_{\V{u}}f$ denotes the Jacobian with respect to $\V{u}$, with other arguments held fixed.

\section{Channel Model and the Main Result}
\label{sec:model-main-result}
We consider fixed linear arrays with arbitrary distinct element positions~\cite[Ch.~2]{VanTrees2002}. For $\nu\in\{r,t\}$, let $x_{\nu,1},\ldots,x_{\nu,N_\nu}$ denote the known positions in wavelengths. Their values do not change with SNR. The steering vector is
\begin{equation}
 \V{a}_\nu(u)=\frac1{\sqrt{N_\nu}}
 [e^{\jmathu2\pi x_{\nu,1}u},\ldots,
  e^{\jmathu2\pi x_{\nu,N_\nu}u}]^T,
 \qquad \nu\in\{r,t\}.
 \label{eq:nula-steering}
\end{equation}
We refer to this general geometry as the NULA model, allowing uniform spacing as a special case. In particular, a \ac{ula} is obtained by $x_{\nu,n}=(n-1)d_\nu$ for a fixed spacing $d_\nu>0$ in wavelengths. The half-wavelength ULA corresponds to $d_\nu=1/2$.

We choose the origin of each array so that $x_{r,1}=x_{t,1}=0$. Indeed, translating either origin changes each ray response only by a scalar phase, which can be absorbed into its complex gain. Conditional on all directions, the independent circular Gaussian gains in \cref{ass:ula} retain their joint distribution after this phase change. They therefore remain independent of the directions. This normalization does not restrict the array geometry.

The physical rays have no intrinsic ordering. We therefore represent the geometry in block $b$ by the random finite multiset
\begin{equation}
 \RS{V}_b
 =\bigl\{\!\bigl\{
 (\rv{u}_{r,b,k},\rv{u}_{t,b,k}):k=1,\ldots,K
 \bigr\}\!\bigr\},
 \label{eq:unordered-geometry}
\end{equation}
and the complete marked ray state by
\begin{equation}
 \RS{S}_b
 =\bigl\{\!\bigl\{
 (\rv{u}_{r,b,k},\rv{u}_{t,b,k},\rv{\alpha}_{b,k}):
 k=1,\ldots,K
 \bigr\}\!\bigr\}.
 \label{eq:unordered-ray-state}
\end{equation}
The multiset notation preserves the \ac{aoa}-\ac{aod}-gain association within each ray. Under the continuous prior imposed below, collisions occur with probability zero, thus the multiset is an ordinary finite set almost surely.

For a deterministic realization $\Set{S}=\{\!\{(u_{r,k},u_{t,k},\alpha_k)\}_{k=1}^K\!\}$,
define the deterministic channel map
\begin{equation}
 \M{H}(\Set{S})
=\sum_{(u_r,u_t,\alpha)\in\Set{S}}
 \alpha\V{a}_r(u_r)\V{a}_t(u_t)^H .
 \label{eq:deterministic-channel-map}
\end{equation}
The random channel matrix in block $b$ is
\begin{equation}
 \RM{H}_b=\M{H}(\RS{S}_b).
\label{eq:sv-channel}
\end{equation}
In what follows, if it is clear from the context, we omit the subscript $b$.

Under the aforementioned setting, we make the following assumptions.
\begin{assumption}[NULA-SV model]
\label{ass:ula}
\mbox{}\par
\begin{enumerate}[label=(U\arabic*)]
 \item For each $\nu\in\{r,t\}$, the element positions in  \eqref{eq:nula-steering} are fixed, known, and pairwise distinct. Fix a closed direction-cosine interval
 \[
  \Set{U}_\nu
  =
  [\underline u_\nu,\overline u_\nu]
  \subseteq[-1,1],
  \qquad
  0<\overline u_\nu-\underline u_\nu\le2.
 \]
 The $2K$ directional variables $\{\rv{u}_{r,k},\rv{u}_{t,k}:k=1,\ldots,K\}$ are mutually independent and satisfy $\Prob{\rv{u}_{\nu,k}\in\Set{U}_\nu}=1$. For each $(\nu,k)$, the random variable $\rv{u}_{\nu,k}$ has a density  $f_{\rv{u}_{\nu,k}}$ with respect to Lebesgue measure on $\Set{U}_\nu$, where
\[
  f_{\rv{u}_{\nu,k}}\in C(\Set{U}_\nu),
  \qquad
  \int_{\Set{U}_\nu}f_{\rv{u}_{\nu,k}}(u)\,du=1.
 \]
 There exist constants $0<f_{\min}\le f_{\max}<\infty$, independent of  $(\nu,k)$, such that
 \[
  f_{\min}
  \le f_{\rv{u}_{\nu,k}}(u)
  \le f_{\max},
  \qquad
  u\in\Set{U}_\nu.
 \]
 Equivalently, the joint density of all directional variables factors as
 \[
  f_{\RV{u}_r,\RV{u}_t}(\V{u}_r,\V{u}_t)
  =
  \prod_{k=1}^{K}
  f_{\rv{u}_{r,k}}(u_{r,k})
  f_{\rv{u}_{t,k}}(u_{t,k}).
 \]
 \item The gains satisfy  $\rv{\alpha}_k\stackrel{\rm iid}{\sim}\CN(0,\sigma_\alpha^2)$, $0<\sigma_\alpha^2<\infty$, and are independent of the directions and the input.
 \item $K$, $T$, $N_t$ and $N_r$ satisfy
 \[
  N_r\ge K+1,\qquad N_t\ge\max\{K,2\},\qquad T\ge K .
 \]
\end{enumerate}
\end{assumption}

Although the physical ray state in \eqref{eq:unordered-ray-state} is unordered, \cref{ass:ula} specifies its distribution through the indexed random variables $\{(\rv{u}_{r,k},\rv{u}_{t,k},\rv{\alpha}_k):k=1,\ldots,K\}$. We use these indices only as auxiliary Euclidean coordinates throughout the analysis. Accordingly, we write $\V{u}_r$, $\V{u}_t$, and $\V{\alpha}$ for deterministic coordinate vectors under these auxiliary indices. We then write
\begin{align*}
 \M{A}_r(\V{u}_r)
 &=[\V{a}_r(u_{r,1}),\ldots,\V{a}_r(u_{r,K})],~~\M{A}_t(\V{u}_t)
 =[\V{a}_t(u_{t,1}),\ldots,\V{a}_t(u_{t,K})],~~\M{D}(\V{\alpha})
 =\diag(\alpha_1,\ldots,\alpha_K),
\end{align*}
so that
\begin{equation}
 \M{H}(\V{u}_r,\V{u}_t,\V{\alpha})
 =\M{A}_r(\V{u}_r)\M{D}(\V{\alpha})
 \M{A}_t(\V{u}_t)^H.
 \label{eq:ordered-channel-chart}
\end{equation}
The transmitter and the receiver know the model and the \textit{a priori} distributions but not the realization of $\RS{S}$ in the current block. Omitting the block index, the communication observation is
\begin{equation}
 \RM{Y}_\rho=\sqrt\rho\,\RM{H}\RM{X}+\RM{N},
\label{eq:channel}
\end{equation}
The noise entries satisfy \(\rv{n}_{it}\stackrel{\rm iid}{\sim}\CN(0,1)\), independently of the complete channel state and the input. The input obeys \(\E{\|\RM{X}\|_F^2}\le T\) and is independent of the channel state. We define\footnote{Throughout this paper, all logarithms are natural, and rates are measured in nats per channel use.}
\begin{align*}
 C_{\rm blind}(\rho)
 &=\frac1T
 \sup_{P_{\RM{X}}:\,\E{\|\RM{X}\|_F^2}\le T}
 I(\RM{X};\RM{Y}_\rho),\\
 d_{\rm blind}^{\rm cap}
 &=\limsup_{\rho\to\infty}
 \frac{C_{\rm blind}(\rho)}{\log\rho}.
\end{align*}
Formally, the state blocks \(\{\RS{S}_b\}_{b\ge1}\) are independent and identically distributed (i.i.d.), and the noise blocks \(\{\RM{N}_b\}_{b\ge1}\) are i.i.d. The two sequences are mutually independent and independent of the message and the encoder's private information. The input blocks need not be independent of one another. Consequently, the channel is blockwise memoryless, with transition distribution
\[
 P_{\RM{Y}^{1:B}\mid\RM{X}^{1:B}}
 =\prod_{b=1}^{B}P_{\RM{Y}_b\mid\RM{X}_b}.
\]
Although coding may span arbitrarily many blocks, capacity is therefore given by the supremum of the single-block mutual information. The channel model is portrayed in Fig.~\ref{fig:nula-channel-model}.

\begin{figure}[t]
 \centering
 \resizebox{0.9\textwidth}{!}{\input{nula_channel_tikz.tex}}
 \vspace{-3mm}
 \caption{Spatially sparse NULA-SV channel and its blockwise state evolution.}
 \label{fig:nula-channel-model}
\end{figure}
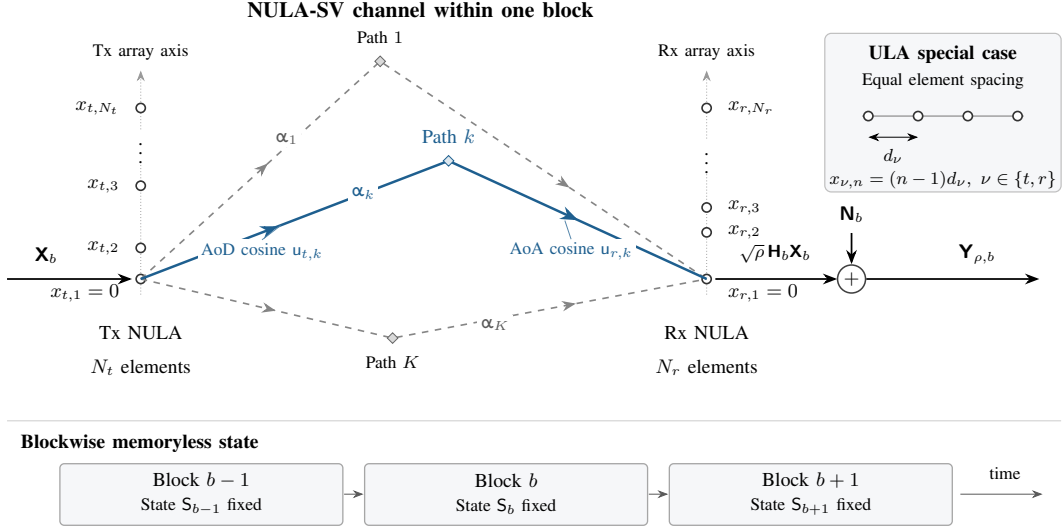

\begin{remark}
If the angles are shared across several gain blocks, the transition kernel no longer factors and the present theorem does not apply. In common sparse MIMO channel models, angular or beam subspaces vary more slowly than the small-scale complex gains \cite{Anjinappa2020,VaChoiHeath2017}.
\end{remark}

\subsection{Main Result and Interpretation}
Our main result is as follows, under \cref{ass:ula}.

\begin{theorem}[\ac{dof} of the \ac{nula}-\ac{sv} channel without \ac{csi}]
\label{thm:capacity}
Consider the channel in \eqref{eq:sv-channel}--\eqref{eq:channel} under \cref{ass:ula} and the average-power constraint \(\E{\|\RM{X}\|_F^2}\le T\). The limit \(\lim_{\rho\to\infty}C_{\rm blind}(\rho)/\log\rho\) exists and equals
\begin{equation}
 d_{\rm blind}^{\rm cap}
 =
 \begin{cases}
 1-\dfrac1T,&K=1,\\[4pt]
 K\left(1-\dfrac{3}{2T}\right),&K\ge2.
 \end{cases}
 \label{eq:capacity}
\end{equation}
\begin{IEEEproof}
See Appendix~\ref{app:proof-thm-capacity}.
\end{IEEEproof}
\end{theorem}

For \(K\ge2\), the \(K\) receive-direction dimensions appear in both the joint and fixed-data images and cancel. This cancellation reflects the $K$ locally distinguishable \ac{aoa} variations in both output maps on their probability-1 regular sets. The remaining loss consists of \(2K\) real gain dimensions and \(K\) \ac{aod} dimensions visible to the precoded input.

It is worth noting that there are single-antenna input distributions attaining \(1-1/T\) for every \(K\). At \(K=T=2\), it already attains the capacity pre-log \(1/2\), as the two-stream Gaussian input does. Therefore, pre-log optimality does not in general require a unique input rank. 

\section{Achievability via Quotient Ranks}
\label{sec:achievability}

\subsection{Proof Architecture}
For $K\ge2$, we choose the fixed Gaussian input $\RM{X}=\M{B}_K\RM{S}$. Write $\RM{Z}=\RM{H}\M{B}_K\RM{S}$ for the noiseless output. Independence of the additive noise gives
\[
 I(\RM{S};\RM{Y}_\rho)
 =I(\RM{Z};\RM{Y}_\rho)
  -I(\RM{Z};\RM{Y}_\rho\mid\RM{S}).
\]
The small-noise argument below turns this identity into
\[
 \lim_{\rho\to\infty}
 \frac{I(\RM{S};\RM{Y}_\rho)}{\log\rho}
 =\frac12\bigl[d_I(\RM{Z})-d_I(\RM{Z}\mid\RM{S})\bigr],
\]
where $d_I(\cdot)$ denotes the information dimension, as will be defined later. We show that the joint real Jacobian rank of the input data and the unknown parameters is $2KT+K$, while the real Jacobian rank conditioned on a fixed input is $4K$. By showing that the regularity conditions including absolute continuity, constant rank on probability-1 open sets, and a finite output second moment are satisfied by the chosen input, we then show that the ranks equal the two information dimensions, respectively.

\subsection{R\'enyi Information Dimension and Its Conditional Version}
Information dimension, introduced by R\'enyi, measures the effective real dimension occupied by a probability distribution through the growth rate of the entropy of a fine quantization \cite{Renyi1959,WuVerdu2010}. For a deterministic vector $\V{x}\in\R^n$ and an integer $m\ge1$, define componentwise quantization by
\[
 \langle\V{x}\rangle_m=\frac{\lfloor m\V{x}\rfloor}{m}.
\]
Since multiplication by a deterministic nonzero constant does not change discrete entropy,
$H(\langle\RV{x}\rangle_m)=H(\lfloor m\RV{x}\rfloor)$.

\begin{definition}[Lower and upper R\'enyi information dimensions]
If $H(\lfloor\RV{x}\rfloor)<\infty$, define
\begin{subequations}
\begin{align}
 \underline d_I(\RV{x})
 &=\liminf_{m\to\infty}
 \frac{H(\quant{\RV{x}}{m})}{\log m},\\
 \overline d_I(\RV{x})
 &=\limsup_{m\to\infty}
 \frac{H(\quant{\RV{x}}{m})}{\log m}.
\end{align}
\end{subequations}
If the two quantities coincide, their common value is denoted by $d_I(\RV{x})$.
\end{definition}

Under the finite-entropy condition above, a discrete distribution has information dimension zero and an absolutely continuous distribution on $\R^r$ has information dimension $r$ \cite{Renyi1959,WuVerdu2010}.

We adopt the conditional information dimension of Geiger and Koch \cite[Definition~1]{GeigerKoch2019}:
\begin{subequations}
\begin{align}
 \underline d_I(\RV{z}\mid\RV{\theta})
 &=\liminf_{m\to\infty}
 \frac{H(\langle \RV{z}\rangle_m\mid\RV{\theta})}{\log m},\\
 \overline d_I(\RV{z}\mid\RV{\theta})
 &=\limsup_{m\to\infty}
 \frac{H(\langle\RV{z}\rangle_m\mid\RV{\theta})}{\log m}.
 \label{eq:conditional-information-dimension}
\end{align}
\end{subequations}
When the two quantities coincide, their common value is denoted by $d_I(\RV{z}\mid\RV{\theta})$. If $H(\langle\RV{z}\rangle_1)<\infty$ and the pointwise conditional information dimension $d_I(P_{\RV{z}\mid\RV{\theta}=\V{\theta}})$ exists for $P_{\RV{\theta}}$-almost every $\V{\theta}$, Lemma~2 of \cite{GeigerKoch2019} gives
\begin{equation}
 d_I(\RV{z}\mid\RV{\theta})
 =\int d_I(P_{\RV{z}\mid\RV{\theta}=\V{\theta}})
 \,P_{\RV{\theta}}(d\V{\theta}).
 \label{eq:conditional-information-dimension-average}
\end{equation}

\subsection{From Constant Rank to Information Dimension}
The following conditions suffice to transfer Jacobian ranks to information dimensions. The data-bearing variable is denoted by $\RV{\theta}$ and the nuisance variable by $\RV{u}$ throughout this subsection.

\begin{definition}[Information-dimension transfer conditions]
\label{ass:regular}
Let $\RV{\theta}$ and $\RV{u}$ be finite-dimensional real random vectors, representing the data and nuisance parameters, respectively, and let $\RV{z}=f(\RV{\theta},\RV{u})$. We impose the following conditions.
\begin{enumerate}[label=(A\arabic*)]
 \item The distribution of $(\RV{\theta},\RV{u})$ is absolutely continuous except for a singular set with probability zero.
 \item There exists an open Borel set $\Set{O}\subseteq\R^{\dim\RV{\theta}+\dim\RV{u}}$ such that $\Prob{(\RV{\theta},\RV{u})\in\Set{O}}=1$. The map \(f\) is \(C^1\) on \(\Set{O}\), and
\[
 \rank Df(\V{\theta},\V{u})=r,
 \qquad
 (\V{\theta},\V{u})\in\Set{O}.
\]
For \(P_{\RV{\theta}}\)-almost every \(\V{\theta}\), define $\Set{O}_{\V{\theta}}
 =
 \{\V{u}:(\V{\theta},\V{u})\in\Set{O}\}$. Then $
 P_{\RV{u}\mid\RV{\theta}=\V{\theta}}
 (\Set{O}_{\V{\theta}})=1$ and
\[
 \rank D_{\V{u}}f(\V{\theta},\V{u})=r_u,
 \qquad
 \V{u}\in\Set{O}_{\V{\theta}}.
\]
 \item $\E{\|\RV{z}\|^2}<\infty$.
\end{enumerate}
\end{definition}

As we shall show later, the transfer conditions (A1)--(A3) are not additional assumptions, since they can be derived from the \ac{nula}-\ac{sv} model (U1)--(U3) and the input distribution specified later. For the moment, let us assume that they are true, from which we may obtain the following lemma.

\begin{lemma}[Information dimension of a rectifiable image]
\label{lem:smooth-image}
Under \cref{ass:regular},
\[
 d_I(\RV{z})=r,\qquad
 d_I(\RV{z}\mid\RV{\theta})=r_u.
\]
\begin{IEEEproof}
See Appendix~\ref{app:proof-lem-smooth-image}.
\end{IEEEproof}
\end{lemma}

\Cref{lem:smooth-image} formally shows that the information dimensions equal the Jacobian ranks under (A1)--(A3). In its spirit, \cref{lem:smooth-image} is similar to the Theorem 14 in \cite{Koliander2016}. However, the lack of identifiability of parameters technically prevents us from reusing that result, thereby \cref{lem:smooth-image} is necessary. We next convert the difference of the Jacobian ranks into the mutual information pre-log.

\subsection{From Information Dimensions to an Achievable Pre-Log}
The seminal work of Wu, Shamai, and Verd\'u showed that the DoF of a scalar interference channel can be characterized through differences of information dimensions: the contribution of each desired signal is measured by the information dimension of the signal-plus-interference observation minus that of the interference alone \cite{WuShamaiVerdu2015}. This structure suggests an analogous decomposition for the problem considered in this paper. Here, the subtracted term is instead the output dimension generated by the unknown channel state when the input is held fixed. The remainder of this subsection formalizes this principle for the NULA-SV model. We first give a slightly specialized version of the classical pre-log result as follows.

\begin{lemma}[High-SNR pre-log]
\label{lem:small-noise}
Let $\RV{x}\in\R^n$ be independent of $\RV{n}\sim\mathcal N(\V{0},\M{I}_n)$, and suppose $H(\lfloor\RV{x}\rfloor)<\infty$. If $d_I(\RV{x})$ exists, then
\begin{equation}
 I(\RV{x};\sqrt\rho\RV{x}+\RV{n})
 =\frac{d_I(\RV{x})}2\log\rho+o(\log\rho).
 \label{eq:small-noise}
\end{equation}
Any fixed positive-definite noise covariance yields the same pre-log. 
\begin{IEEEproof}
Using \cite[Theorem~6]{StotzBolcskei2016} and the existence of $d_I(\RV{x})$, we have
\[
 \limsup_{\rho\to\infty}
 \frac{I(\RV{x};\sqrt\rho\,\RV{x}+\RV{n})}
 {\frac12\log\rho}
 =d_I(\RV{x}).
\]
To obtain the corresponding lower limit, set $m=\lfloor\sqrt\rho\rfloor$. By monotonicity of Gaussian channel mutual information in SNR and the quantization lower bound proved in \cref{lem:quantization-lower},
\[
 \begin{aligned}
 I(\RV{x};\sqrt\rho\,\RV{x}+\RV{n})
 &\ge I(\RV{x};m\RV{x}+\RV{n}) \ge H(\quant{\RV{x}}{m})-C_n,
 \end{aligned}
\]
where $C_n$ only depends on $n$. Since $H(\quant{\RV{x}}{m})/\log m\to d_I(\RV{x})$ and $\log m/(\frac12\log\rho)\to1$, it follows that
\[
 \liminf_{\rho\to\infty}
 \frac{I(\RV{x};\sqrt\rho\,\RV{x}+\RV{n})}
 {\frac12\log\rho}
 \ge d_I(\RV{x}).
\]
Combining the two bounds proves \eqref{eq:small-noise}. For any fixed positive-definite noise covariance, the same pre-log follows by whitening, since invertible linear transformations preserve information dimension and the entropy finiteness of the quantized variable~\cite[Lemma~2]{StotzBolcskei2016}.
\end{IEEEproof}
\end{lemma}

With \cref{lem:small-noise}, we obtain the following result.
\begin{proposition}[High-SNR quotient-rank formula with nuisance parameters]
\label{prop:quotient-rank}
Let $\RV{z}=f(\RV{\theta},\RV{u})$, where $\RV{\theta}$ is the data-bearing variable and $\RV{u}$ is the nuisance variable, and let $\RV{n}\sim\mathcal N(\V{0},\M{\Sigma})$ be independent of $(\RV{\theta},\RV{u})$, where $\M{\Sigma}\succ0$ is deterministic and does not depend on $\rho$. Set $\RV{y}_\rho=\sqrt\rho\RV{z}+\RV{n}$. Under \cref{ass:regular},
\begin{equation}
 I(\RV{\theta};\RV{y}_\rho)
 =\frac{r-r_u}{2}\log\rho+o(\log\rho).
 \label{eq:nuisance}
\end{equation}
\begin{IEEEproof}
See Appendix~\ref{app:proof-cor-quotient-rank}.
\end{IEEEproof}
\end{proposition}

The rank difference $r-r_u$ is the quotient dimension between the joint noiseless image and the image that remains when the data-bearing variable is fixed, which determines the achievable pre-log under the stated regularity conditions. The remainder of the paper computes this quotient for the NULA-SV channel and then proves that no input distribution can exceed it.

\subsection{The Jacobian Rank of the Gaussian Input}
The physical state in \eqref{eq:unordered-ray-state} is unchanged under simultaneous permutations of the indexed path parameters. Outside the path-collision set, this produces only a finite permutation ambiguity in the indexed Euclidean representation. Such a discrete ambiguity does not change local Jacobian ranks. Moreover, revealing or suppressing the corresponding permutation costs at most $\log K!$ nats and hence does not affect an information dimension or the corresponding \ac{dof}. We therefore compute all Jacobian ranks directly in the indexed coordinates specified in \cref{ass:ula}.

\begin{lemma}[Almost-everywhere full rank condition]
\label{lem:confluent}
Let $\lambda_1,\ldots,\lambda_N$ be fixed distinct real numbers, and denote $\V{a}(u)=N^{-1/2}[e^{\jmathu\lambda_1u},\ldots, e^{\jmathu\lambda_Nu}]^T$ and $\M{A}(\V{u})=[\V{a}(u_1),\ldots,\V{a}(u_K)]$. On any fixed nondegenerate compact interval of directional cosines, $\M{A}(\V{u})$ has rank $K$ for almost every $\V{u}$ when $N\ge K$. When $N\ge K+1$, the matrices $[\M{A}(\V{u}),\dot{\V{a}}(u_k)]$, $k=1,\ldots,K$, have rank $K+1$ for almost every $\V{u}$ simultaneously. 
\begin{IEEEproof}
See Appendix~\ref{app:proof-lem-confluent}.
\end{IEEEproof}
\end{lemma}

The rank condition of the matrix $[\M{A}(\V{u}),\dot{\V{a}}(u_k)]$ is the minimal condition on the array manifold that will be used below. Note that distinction of directions does not ensure deterministically the rank condition for an arbitrary \ac{nula}. For example, positions $(0,1,3)$ in wavelengths and directional cosines $(-1/3,0,1/3)$ give, up to normalization,
\[
 \begin{bmatrix}
 1&1&1\\
 \omega^*&1&\omega\\
 1&1&1
 \end{bmatrix},
 \qquad \omega=e^{\jmathu2\pi/3},
\]
which has rank two although every pair of columns is independent. Fortunately, according to \cref{lem:confluent}, such events have probability zero under (U1). This enables us to obtain the following channel state rank result.

\begin{lemma}[Local rank of the SV channel parameterization]
\label{lem:channel-rank}
Under \cref{ass:ula}, exclude zero gains and the null sets where $\rank\M{A}_t<K$ or $\rank[\M{A}_r,\dot{\V{a}}_r(u_{r,k})]<K+1$ for some $k$. Then the real Jacobian rank of $(\V{u}_r,\V{u}_t,\V{\alpha})
 \longmapsto \M{H}(\V{u}_r,\V{u}_t,\V{\alpha})$ is $4K$.
\begin{IEEEproof}
See Appendix~\ref{app:proof-lem-channel-rank}.
\end{IEEEproof}
\end{lemma}

In the next two lemmata, we extend this channel state rank result to noiseless channel outputs. To this end, we choose explicitly the candidate input distribution for showing the achievability. In particular, for $K\geq 2$, define the semi-unitary antenna selection matrix $\M{B}_K=[\V{e}_1,\ldots,\V{e}_K]\in\C^{N_t\times K}$. The input is fixed to the Gaussian matrix
\begin{equation}
  \begin{aligned}
   \RM{X}&=\M{B}_K\RM{S},
   \quad\RM{S}=[\rv{s}_{ij}]\in\C^{K\times T},~~\rv{s}_{ij}\stackrel{\rm iid}{\sim}\CN(0,1/K),
   \quad\RM{S}\perp(\RV{u}_r,\RV{u}_t,\RV{\alpha},\RM{N}).
  \end{aligned}
  \label{eq:gaussian-ach-input}
\end{equation}
Consequently,
 \[
  \E{\|\RM{X}\|_F^2}
  =\E{\|\RM{S}\|_F^2}
  =\sum_{i=1}^{K}\sum_{t=1}^{T}
  \E{|\rv{s}_{it}|^2}=T.
 \]

\begin{lemma}[Local rank of the unconditional output]
\label{lem:joint-output-rank}
For $K\ge2$, denote $\M{X}=\M{B}_K\M{S}$ and $\M{Z}=\M{H}(\V{u}_r,\V{u}_t,\V{\alpha})\M{B}_K\M{S}$, where $\M{S}\in\mathbb{C}^{K\times T}$. At every point with $\rank\M{S}=K$, nonzero gains, $\rank(\M{B}_K^H\M{A}_t)=K$, and $\rank[\M{A}_r,\dot{\V{a}}_r(u_{r,k})]=K+1$ for all $k$, we have
\[
 \rank_\R
 D_{(\V{u}_r,\V{u}_t,\V{\alpha},\M{S})}\M{Z}
 =2KT+K.
\]
\begin{IEEEproof}
See Appendix~\ref{app:proof-lem-joint-output-rank}.
\end{IEEEproof}
\end{lemma}

As detailed in Appendix~\ref{app:proof-lem-joint-output-rank}, the core idea of Lemma~\ref{lem:joint-output-rank} is that when $\M{S}$ is unknown, the triple $(\V{u}_t,\V{\alpha},\M{S})$ is equivalent to a (real) $2KT$-dimensional object under the available observations. Since the triple itself is $(2KT+3K)$-dimensional, we see that there is a $3K$-dimensional fiber that is unidentifiable by the observations, which corresponds locally to the AoDs and the complex gains. In the next lemma, we show that fixing an input $\M{S}$ with full row rank eliminates this local ambiguity on the regular set.

\begin{lemma}[Local rank of the conditional output]
\label{lem:fixed-input-rank}
For $K\ge2$ and every full-row-rank $\M{S}\in\C^{K\times T}$,
\[
 \rank_\R D_{(\V{u}_r,\V{u}_t,\V{\alpha})}
 [\M{H}(\V{u}_r,\V{u}_t,\V{\alpha})\M{B}_K\M{S}]
 =4K
\]
whenever all gains are nonzero, $\rank(\M{B}_K^H\M{A}_t)=K$, and $\rank[\M{A}_r,\dot{\V{a}}_r(u_{r,k})]=K+1$ for every $k$. These conditions hold with probability one under \cref{ass:ula}.
\begin{IEEEproof}
See Appendix~\ref{app:proof-lem-fixed-input-rank}.
\end{IEEEproof}
\end{lemma}

Next, we resolve the remaining regularity conditions.

\begin{corollary}[probability-1 regular-chart cover]
\label{cor:regular-chart-cover}
Let $K\ge2$. Define the regular set of channel parameters
\[
 \Set{O}_{\rm ch}
 =
 \left\{
 \begin{array}{l}
 u_{\nu,k}\in\operatorname{int}\Set{U}_\nu
 \quad(\nu\in\{r,t\},\ k=1,\ldots,K),\\
 u_{r,i}\ne u_{r,j},\
 u_{t,i}\ne u_{t,j}\quad(i\ne j),\\
 \alpha_k\ne0\quad(k=1,\ldots,K),\\
 \rank(\M{B}_K^H\M{A}_t)=K,\\
 \rank[\M{A}_r,\dot{\V{a}}_r(u_{r,k})]=K+1,~(k=1,\ldots,K)
 \end{array}
 \right\}.
\]
Also define the joint regular set $\Set{O}_{\rm joint}
 =\Set{O}_{\rm ch}
 \times
 \{\M{S}\in\C^{K\times T}:\rank\M{S}=K\}$. Under (U1)--(U3) of \cref{ass:ula}, we have $\Prob{(\RV{u}_r,\RV{u}_t,\RV{\alpha})\in\Set{O}_{\rm ch}}=1$. If $\RM{S}$ is independent of the channel parameters and has an absolutely continuous distribution on $\C^{K\times T}$, then $\Prob{(\RV{u}_r,\RV{u}_t,\RV{\alpha},\RM{S})
 \in\Set{O}_{\rm joint}}=1$. On $\Set{O}_{\rm joint}$, the joint map $(\V{u}_r,\V{u}_t,\V{\alpha},\M{S})
\longmapsto \M{H}(\V{u}_r,\V{u}_t,\V{\alpha})\M{B}_K\M{S}$ is $C^1$ and has constant real Jacobian rank $2KT+K$. For every fixed $\M{S}$ with full row rank, the conditional map $(\V{u}_r,\V{u}_t,\V{\alpha})
\longmapsto \M{H}(\V{u}_r,\V{u}_t,\V{\alpha})\M{B}_K\M{S}$ on $\Set{O}_{\rm ch}$ is $C^1$ and has constant real Jacobian rank $4K$. Consequently, the joint and conditional maps satisfy condition (A2) of \cref{ass:regular}.

\begin{IEEEproof}
See Appendix~\ref{app:proof-cor-regular-chart-cover}.
\end{IEEEproof}
\end{corollary}

\begin{lemma}[Regularity of the Gaussian input]
\label{lem:small-ball-ui}
Let $K\ge2$ and let $(\RM{X},\RM{S})$ be given by \eqref{eq:gaussian-ach-input}. Under \cref{ass:ula}, the map $(\V{u}_r,\V{u}_t,\V{\alpha},\M{S})
 \longmapsto \M{H}(\V{u}_r,\V{u}_t,\V{\alpha})\M{B}_K\M{S}$ satisfies (A1)--(A3) of \cref{ass:regular}, with the data variable $\RV{\theta}=\realify(\RM{S})$ and nuisance variable $\RV{u}=(\RV{u}_r,\RV{u}_t,\realify(\RV{\alpha}))$.
\begin{IEEEproof}
See Appendix~\ref{app:proof-lem-small-ball-ui}.
\end{IEEEproof}
\end{lemma}

Now, by collecting all the aforementioned results, we arrive at the following $K\geq 2$ achievability result.
\begin{proposition}[Achievable pre-log of the Gaussian input]
\label{prop:full-ach}
For $K\ge2$, under \cref{ass:ula}, the Gaussian input in \eqref{eq:gaussian-ach-input} achieves
\begin{equation}
 d_{\rm Gauss}
 =K\left(1-\frac{3}{2T}\right).
 \label{eq:full-ach}
\end{equation}
\begin{IEEEproof}
See Appendix~\ref{app:proof-prop-full-ach}.
\end{IEEEproof}
\end{proposition}

\cref{prop:full-ach} is not yet a capacity statement, because an optimizing input may be supported (essentially) on a lower-rank subset, or have a varying support as the SNR grows. 

Next, we consider the $K=1$ case by investigating the single-antenna transmit strategy. We would like to highlight that this strategy can in fact be extended to every $K\ge1$. In particular, it recovers the classical single-input block-fading pre-log $1-1/T$ \cite{ZhengTse2002} and is sufficient for the $K=1$ lower bound.

\begin{proposition}[Single-stream achievability]
\label{prop:single-beam}
Under \cref{ass:ula}, for every $K\ge1$, the fixed input $\RM{X}=\V{e}_1\RV{s}^{T}$ with $\RV{s}\sim\CN(\V{0},\M{I}_T)$ independent of the channel and noise achieves a per-channel-use pre-log of at least $1-1/T$. Consequently,
\[
 d_{\rm blind}^{\rm cap}\ge1-\frac1T.
\]
\begin{IEEEproof}
See Appendix~\ref{app:proof-prop-single-beam}.
\end{IEEEproof}
\end{proposition}

Intuitively, selecting the first transmit and receive antennas removes both AoA and AoD from the scalar observation (since an antenna at a fixed position carry a constant phase), leaving only one complex fading coefficient per block.

\section{Converse via Genie-aided Bounds}\label{sec:converse}
The preceding section establishes achievability for fixed input distributions using information dimension. However, a capacity converse must also cover input distributions that vary with SNR. In this section, we obtain the required uniform bounds by introducing genie-aided observations, treating $K=1$ separately from $K\ge2$.

Define the random pathwise channel matrix
\begin{equation}\label{random_pathwise}
 \RM{G}_{\rm path}:=\M{D}(\RV{\alpha})\M{A}_t(\RV{u}_t)^H.
\end{equation}
The pathwise genie provides the following physically unobservable observation:
\begin{equation}
 \begin{aligned}
  \RM{Z}_\rho
  &=
   \sqrt\rho\,\RM{G}_{\rm path}\RM{X}+\RM{N}_0,~~\RM{N}_0
  \sim\CN(\M{0},K^{-1}\M{I}_K\otimes \M{I}_T),
 \end{aligned}
 \label{eq:pathwise-genie}
\end{equation}
where $\RM{N}_0$ is independent of the input and channel parameters. Next we show that the pathwise genie indeed constitutes an upper bound for the achievable mutual information.

\begin{proposition}[Pathwise-genie reduction]
\label{prop:pathwise-genie}
For the original observation $\RM{Y}_\rho$ in \eqref{eq:channel} and the pathwise observation in \eqref{eq:pathwise-genie}, there exists a stochastic kernel $Q_\rho$ such that $P_{\RM{Y}_\rho\mid\RM{X}} = P_{\RM{Z}_\rho\mid\RM{X}}Q_\rho$. Equivalently, one can construct an observation $\widetilde{\RM{Y}}$ with $P_{\widetilde{\RM{Y}}\mid\RM{X}}=P_{\RM{Y}_\rho\mid\RM{X}}$ such that $\RM{X}\longrightarrow\RM{Z}_\rho\longrightarrow\widetilde{\RM{Y}}$ is a Markov chain. Consequently, every admissible input distribution satisfies~\cite[Ch.~3]{PolyanskiyWu2025}
\begin{equation}
 I(\RM{X};\RM{Y}_\rho)
 \le
 I(\RM{X};\RM{Z}_\rho).
 \label{eq:pathwise-data-processing}
\end{equation}
\begin{IEEEproof}
See Appendix~\ref{app:proof-prop-pathwise-genie}.
\end{IEEEproof}
\end{proposition}

\subsection{The \texorpdfstring{$K=1$}{} Case}
In the $K=1$ case, we use an additional geometric genie that reveals the complete geometry
$\RS{V}=(\rv{u}_r,\rv{u}_t)$ to the receiver.
\begin{proposition}[Single-path geometric genie converse]
\label{prop:geom-genie}
For $K=1$, under \cref{ass:ula},
\begin{equation}
 d_{\rm blind}^{\rm cap}
 \le 1-\frac1T.
 \label{eq:geom-genie}
\end{equation}
\begin{IEEEproof}
See Appendix~\ref{app:proof-thm-geom-genie}.
\end{IEEEproof}
\end{proposition}

Given the revealed geometry, projection onto the receive steering vector is a sufficient statistic and reduces the channel to the classical scalar noncoherent block fading model. Together with \cref{prop:single-beam}, this gives the exact $K=1$ \ac{dof}.

\begin{corollary}[Exact single-path DoF]
\label{cor:single-path-capacity}
\[
 d_{\rm blind}^{\rm cap}=1-\frac1T,\qquad K=1.
\]
\end{corollary}

\begin{IEEEproof}
This is a straightforward corollary of \cref{prop:single-beam} and \cref{prop:geom-genie}.
\end{IEEEproof}

\subsection{The \texorpdfstring{$K\geq 2$}{} Converse: Target and Proof Architecture}
For the remainder of the converse, assume $K\ge2$. The geometric genie used for the $K=1$ case turns out to be too strong to recover the desired multiple-path pre-log, thus we work directly with the pathwise observation $\RM{Z}_\rho$. Define the peak- and average-power pathwise mutual information envelopes
\begin{subequations}
\begin{align}
 J_{\rm pk}^{Z}(\rho;A)
 &:=
 \sup_{P_{\RM{X}}:\Prob{\|\RM{X}\|_F\le A}=1}
 I(\RM{X};\RM{Z}_\rho),\\
 J_{\rm av}^{Z}(\rho)
 &:=
 \sup_{P_{\RM{X}}:\E{\|\RM{X}\|_F^2}\le T}
 I(\RM{X};\RM{Z}_\rho).
 \label{eq:genie-envelopes}
\end{align}
\end{subequations}
Let $D_\star :=KT-\frac{3K}{2}$. By \eqref{eq:pathwise-data-processing}, the DoF converse reduces to proving
\[
 \limsup_{\rho\to\infty}
 \frac{J_{\rm av}^{Z}(\rho)}{\log\rho}
 \le D_\star.
\]

The central obstacle is that the optimal input distribution may depend on $\rho$. To tackle this issue, we first focus on a fixed-amplitude shell, where $\sigma_1(\RM{X})\asymp1$, and fix a single scale for the second largest singular value. Suppose that, for some $a\ge0$,
\begin{equation}
 \sigma_2(\RM{X})\asymp\rho^{-a/2},
 \qquad b=(1-a)^+\in[0,1]
 \label{eq:single-sigma2-scale}
\end{equation}
hold uniformly over the support. Then $b$ is the fraction of the second spatial mode that remains above the noise floor, because $\log(1+\rho\sigma_2(\RM{X})^2)
 =b\log\rho+O(1)$.

The proof is organized around
\[
 I(\RM{X};\RM{Z}_\rho)
 =h(\RM{Z}_\rho)-h(\RM{Z}_\rho\mid\RM{X}).
\]
For the same input distribution satisfying \eqref{eq:single-sigma2-scale}, we establish the two complementary estimates
\begin{subequations}
\begin{align}
 h(\RM{Z}_\rho)
 &\!\le\! [(1\!-\!b)(K\!+\!T\!-\!1)\!+\!bKT]\log\rho\!+\!o(\log\rho),
 \label{eq:roadmap-joint-entropy}\\
 h(\RM{Z}_\rho\!\mid\!\RM{X})
 &\!\ge\!\left(K+\frac{bK}{2}\right)\log\rho-o(\log\rho).
 \label{eq:roadmap-conditional-entropy}
\end{align}
\end{subequations}
The conditional output lower bound strengthens with $b$ because more \ac{aod} variation is resolvable, whereas the rank-1 tube output upper bound tightens as $b$ decreases. Subtracting them gives
\begin{equation}
 D(a)=(1-b)(T-1)+bD_\star,
 \qquad b=(1-a)^+.
 \label{eq:roadmap-affine-envelope}
\end{equation}
Thus at every single $\sigma_2$-scale, the \ac{dof} is bounded by $D_\star$, since $D_\star\ge T-1$ for $T\ge K\ge2$. An arbitrary input distribution with bounded peak power is handled later by jointly binning $\log(\rho\sigma_1^2)$ and $\log(1+\rho\sigma_2^2)$, and is finally lifted to the original case with the average power constraint. We first show why the same second largest singular value controls both entropy estimates on a fixed-amplitude shell.

\subsection{The Second Singular Value as a Common Scale for Rank-1 Proximity and AoD Resolvability}
Let $\Set{R}_1^{\rm in}=\{\V{b}\V{s}^T:\V{b}\in\C^{N_t},\ \V{s}\in\C^T\}$ be the rank-1 input set. The next lemma shows that $\sigma_2(\M{X})$ is related to the absolute distance to this set. It thereby converts an upper bound of the second-mode effective SNR into the radius of the rank-1 tube, which is useful later.
\begin{lemma}[Distance to the rank-1 set via the second singular value]
\label{lem:sigma2-rankone-distance}
For any $\M{X}\in\C^{N_t\times T}$, let $r_{\M{X}}=\min\{N_t,T\}$. Then
\begin{equation}
 \sigma_2(\M{X})
 \le
 \dist_F(\M{X},\Set{R}_1^{\rm in})
 \le
 \sqrt{r_{\M{X}}-1}\,\sigma_2(\M{X}).
 \label{eq:sigma2-rankone-distance}
\end{equation}
\begin{IEEEproof}
The Eckart--Young identity gives~\cite[Sec.~7.4.2]{HornJohnson2013}
\[
 \dist_F(\M{X},\Set{R}_1^{\rm in})
 =
 \left(\sum_{j=2}^{r_{\M{X}}}\sigma_j(\M{X})^2\right)^{1/2}.
\]
Using also
\[
 \sigma_j(\M{X})\le\sigma_2(\M{X}),
 \qquad j\ge2,
\]
yields \eqref{eq:sigma2-rankone-distance}.
\end{IEEEproof}
\end{lemma}

In particular, the lemma also shows that tracking $\sigma_2$ controls every higher singular value and hence every approach to the rank-1 set. This is convenient since we do not need to use other singular values to construct more sophisticated bounds for the rank-1 tube.

Next, we quantify how $\sigma_2(\M{X})$ lower bounds the the \ac{aod} sensitivity given the observations at the Rx. For a deterministic input matrix $\M{X}$, let us define the \ac{aod}-projected temporal signal as
\begin{equation}
\V{v}_{\M{X}}(u) = \V{a}_t(u)^H\M{X}.
\end{equation}
For the fixed NULA, $\V{v}_{\M{X}}(u)$ is a vector-valued exponential polynomial with frequencies in the fixed set $\Set{\Lambda}_t=\{-2\pi x_{t,n}:n=1,\ldots,N_t\}$. Define the vector of pairwise Wronskians as
\begin{equation}
 \V{w}_{\M{X}}(u)
 =\V{v}_{\M{X}}(u)\wedge\dot{\V{v}}_{\M{X}}(u),
\end{equation}
whose frequencies belong to $\Set{\Lambda}_t+\Set{\Lambda}_t$. Let $\|\V{w}_{\M{X}}\|_{\rm coef}$ denote the Euclidean norm of this finite coefficient vector. In particular, for a scalar-valued exponential polynomial \(f(u)=\sum_{\ell=1}^{L}c_\ell e^{\jmathu\lambda_\ell u}\) with distinct real frequencies and $c_\ell\in\mathbb C$, define
\begin{equation}
\|f\|_{\rm coef}=\Big[\sum_{\ell=1}^L |c_\ell|^2\Big]^{\frac{1}{2}}.
\end{equation}
For a vector-valued exponential polynomial $\V{v}(u)=\sum_{\ell=1}^{L}\V{c}_\ell e^{\jmathu\lambda_\ell u}=[v_1,\dotsc,v_T]^T$ where $\V{c}_\ell\in\mathbb C^T$, define
\begin{equation}
\|\V{v}\|_{\rm coef}=\Big[\sum_{t=1}^T \|v_t\|_{\rm coef}^2\Big]^{\frac{1}{2}}
\end{equation}
With the aid of this construction, we may use the standard Wronskian arguments to show the following result
\cite{EremenkoGabrielov2002,KravvarritisMitrouli2009}.

\begin{proposition}[Uniform AoD sensitivity from the second exterior power]
\label{thm:wronskian}
There exists $c>0$, depending only on the fixed transmit positions and $(N_t,T)$, such that
\begin{equation}
 \|\V{w}_{\M{X}}\|_{\rm coef}
 \ge c\,\sigma_1(\M{X})\sigma_2(\M{X}).
 \label{eq:w-sigma}
\end{equation}
\begin{IEEEproof}
See Appendix~\ref{app:proof-thm-wronskian}.
\end{IEEEproof}
\end{proposition}

Together, \cref{lem:sigma2-rankone-distance,thm:wronskian} give the
same quantity $\sigma_2(\M{X})$ two complementary meanings on a fixed-amplitude shell, in the sense that it controls proximity to the rank-1 input set and lower bounds the \ac{aod} sensitivity. These two meanings facilitate the output entropy upper bound and the conditional output entropy lower bound, respectively.

\subsection{Conditional-Output Entropy from AoD Resolvability}
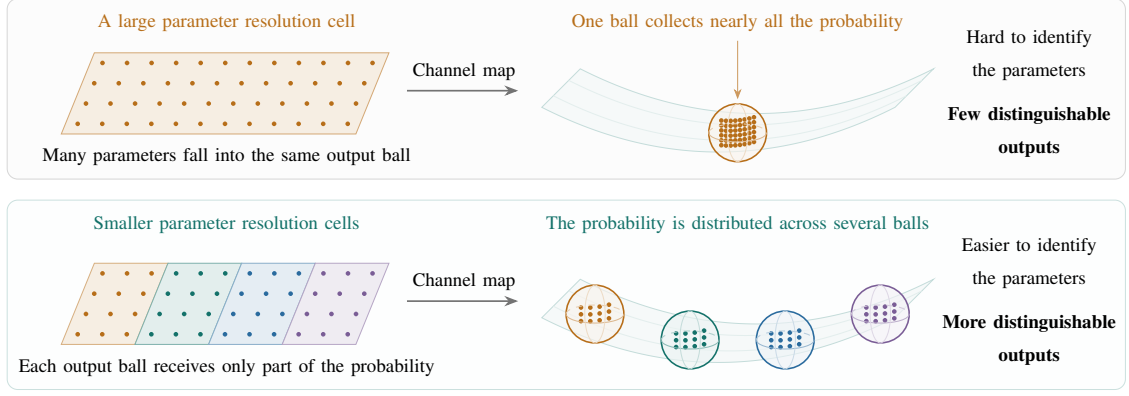
\begin{figure}[t]
\centering
\resizebox{0.95\textwidth}{!}{\input{preimage_entropy_tikz.tex}}
\caption{Graphical illustration of \cref{lem:preimage-volume}: Conditioned on a fixed input, smaller parameter resolution cells yield more distinguishable outputs, and hence lead to larger output entropy.}
\label{fig:preimage-entropy}
\end{figure}
We now convert the projective sensitivity in \cref{thm:wronskian} into a uniform output entropy lower bound for a fixed input. Thanks to the pathwise genie, we may now handle each path separately. Let us first consider the noiseless case. In the following lemma, we control the volume of the parameter preimage, under a fixed input and a $\varepsilon$-ball around the corresponding single-path output.
\begin{lemma}[Uniform projective preimage-volume bound]
\label{lem:preimage-volume}
Fix a finite set \(\Set{\Lambda}=\{\lambda_1,\ldots,\lambda_{|\Set{\Lambda}|}\}\) of distinct real frequency parameters for the exponential functions \(e^{\mathrm j\lambda_\ell u}\), and a nondegenerate compact interval $\Set{I}$. All constants in this lemma depend only on $(\Set{\Lambda},\Set{I},T)$. Let $\V{v}:\Set{I}\to\C^T$ have exponential polynomial components with frequencies in $\Set{\Lambda}$ and satisfy
\[
 \|\V{v}\|_{\rm coef}=1,\qquad
 \|\V{v}\wedge\dot{\V{v}}\|_{\rm coef}\ge\delta>0.
\]
There exists $\kappa=\kappa(\Set{\Lambda},\Set{I},T)>0$ such that, for every $0<\tau<1$, there is a measurable set (``good set'') $\Set{G}_\tau\subseteq\Set{I}$ satisfying $|\Set{I}\setminus\Set{G}_\tau|=O(\tau^\kappa)$, such that for every $u\in\Set{G}_\tau$, $\alpha\ne0$, and $\varepsilon>0$, it holds that
\begin{equation}
 {\rm Vol}(\Set{P}_\varepsilon(\alpha,u))
 =O\!\left(
 \tau^{-4}\varepsilon^2
 \min\left\{1,\frac{\varepsilon}{|\alpha|\delta}\right\}
 \right),
 \label{eq:preimage-volume}
\end{equation}
where the parameter resolution cell is defined by
\[
 \Set{P}_\varepsilon(\alpha,u)
 :=\left\{(\alpha',u')\in\C\times\Set{G}_\tau:
 \|\alpha'\V{v}(u')-\alpha\V{v}(u)\|_2\le\varepsilon
 \right\}.
\]
\begin{IEEEproof}
See \cref{app:proof-lem-preimage-volume}.
\end{IEEEproof}
\end{lemma}

For a fixed output resolution $\varepsilon$, \eqref{eq:preimage-volume} bounds the volume of parameter pairs mapped into an $\epsilon$-ball around a given noiseless output. The factor $\varepsilon^2$ comes from the complex-gain cross section, while $\min\{1,\varepsilon/(|\alpha|\delta)\}$ describes the scaling of the total admissible \ac{aod} length. To elaborate further, \cref{lem:preimage-volume} gives rise to an upper bound on the probability within each $\varepsilon$-ball in the $\RV{z}$-coordinate, via the volume calculation of its preimage in the $(\rv{\alpha},\rv{u})$ coordinate, as portrayed in Fig.~\ref{fig:preimage-entropy}. This naturally gives a lower bound on the conditional entropy of the single-path output $\RV{z}$. After accounting for the exceptional event and applying the Gaussian quantization bound in \cref{lem:quantization-lower}, we obtain the following noisy single-path output entropy lower bound.

\begin{proposition}[Uniform single-path fixed-input output entropy lower bound]
\label{thm:one-path-entropy}
Suppose $\rv{u}$ is supported on the fixed interval $\Set{U}_t$, with a density bounded both from above and away from zero, and let $\rv{\alpha}\sim\CN(0,\sigma_\alpha^2)$ be independent of $\rv{u}$. Let $\RV{n}\sim\CN(\V{0},\M{\Sigma}_n)$ be independent of $(\rv{u},\rv{\alpha})$, where $\M{\Sigma}_n\succ0$ is fixed. For every deterministic matrix sequence in the amplitude shell $1\le\sigma_1(\M{X}_\rho)\le C_1$ as $\rho\rightarrow \infty$,
\begin{equation}
 h\!\left(
 \sqrt\rho\,\rv{\alpha}\V{a}_t(\rv{u})^H
 \M{X}_\rho+\RV{n}
 \right)
 \ge
 \log\rho
 +\frac12\log\!\left(1+\rho\sigma_2(\M{X}_\rho)^2\right)
 -o(\log\rho).
 \label{eq:one-path-entropy}
\end{equation}
For all sufficiently large $\rho$, the $o(\log\rho)$ remainder in \eqref{eq:one-path-entropy} can be replaced by the deterministic function
\[
 \mathcal R_{\rm sp}(\rho)
 =B_1\sqrt{\log\rho}
 +B_2(\log\rho)e^{-b\sqrt{\log\rho}}+B_3,
\]
where $B_1,B_2,B_3,b>0$ depend only on the fixed array positions, dimensions, shell bound, gain variance, noise covariance, and the distribution of the \acp{aod}. In particular, $\mathcal R_{\rm sp}(\rho)=o(\log\rho)$ uniformly over the entire shell, including $\sigma_2(\M{X}_\rho)=0$.
\begin{IEEEproof}
See Appendix~\ref{app:proof-thm-one-path-entropy}.
\end{IEEEproof}
\end{proposition}

Now, we lift the single-path result given by \cref{thm:one-path-entropy} to a shell-wise conditional output entropy bound that allows an arbitrary distribution of $\sigma_2(\RM{X})$ within a fixed-amplitude shell. Let $\RV{z}_{\rho,k}$ denote the $k$th row of $\RM{Z}_\rho$. Conditioned on $\RM{X}=\M{X}$, we see that
\begin{equation}
\RV{z}_{\rho,k}=\sqrt{\rho}\,\rv{\alpha}_k \V{a}_t(\rv{u}_{t,k})^H \M{X}+\RV{n}_k,
\end{equation}
which takes the same form as the quantity in \eqref{eq:one-path-entropy}. In addition, under the condition $\RM{X}=\M{X}$, the path gains, \acp{aod}, and noise rows are independent across $k$, and hence
\[
 h(\RM{Z}_\rho\mid\RM{X}=\M{X})
 =\sum_{k=1}^K
 h(\RV{z}_{\rho,k}\mid\RM{X}=\M{X}).
\]
For each fixed realization $\M{X}$ in the fixed-amplitude shell, we may normalize with respect to $\sigma_1(\M{X})$ yielding $\sigma_1(\M{X}/\sigma_1(\M{X}))=1$ and $\rho \sigma_1(\M{X})^2\sigma_2(\M{X}/\sigma_1(\M{X}))^2 =\rho\sigma_2(\M{X})^2$. Moreover, $\sigma_1(\M{X})=\Theta(1)$ uniformly on the shell, so $\rho \sigma_1(\M{X})^2=\Theta(\rho)$ and $\log(\rho \sigma_1(\M{X})^2)=\log\rho+O(1)$. Applying \cref{thm:one-path-entropy} with $\M{\Sigma}_n=K^{-1}\M{I}_T$ and SNR $\rho \sigma_1(\M{X})^2$ to each conditional row therefore gives
\[
 h(\RV{z}_{\rho,k}\mid\RM{X}=\M{X})
 \ge
 \log\rho
 +\frac12\log\!\left(1+\rho\sigma_2(\M{X})^2\right)
 -o(\log\rho),
\]
uniformly over $\M{X}$ in the shell. Summing over $k$ and averaging over $\RM{X}$ yields 
\begin{equation}
 h(\RM{Z}_\rho\mid\RM{X})
 \ge K\log\rho
 +\frac K2\E{\log(1+\rho\sigma_2(\RM{X})^2)}
 -o(\log\rho).
 \label{eq:pathwise-conditional-entropy}
\end{equation}
Note that this expectation bound is valid without the assumption of a single $\sigma_2(\RM{X})$-scale. Now, if \eqref{eq:single-sigma2-scale} holds uniformly on the support, then
\[
 \E{\log(1+\rho\sigma_2(\RM{X})^2)}
 =b\log\rho+O(1),
 \qquad b=(1-a)^+.
\]
Substitution into \eqref{eq:pathwise-conditional-entropy} yields
\begin{equation}\label{eq:single_scale_conditional}
 h(\RM{Z}_\rho\mid\RM{X})
 \ge\left(K+\frac{bK}{2}\right)\log\rho-o(\log\rho),
\end{equation}
which is exactly \eqref{eq:roadmap-conditional-entropy}. The result \eqref{eq:single_scale_conditional} indeed requires a single $\sigma_2$ scale, and will be further lifted to the generic multiscale case using the expectation bound \eqref{eq:pathwise-conditional-entropy}.

\subsection{Output Entropy from Rank-1 Tube Concentration}
\begin{figure}[t]
\centering
\resizebox{0.75\columnwidth}{!}{\input{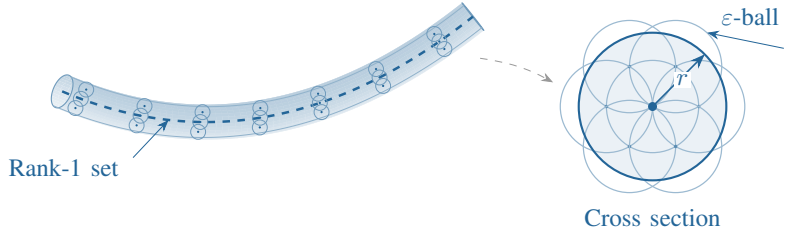}}
\caption{Graphical illustration of a rank-1 tube covering using radius-$\varepsilon$ ambient balls. The cross section radius $r$ is controlled by the second largest singular value $\sigma_2(\RM{X})$.}
\label{fig:rank1-tube-cover}
\end{figure}

We next bound the unconditional output entropy $h(\RM{Z}_\rho)$. When $\sigma_2(\RM{X})$ is small, the noiseless output lies near the rank-1 output set. This concentration limits its covering number and hence its noisy entropy, as portrayed in Fig.~\ref{fig:rank1-tube-cover}. Define the truncated input rank-1 tube
\begin{equation}
 \begin{aligned}
 \Set{T}_1^{\rm in}(C_1,\delta)
 :=\{\M{X}\in\C^{N_t\times T}:\;&\|\M{X}\|_F\le C_1,~\dist_F(\M{X},\Set{R}_1^{\rm in})\le\delta\}.
 \end{aligned}
 \label{eq:input-rankone-tube}
\end{equation}
We may represent this tube in terms of $\sigma_2(\M{X})$ by using \cref{lem:sigma2-rankone-distance}
\begin{equation}
 \begin{aligned}
 \Set{T}_1^{\rm in}(C_1,\delta)
 &\subseteq
 \{\M{X}:\|\M{X}\|_F\le C_1,\ \sigma_2(\M{X})\le\delta\}\\
 &\subseteq
 \Set{T}_1^{\rm in}(C_1,\max\{1,\sqrt{\min\{N_t,T\}-1}\}\delta).
 \end{aligned}
 \label{eq:spectral-geometric-tube}
\end{equation}
Thus a constraint on $\sigma_2$ and a rank-1 tube in the Frobenius norm differ only by a change of radius depending on the dimensionality. Under the pathwise genie channel map, this input tube is turned into a corresponding tube around the rank-1 output subset. 

To bound the output entropy, we use the standard technique of covering numbers. In particular, the KL-covering bound \cite[Theorem~32.4]{PolyanskiyWu2025}, has the following Gaussian specialization. Let $\Set{A}\subset\R^m$ be bounded, possibly depending on $\rho$, and let $N_\epsilon(\Set{A})$ be its Euclidean covering number.  If $\RV{s}$ is supported on $\Set{A}$ and is independent of $\RV{n}\sim\mathcal N(\V{0},\M{\Sigma}_n)$, where $\M{\Sigma}_n\succ0$ is fixed, then
\begin{equation}\label{cover_entropy}
 h(\sqrt\rho\RV{s}+\RV{n})
 \le
 \log N_{\rho^{-1/2}}(\Set{A})+C_{m,\M{\Sigma}_n}.
\end{equation}
Indeed, if $\V{c}_j$ is a covering center for $\V{s}$ and $P_{\V{s}}=\mathcal N(\sqrt\rho\V{s},\M{\Sigma}_n)$, then
\[
 D(P_{\V{s}}\Vert P_{\V{c}_j})
 \le \frac{1}{2\lambda_{\min}(\M{\Sigma}_n)}
 \quad\text{whenever}\quad
 \|\V{s}-\V{c}_j\|_2\le\rho^{-1/2}.
\]
Thus \eqref{cover_entropy} follows from the cited bound and $h(\sqrt\rho\RV{s}+\RV{n})=h(\RV{n})+ I(\RV{s};\sqrt\rho\RV{s}+\RV{n})$. The result can be extended to complex vectors and matrices by vectorization and realification.

Naturally, our next task is to bound the covering number of the rank-1 output tube. Define the rank-1 output set as
\[
 \Set{R}_1^{\rm out}(R)
 =
 \{\M{Q}\in\C^{K\times T}:\rank\M{Q}\le1,\ \|\M{Q}\|_F\le R\},
\]
and define its ambient tube of radius $r$ by
\[
 \Set{T}_1(R,r)
 =
 \{\M{S}:\dist_F(\M{S},\Set{R}_1^{\rm out}(R))\le r\}.
\]
The following lemma bounds the covering number of the output rank-1 tube $\Set{T}_1(R,r)$.

\begin{lemma}[Rank-1 tube covering]
\label{lem:standard-rankone-cover}
There exist constants $C,c>0$ depending only on $(K,T)$ such that, for $0<\epsilon\le R$ and $0\le r\le R$,
\begin{align}
 N_\epsilon(\Set{T}_1(R,r))
 &\le
 C\left(\frac{cR}{\epsilon}\right)^{2(K+T-1)}
 \left(1+\frac{cr}{\epsilon}\right)^{2(KT-K-T+1)}.
 \label{eq:standard-rankone-cover}
\end{align}
\begin{IEEEproof}
See \cref{app:proof-lem-standard-rankone-cover}.
\end{IEEEproof}
\end{lemma}

Now, by handling the probability tails outside the rank-1 output tube, we obtain the following result.
\begin{lemma}[Output-entropy upper bound for a rank-1 tube]
\label{lem:tube-entropy}
Let $\RM{Z}_\rho$ be the pathwise observation in \eqref{eq:pathwise-genie}. Suppose
\[
 \Prob{\RM{X}\in\Set{T}_1^{\rm in}(C_1,\delta)}=1,
 \qquad
 0\le\delta\le\delta_{\max}<\infty,
\]
where $C_1$ and $\delta_{\max}$ are fixed independently of $\rho$, while the deterministic radius $\delta$ may vary with $\rho$. Then, uniformly over all such distributions and radius sequences, we have
\begin{align}
 h(\RM{Z}_\rho)
 &\le
 (K+T-1)\log\rho+(KT\!-\!K\!-\!T\!+\!1)\log(1\!+\!\rho\delta^2)
 +o(\log\rho).
 \label{eq:tube-entropy}
\end{align}
More precisely, for $\rho\ge e^2$ the $o(\log\rho)$ remainder can be bounded by $\mathcal R_{\rm tube}(\rho)=C(1+\log\log\rho)$, with $C$ independent of the input distribution and of $\delta\in[0,\delta_{\max}]$.
\begin{IEEEproof}
See Appendix~\ref{app:proof-lem-tube-entropy}.
\end{IEEEproof}
\end{lemma}

Under the fixed-amplitude shell and the single-scale condition in \eqref{eq:single-sigma2-scale}, \cref{lem:sigma2-rankone-distance} places the input support in a rank-1 input tube of radius $O(\rho^{-a/2})$. Hence \cref{lem:tube-entropy} gives \eqref{eq:roadmap-joint-entropy}, whereas \eqref{eq:single_scale_conditional} gives \eqref{eq:roadmap-conditional-entropy}. Subtracting the latter from the former therefore yields the single-scale DoF upper bound
\[
 \limsup_{\rho\to\infty}
 \frac{I(\RM{X}_\rho;\RM{Z}_\rho)}{\log\rho}
 \le D(a)\le D_\star,
\]
where $D(a)$ is defined in \eqref{eq:roadmap-affine-envelope}. In the next subsection, we derive a multiscale version by applying the two component bounds within every cell of the two-dimensional effective-SNR grid.

\subsection{Full Multiscale Converses}
\begin{figure}[t]
\centering
\resizebox{0.85\textwidth}{!}{\input{theorem2_grid_tikz.tex}}
\caption{Graphical illustration of the two-dimensional effective-\ac{snr} grid in \cref{thm:peak}.}
\label{fig:theorem2-grid}
\end{figure}
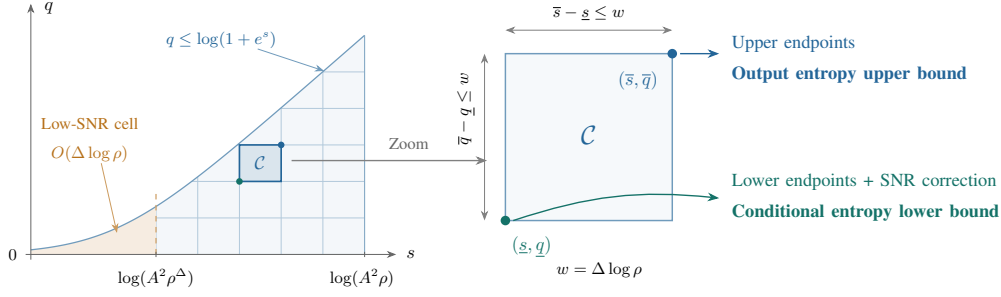

In this subsection, we collect all preceding results to yield the final multiscale converses. We start under a peak power constraint and later lift the result to average power. Fix $\Delta>0$ and put $w=\Delta\log\rho$. Inputs with $\rho\sigma_1^2\le A^2\rho^\Delta$ form a single low-SNR cell whose mutual information is only $O(\Delta\log\rho)$. On its complement, we use the logarithmic coordinates
\[
 s=\log(\rho\sigma_1^2),
 \qquad
 q=\log(1+\rho\sigma_2^2),
\]
and partition their feasible region into rectangular cells of side lengths at most $w$, as portrayed in Fig.~\ref{fig:theorem2-grid}. This two-dimensional grid helps us obtain the following result.
\begin{theorem}[Multiscale converse under peak power constraint]
\label{thm:peak}
Let $K\ge2$ and $T\ge K$. For every fixed $A\ge1$ independent of $\rho$,
\begin{equation}
 \limsup_{\rho\to\infty}
 \frac{J_{\rm pk}^{Z}(\rho;A)}{\log\rho}
 \le D_\star.
 \label{eq:peak-envelope-converse}
\end{equation}
\begin{IEEEproof}
See Appendix~\ref{app:proof-thm-peak}.
\end{IEEEproof}
\end{theorem}

To generalize the result to the average power constraint, we split the input at the growing threshold $\|\RM{X}\|_F^2=\rho^\epsilon$. On the typical event, rescaling would put the conditional input in the unit peak-power ball and changes the effective SNR to $\rho^{1+\epsilon}$. Markov's inequality can then be used to control the weighted high-power tail, as shown in the following lemma.

\begin{lemma}[Converse under average power constraint]
\label{lem:power-tail}
Under the conditions of \cref{thm:peak}, the bound \eqref{eq:peak-envelope-converse} for $A=1$ implies
\begin{equation}
 \limsup_{\rho\to\infty}
 \frac{J_{\rm av}^{Z}(\rho)}{\log\rho}
 \le D_\star.
 \label{eq:average-envelope-converse}
\end{equation}
\begin{IEEEproof}
See Appendix~\ref{app:proof-lem-power-tail}.
\end{IEEEproof}
\end{lemma}

The main theorem thus follows by combining the achievability and converse arguments.

\section{Discussions}
\label{sec:projected-geometry}
The aforementioned capacity result is specific to the blockwise memoryless NULA-SV channel. Nevertheless, its achievability argument can be applied to more general array structures as long as the input distribution is fixed, as will be detailed in this section. We will also discuss the implications of our results in the context of \acf{isac}.

\subsection{Generic Achievability}
Consider the separable finite-ray model
\begin{equation}
 \M{H}(\V{\xi},\V{\eta},\V{\alpha})
 =\sum_{k=1}^K
 \alpha_k\V{a}_r(\V{\xi}_k)\V{a}_t(\V{\eta}_k)^H,
 \label{eq:general-separable-channel}
\end{equation}
where \(\V{\xi}_k\in\mathbb R^{p_r}\) and \(\V{\eta}_k\in\mathbb R^{p_t}\) collect the real parameters describing the receive and transmit array responses of the \(k\)th path, respectively (e.g., \acp{aoa} and \acp{aod}). Let
\[
 \RM{X}=\M{B}\RM{S},\qquad
 \M{B}\in\C^{N_t\times L},\qquad K\le L\le T,
\]
where $\M{B}$ is a fixed precoder with full column rank. The entries of $\RM{S}$ are independent $\CN(0,\|\M{B}\|_F^{-2})$ variables, independent of the channel state and noise. Thus $\E{\|\M{B}\RM{S}\|_F^2}=T$, and $\RM{S}$ has a nondegenerate density on $\C^{L\times T}$ and full row rank almost surely. This distributional requirement ensures that data variations can explore all $L\times T$ complex directions.

In order to exclude a nonzero complex scaling factor that can be absorbed into the path gain, for a vector $\V{v}$, let us use $[\V{v}]_{\rm line}$ to denote its equivalence class under multiplication by a nonzero complex scalar. Next, we define the generic real projective rank of the projected steering vector by
\begin{equation}
 q_{\M{B}}
 =\rank_{\R}D_{\V{\eta}}[\M{B}^H\V{a}_t(\V{\eta})]_{\rm line}.
 \label{eq:projected-steering-rank}
\end{equation}
We further define the steering matrices
\[
 \begin{aligned}
 \M{A}_r(\V{\xi})
 &=(\V{a}_r(\V{\xi}_1),\ldots,\V{a}_r(\V{\xi}_K)),~~\M{A}_t(\V{\eta})
=(\V{a}_t(\V{\eta}_1),\ldots,\V{a}_t(\V{\eta}_K)),
 \end{aligned}
\]
and denote
\begin{subequations}
\begin{align}
 d_R
 &=\rank_{\R}D_{\V{\xi}}\operatorname{col}\M{A}_r(\V{\xi}),
 \label{eq:receive-grassmann-rank}\\
 d_H(\M{B})
 &=\rank_{\R}D_{(\V{\xi},\V{\eta},\Real\V{\alpha},\Imag\V{\alpha})}
 \bigl(\M{H}(\V{\xi},\V{\eta},\V{\alpha})\M{B}\bigr).
 \label{eq:effective-channel-rank}
\end{align}
\end{subequations}
Here $\operatorname{col}\M{A}_r$ is a point on the Grassmann manifold~\cite{ZhengTse2002}, whereas $d_H(\M{B})$ is the ordinary real Jacobian rank of the effective channel matrix.

\begin{proposition}[Conditional pre-log formula for a fixed precoder]
\label{prop:fixed-precoder-conditional}
Consider \eqref{eq:general-separable-channel} with the independent Gaussian input above and the independent additive Gaussian noise in \eqref{eq:channel}. Suppose that the following conditions hold.
\begin{enumerate}[label=(G\arabic*)]
 \item The joint distribution of
 $\RV{\xi},\RV{\eta},\Real\RV{\alpha},\Imag\RV{\alpha}$
 is absolutely continuous with respect to real Lebesgue measure.
 \item On a probability-1 open set of state coordinates, the steering
 maps are $C^1$, the gains are nonzero, $\M{A}_r$ has column rank $K$,
 and the matrix  $\M{A}_t(\V{\eta})^H\M{B}$ has row rank $K$. The two ranks $d_R$ and $d_H(\M{B})$ in \eqref{eq:receive-grassmann-rank}--\eqref{eq:effective-channel-rank}
 are constant on that set.
 \item The noiseless output
 $\RM{Z}=\M{H}(\RV{\xi},\RV{\eta},\RV{\alpha})\M{B}\RM{S}$
 satisfies $\E{\|\RM{Z}\|_F^2}<\infty$.
\end{enumerate}

The joint noiseless output map $(\V{\xi},\V{\eta},\V{\alpha},\M{S})\mapsto\M{H}(\V{\xi},\V{\eta},\V{\alpha})\M{B}\M{S}$ has real rank $r=2KT+d_R$ on a probability-1 open set. For almost every fixed $\M{S}$, the conditional map $(\V{\xi},\V{\eta},\V{\alpha})\mapsto\M{H}(\V{\xi},\V{\eta},\V{\alpha})\M{B}\M{S}$ has the constant real rank $r_u=d_H(\M{B})$ on a probability-1 open set of states.

Consequently, the pre-log achieved by this fixed input is
\begin{equation}
 \begin{aligned}
 d_{\M{B}}
 &:=\lim_{\rho\to\infty}
 \frac{I(\RM{S};\RM{Y}_\rho)}{T\log\rho}=\frac{2KT+d_R-d_H(\M{B})}{2T}.
 \end{aligned}
 \label{eq:fixed-precoder-prelog}
\end{equation}
\end{proposition}

\begin{IEEEproof}
See Appendix~\ref{app:proof-fixed-precoder-conditional}.
\end{IEEEproof}

The preceding proposition gives an achievable pre-log for a specified Gaussian input. The next proposition replaces the geometric condition (G2) by sufficient
rank conditions and computes the resulting ranks explicitly

\begin{proposition}[Sufficient conditions for the explicit pre-log]
\label{prop:explicit-precoder-conditions}
Consider the model and Gaussian input of \cref{prop:fixed-precoder-conditional}, and assume (G1) and (G3), without assuming (G2). Define
\[
 \V{b}(\V{\eta})=\M{B}^H\V{a}_t(\V{\eta}),\qquad
 \M{Q}_r=\M{I}_{N_r}-\M{A}_r\M{A}_r^\dagger.
\]
Suppose that the following conditions hold on a probability-1 open set of state coordinates.
\begin{enumerate}[label=(R\arabic*)]
 \item The steering maps are $C^1$, all gains are nonzero, and
 \[
  \rank_{\C}\M{A}_r=K,\qquad
  \rank_{\C}(\M{A}_t^H\M{B})=K.
 \]
 \item For every $k=1,\ldots,K$, the receive response derivatives
 retain all $p_r$ real directions outside the receive column space:
 \[
  \rank_{\R}\!\left(
   \M{Q}_rD_{\V{\xi}_k}\V{a}_r(\V{\xi}_k)
  \right)=p_r.
 \]
 \item For every $k=1,\ldots,K$, the projected transmit response
 has the same constant real projective rank $q_{\M{B}}$:
 \[
  \rank_{\R}\!\left[
   \left(\M{I}_L-
    \frac{\V{b}(\V{\eta}_k)\V{b}(\V{\eta}_k)^H}
         {\|\V{b}(\V{\eta}_k)\|_2^2}\right)
   D_{\V{\eta}_k}\V{b}(\V{\eta}_k)
  \right]=q_{\M{B}}.
 \]
\end{enumerate}
In (R2)--(R3), the parameter variations are real and the output is realified before taking rank. Then
\[
 d_R=Kp_r,\qquad
 d_H(\M{B})=K(p_r+q_{\M{B}}+2),
\]
and the Gaussian input achieves
\begin{equation}
 d_{\M{B}}
 =K\left(1-\frac{q_{\M{B}}+2}{2T}\right).
 \label{eq:projected-geometry-prelog}
\end{equation}
\end{proposition}
\begin{IEEEproof}
See Appendix~\ref{app:proof-explicit-precoder-conditions}.
\end{IEEEproof}

To elaborate, Condition (R1) separates the path coefficients at both Tx and Rx. Condition (R2) ensures that every receiver-side parameter (e.g., \ac{aoa}) variation changes the receive subspace, while (R3) counts transmit variations after removing the complex scaling absorbed by the gain. Thus the $Kp_r$ receiver-side parameter dimensions occur in both (conditional and unconditional) output dimensions and cancel in their difference. For the \ac{nula} subarray precoder in \eqref{eq:gaussian-ach-input}, \cref{lem:confluent} verifies the required receive ranks and $q_{\M{B}_K}=1$ for $K\ge2$.

\subsection{Examples of Other Array Structures}\label{ssec:examples_arrays}
We next give two concrete array structures for which $q_{\M{B}}$ can be evaluated explicitly. For both examples, we assume that the joint state distribution is absolutely continuous as in (G1). The responses in both examples are bounded uniformly over the directions. Using in addition the fact that $\E{\|\RV{\alpha}\|^2}<\infty$, for a fixed constant $C$,
$$
\E{\|\RM{H}\M{B}\RM{S}\|_F^2}\leq C\E{\|\RV{\alpha}\|^2}\E{\|\RM{S}\|_F^2}<\infty,
$$
which implies (G3). The remaining conditions (R1)--(R2) are verified in Appendix~\ref{app:proof-array-examples}, which in turn ensure (G2) via the computed $q_{\M{B}}$ and \cref{prop:explicit-precoder-conditions}. Next, we specify the receiver constructions and compute $q_{\RM{B}}$ explicitly.

\subsubsection{Uniform planar array with antenna selection}
\label{sec:upa-selection-example}
For a \ac{upa}, let us denote its response at index $(m,n)$ as $a_{m,n}(u,v)\propto e^{\jmathu(mu+nv)}$. Here $u,v$ are the two directional cosines multiplied by the fixed nonzero spacing factors. Let $\M{B}$ be an antenna selection matrix that selects $L$ distinct elements with indices $(m_\ell,n_\ell)$, $\ell=1,\ldots,L$. Their relative responses satisfy
\[
 \frac{b_\ell(u,v)}{b_1(u,v)}
 =e^{\jmathu[(m_\ell-m_1)u+(n_\ell-n_1)v]}.
\]
Consequently, we obtain
\begin{equation}
 q_{\M{B}}
 =\rank_{\R}
 \begin{bmatrix}
  m_2-m_1&n_2-n_1\\
  \vdots&\vdots\\
  m_L-m_1&n_L-n_1
 \end{bmatrix},
 \label{eq:upa-selection-projective-rank}
\end{equation}
which equals zero when $L=1$. Selecting merely collinear elements with $L\geq 2$ gives $q_{\M{B}}=1$, while selecting at least three noncollinear elements gives $q_{\M{B}}=2$. For example, the elements $(0,0),(1,0),(0,1)$ retain both directional variations. 

\subsubsection{Rectangular aperture with Fourier-mode selection}
\label{sec:fourier-mode-example}
For a continuous rectangular aperture, a Fourier representation provides a convenient finite-mode signal model~\cite{Pizzo2020}. Here we apply this representation to the finite-ray \ac{sv} channel. Let us normalize each aperture coordinate to $[-1/2,1/2]$. The orthonormal mode with integer index $(m,n)$ is given by
\[
 \phi_{m,n}(x,y)=e^{\jmathu2\pi(mx+ny)}.
\]
Its inner product with a far-field plane wave of normalized spatial frequencies $(u,v)$ is
\begin{equation}
 b_{m,n}(u,v) =\operatorname{sinc}(u-m)\operatorname{sinc}(v-n),
 \label{eq:aperture-mode-response}
\end{equation}
where $\operatorname{sinc}(x)=\frac{\sin(\pi x)}{\pi x}$ with $\operatorname{sinc}(0)=1$. Select $L$ modes from a fixed finite mode set. The finite input and output mode counts take the roles of $N_t,N_r$, and the transmitted modal coefficients obey the same average power constraint.

It is clear that a single mode gives $q_{\M{B}}=0$. If all selected modes have a common second index and $L\geq 2$, their distinct first indices give $q_{\M{B}}=1$ almost everywhere. If the selection contains $(0,0),(1,0),(0,1)$, then $q_{\M{B}}=2$ generically, since on the probability-1 event of $b_{0,0}\neq 0$, we have
\begin{equation}\label{eq:ratios_fourier}
 \frac{b_{1,0}}{b_{0,0}}=\frac{u}{1-u},
 \qquad
 \frac{b_{0,1}}{b_{0,0}}=\frac{v}{1-v}.
\end{equation}

We may now conclude that the method of quotient rank yield useful achievable \ac{dof} lower estimates for generic array structures. Nevertheless, it is in general difficult to obtain useful converse bounds, or prove the optimality of the obtained \ac{dof} estimates.

\subsection{Application to \ac{isac}: \ac{dof} Tradeoff}
Being one of the major usage scenarios in \ac{6g} networks, \ac{isac} aims to improve resource efficiency and facilitate mutual benefits between sensing and communication, by integrating the two functionalities in a unified platform \cite{yuanhaoNW}. In its simplistic form, an \ac{isac} system may be modelled as the following pair of channels\footnote{Here we consider a bistatic system so that the \acp{aoa} can be independent of \acp{aod}. The $\RM{X}$ is assumed to be known at the sensing Rx, which can be achieved by, for example, cable connections.}
\begin{equation}
\RM{Y}_{\rm c}=\RM{H}_{\rm c}\RM{X}+\RM{N}_{\rm c},~~\RM{Y}_{\rm s}=\RM{H}_{\rm s}\RM{X}+\RM{N}_{\rm s},
\end{equation}
where $\RM{Y}_{\rm c}\in\mathbb{C}^{N_{\rm c}\times T}$, $\RM{Y}_{\rm s}\in\mathbb{C}^{N_{\rm s}\times T}$, $\RM{H}_{\rm c}\in\mathbb{C}^{N_{\rm c}\times N_{\rm t}}$, and $\RM{H}_{\rm s}\in\mathbb{C}^{N_{\rm s}\times N_{\rm t}}$ denote the received communication and sensing signals, the communication and sensing channels, respectively. The sensing channel $\RM{H}_{\rm s}$ is also known as the target response matrix. The sensing and communication noises, $\RM{N}_{\rm s}$ and $\RM{N}_{\rm c}$, are assumed to be entrywise zero-mean i.i.d. proper complex Gaussian noises and are independent of each other, with variances being $\sigma_{\rm s}^2$ and $\sigma_{\rm c}^2$, respectively. The transmitted signal $\RM{X}$ is shared between sensing and communication functionalities, and is assumed known to the sensing receiver. An injection $\RV{\eta}\mapsto\RM{H}_{\rm s}(\RV{\eta})$ determines the sensing channel, and the sensing task is to estimate $\RV{\eta}$.

Under the aforementioned model, when $\RM{H}_{\rm c}$ is known to the communication receiver, the sensing-communication tradeoff has been investigated \cite{XiongTIT,mmse_isac}. In particular, when the sensing performance is optimal, there is typically a communication \ac{dof} loss. By contrast, when $\RM{H}_{\rm c}$ is not known, which is more relevant to the topic of this paper, \cite{noncoherent_isac} shows that no sensing-induced \ac{dof} loss exists, if the entries of $\RM{H}_{\rm c}$ are i.i.d. and follow $\mathcal{CN}(0,1)$ (also known as independent Rayleigh block fading). This distinction originates from the fact that the communication-\ac{dof}-achieving input distribution in the independent Rayleigh block fading case is unitary-invariant, which is even more restrictive than the sensing-induced orthogonal constraints. To elaborate, the communication-optimal information carriers live on the $2K(T-K)$-dimensional Grassmann manifold, while the sensing-optimal constraint only confines the input to the $2K(T-K/2)$-dimensional Stiefel manifold.

In this paper, we have shown that in the \ac{nula}-\ac{sv} model, the \ac{dof} loss due to unknown \ac{csi} is at most $3K/2T$, which is smaller than $K^2/2T$ \ac{dof} loss induced by sensing constraints when $K\geq 4$, as long as the sensing-optimal sample covariance matrix is unique (see \cite{XiongTIT} for a detailed discussion on the uniqueness issue). This implies that the sensing-communication \ac{dof} tradeoff may exist in the \ac{nula}-\ac{sv} model, even if the communication Rx does not have \ac{csi}.

Next, let us give a concrete example. We consider the case where $N_t=K\ge4$, $N_c,N_s\ge K+1$, and $T\ge K$. Both communication and sensing channels satisfy the finite-ray model
$$
\RM{H}_i =\sum_{k=1}^{K}\rv{\alpha}_{i,k} \V{a}_{i,r}(\rv{u}_{i,r,k})\V{a}_t(\rv{u}_{i,t,k})^H,
~~i\in\{{\rm c,s}\},
$$
where all directional cosines follow $\mathcal{U}(-1,1)$, while the complex gains follow independent proper complex Gaussian distributions. Both the communication Tx and Rx know the distributions of these parameters but not their realizations. We assume that the antenna elements are located at integer multiples of half-wavelength points, namely $x_n=\frac{m_n}{2}$, with $m_n$'s being mutually distinct integers. Thereby the \textit{a priori} covariance matrix is given by ${\rm Cov}({\rm vec}~ \RM{H}_{\rm s})=\beta\M{I}$, where $\beta>0$ is a constant.

Let us now consider the sensing task of estimating the full target response matrix $\RM{H}_{\rm s}$, under the constraint of using linear estimators, with the performance metric being the \ac{mse}. The optimal estimator in this case is known to be the \ac{lmmse} estimator, whose \ac{mse} is given by
\begin{equation}
J(\RM{R}_{\RM{X}})=N_{\rm s}\tr\left(\beta^{-1}\M{I}+\frac{T}{\sigma_{\rm s}^2}\RM{R}_{\RM{X}}\right)^{-1},
\end{equation}
where $\RM{R}_{\RM{X}}=(1/T)\RM{X}\RM{X}^H$. Under the average power constraint $\mathbb{E}[\tr\RM{R}_{\RM{X}}]\leq P$, Jensen's inequality yields
\begin{align}
\mathbb{E}\{J(\RM{R}_{\RM{X}})\}&\geq N_{\rm s}K\mathbb{E}\left[\beta^{-1}+\frac{T}{K\sigma_{\rm s}^2}\tr\RM{R}_{\RM{X}}\right]^{-1}\nonumber \\
&\geq \frac{N_{\rm s}K}{\beta^{-1}+\frac{T}{K\sigma_{\rm s}^2}P},
\end{align}
implying that the sensing-optimal $\RM{R}_{\RM{X}}$ is given by $\RM{R}_{\RM{X}}^\star=\frac{P}{K}\M{I}$ almost surely. Therefore, the sensing-optimal sample covariance matrix is unique and full-rank, and hence this example indeed exhibits a sensing-communication \ac{dof} tradeoff. 
\begin{remark}
Intuitively, the sensing task wish to learn every unknown parameter with high accuracy, and hence assign equal power to every direction. By contrast, the communication task does not need to assign additional resources to learn \acp{aoa}, since the accuracy ensured by blind estimation is already sufficient to achieve the \ac{dof}.
\end{remark}

We note that our results are restricted to this specific example, and especially under the assumption of using linear estimators. Our hope is that the proposed analytical framework can be extended and applied to reveal more general sensing-communication \ac{dof} tradeoffs.

\section{Conclusion}
\label{sec:conclusion}
In this paper, we have determined the capacity pre-log of the blockwise memoryless NULA-SV channel without prior \ac{csi}, for arbitrary fixed distinct array positions, including \acp{ula} as an important special case. For \(K\ge2\), the quotient rank argument gives the net loss of \(3K\) real dimensions, and a fixed \(K\)-antenna Gaussian input attains the resulting pre-log. For $K=1$, a single-antenna transmission strategy attains the \ac{dof} by reducing the observation to a scalar block-fading channel. The uniform converse shows that input distributions that depend on the \ac{snr} cannot achieve higher \acp{dof}.

In general, the quotient rank argument suggests that the achievability analysis also apply to other finite-ray array models, as illustrated by the \ac{upa} and rectangular aperture examples. However, extending the capacity theorem requires a matching converse for their additional degeneracy sets, which deserves future investigations.

\appendices
\section{Information-Dimension Results}
\subsection{\texorpdfstring{Proof of \cref{lem:smooth-image}}{Proof of Lemma \getrefnumber{lem:smooth-image}}}
\label{app:proof-lem-smooth-image}
\label{app:rectifiable-transfer-proof}
\begin{IEEEproof}
Let $\RV{z}=f(\RV{\theta},\RV{u})\in\R^n$. By (A1)--(A2), the constant-rank theorem and second countability~\cite[Thm.~4.12]{Lee2013} give a countable open cover $\{\Set{G}_j\}_{j\ge1}$ of the probability-1 regular set. After shrinking and refining the charts, each has the normal form
\[
 \psi_j\circ f\circ\phi_j^{-1}(\V{s},\V{w})
 =(\V{s},\V{0}_{n-r}),
 \qquad \V{s}\in\Set{D}_j\subset\R^r,
\]
where $\Set{D}_j$ is bounded and $\iota_j(\V{s})=\psi_j^{-1}(\V{s},\V{0}_{n-r})$ is bi-Lipschitz on $\Set{D}_j$. Such a refinement is possible because the coordinate maps are $C^1$ diffeomorphisms and are locally bi-Lipschitz.
Disjointify the cover by
\[
 \Set{E}_1=\Set{G}_1,\qquad
 \Set{E}_j=\Set{G}_j\setminus\bigcup_{i<j}\Set{G}_i,
\]
and let $\rv{J}=j$ on $\{(\RV{\theta},\RV{u})\in\Set{E}_j\}$, discarding null pieces. On every nonnull piece, $\phi_j(\RV{\theta},\RV{u})$ has an absolutely continuous conditional distribution by (A1). Its first $r$ coordinates $\RV{s}_j$ are therefore absolutely continuous and bounded, and $\RV{z}=\iota_j(\RV{s}_j)$ on that piece. Absolute continuity and bi-Lipschitz invariance of information dimension \cite{Renyi1959,StotzBolcskei2016} give $d_I(\RV{z}\mid\rv{J}=j)=r$. If $r=0$, the map is constant on each normal-form chart and the same conclusion holds with dimension zero.

To control the total information dimension over countably many charts, fix $M\ge1$, and define the finite label
\[
 \rv{J}_M=
 \begin{cases}
 \rv{J},&\rv{J}\le M,\\
 0,&\rv{J}>M.
 \end{cases}
\]
By (A3), the unconditional and every nonnull conditional distribution of $\RV{z}$ have finite second moments and hence finite entropy in their integer parts~\cite[Remark~2 and Lemma~11]{StotzBolcskei2016}. Thus we have
\[
 H(\quant{\RV{z}}{m}\mid\rv{J}_M)
 \le H(\quant{\RV{z}}{m})
 \le H(\rv{J}_M)+H(\quant{\RV{z}}{m}\mid\rv{J}_M).
\]
For a nonnull tail event, we have
\[
 \begin{aligned}
 H(\quant{\RV{z}}{m}\mid\rv{J}>M)
 &\le H(\lfloor\RV{z}\rfloor\mid\rv{J}>M)+n\log(m+1).
 \end{aligned}
\]
Consequently, it holds that $0\leq \underline{d}_I(\RV{z}|\rv{J}>M)\leq \overline{d}_I(\RV{z}|\rv{J}>M)\leq n$. Expanding $H(\quant{\RV{z}}{m}\mid\rv{J}_M)$ and dividing by $\log m$, we obtain
$$
\begin{aligned}
\frac{H(\quant{\RV{z}}{m}\mid \rv{J}_M)}{\log m}
={}&\sum_{j=1}^{M}\Prob{\rv{J}=j}
\frac{H(\quant{\RV{z}}{m}\mid \rv{J}=j)}{\log m}+\Prob{\rv{J}>M}\frac{H(\quant{\RV{z}}{m}\mid \rv{J}>M)}{\log m},
\end{aligned}
$$
which implies that
$$
(1-\Prob{\rv{J}>M})r
 \le\underline d_I(\RV{z}|\rv{J}_M)
 \le\overline d_I(\RV{z}|\rv{J}_M)\le(1-\Prob{\rv{J}>M})r+\Prob{\rv{J}>M}n.
$$
But removing the condition on $\rv{J}_M$ only leads to a negligible entropy increase in the limit of $m\rightarrow \infty$, and hence
\[
 \begin{aligned}
 (1-\Prob{\rv{J}>M})r
 &\le\underline d_I(\RV{z})
 \le\overline d_I(\RV{z})\le(1-\Prob{\rv{J}>M})r+\Prob{\rv{J}>M}n.
 \end{aligned}
\]
Since $\Prob{\rv{J}>M}\to0$ as $M\to \infty$, this proves $d_I(\RV{z})=r$.

For $P_{\RV{\theta}}$-almost every $\V{\theta}$, (A1) gives absolute continuity of $P_{\RV{u}\mid\RV{\theta}=\V{\theta}}$, (A2) gives constant rank $r_u$ on its probability-1 open section, and (A3) gives a finite conditional second moment. Applying the same chart argument to $\V{u}\mapsto f(\V{\theta},\V{u})$ yields
\[
 d_I(P_{\RV{z}\mid\RV{\theta}=\V{\theta}})=r_u
 \quad\text{for }P_{\RV{\theta}}\text{-a.e. }\V{\theta}.
\]
Finally, (A3) implies $H(\lfloor\RV{z}\rfloor)<\infty$, and hence the conditional dimension identity \cite[Lemma~2]{GeigerKoch2019} gives
\[
 d_I(\RV{z}\mid\RV{\theta})
 =\int d_I(P_{\RV{z}\mid\RV{\theta}=\V{\theta}})
 \,P_{\RV{\theta}}(d\V{\theta})=r_u.
\]
Thus the proof is completed.
\end{IEEEproof}

\subsection{\texorpdfstring{Proof of \cref{prop:quotient-rank}}{Proof of Proposition \getrefnumber{prop:quotient-rank}}}
\label{app:proof-cor-quotient-rank}
\begin{IEEEproof}
Independence of the noise gives the Markov chain $(\RV{\theta},\RV{u})\to\RV{z}\to\RV{y}_\rho$ and hence
\[
 I(\RV{\theta};\RV{y}_\rho)
 =I(\RV{z};\RV{y}_\rho)
  -I(\RV{z};\RV{y}_\rho\mid\RV{\theta}).
\]
By \cref{lem:smooth-image,lem:small-noise},
\[
 \frac{I(\RV{z};\RV{y}_\rho)}{\log\rho}\longrightarrow\frac r2,
\]
and, for $P_{\RV{\theta}}$-almost every $\V{\theta}$,
\[
 \frac{I(\RV{z};\RV{y}_\rho\mid\RV{\theta}=\V{\theta})}
 {\log\rho}\longrightarrow\frac{r_u}{2}.
\]
It remains to justify averaging the latter limit. Put $M(\V{\theta})=\E{\|\RV{z}\|^2\mid
\RV{\theta}=\V{\theta}}$ and $c_{\M{\Sigma}}=1/[n\lambda_{\min}(\M{\Sigma})]$. Whitening the noise and applying the Gaussian maximum-entropy bound~\cite[Thm.~2.8]{PolyanskiyWu2025} give, for $\rho\ge e$,
\begin{equation}
 \begin{aligned}
 0&\le\frac{I(\RV{z};\RV{y}_\rho\mid
 \RV{\theta}=\V{\theta})}{\log\rho}\\
 &\le\frac n2\,
 \frac{\log[1+\rho c_{\M{\Sigma}}M(\V{\theta})]}{\log\rho}\\
 &\le\frac n2\bigl\{1+
 \log[1+c_{\M{\Sigma}}M(\V{\theta})]\bigr\}.
 \end{aligned}
 \label{eq:conditional-mi-dominator}
\end{equation}
The last bound is integrable: Jensen's inequality and (A3) imply
\[
 \begin{aligned}
 \E{\log[1+c_{\M{\Sigma}}M(\RV{\theta})]}
 &\le\log[1+c_{\M{\Sigma}}
 \E{\|\RV{z}\|^2}]\\
 &<\infty.
 \end{aligned}
\]
Dominated convergence therefore gives
\[
 \frac{I(\RV{z};\RV{y}_\rho\mid\RV{\theta})}{\log\rho}
 \longrightarrow\frac{r_u}{2}.
\]
Substitution into the mutual-information difference yields
\[
 \lim_{\rho\to\infty}
 \frac{I(\RV{\theta};\RV{y}_\rho)}{\log\rho}
 =\frac{r-r_u}{2},
\]
which is equivalent to \eqref{eq:nuisance}.
\end{IEEEproof}

\section{Jacobian Rank and Regularity}
\subsection{\texorpdfstring{Proof of \cref{lem:confluent}}{Proof of Lemma \getrefnumber{lem:confluent}}}
\label{app:proof-lem-confluent}
\begin{IEEEproof}
It suffices to show that the determinant of a square submatrix formed by selecting \(K+1\) rows of \([\M{A}(\V{u}),\dot{\V{a}}(u_k)]\) is a complex-valued real-analytic function of the directions that is not identically zero. For \(m\) selected distinct frequencies, choose distinct real numbers \(t_1,\ldots,t_m\) and an interior point \(u_0\) of the direction interval \(\Set{I}\). Taylor expansion gives
\begin{multline}
 \det[e^{\jmathu\lambda_n(u_0+\epsilon t_j)}]_{n,j=1}^m
 =e^{\jmathu u_0\sum_n\lambda_n}
 \frac{(\jmathu\epsilon)^{m(m-1)/2}}{\prod_{r=0}^{m-1}r!}\prod_{n<\ell}(\lambda_\ell-\lambda_n)
 \prod_{i<j}(t_j-t_i)
 +O(\epsilon^{m(m-1)/2+1}).
 \label{eq:ordinary-expansion}
\end{multline}
Taking \(m=K\) proves that an ordinary minor is not identically zero. Since this determinant is real-analytic in the directions, it is nonzero almost everywhere, and hence \(\M{A}(\V u)\) has full column rank almost everywhere.

Now, take \(m=K+1\), evaluate the first \(K\) columns at \(u_0+\epsilon t_j\), and use the derivative at \(u_0+\epsilon t_k\) as the last column. Its leading term is
\begin{equation}
 e^{\jmathu u_0\sum_n\lambda_n}
 \frac{\det[(\jmathu\lambda_n)^r]_{n=1,\ r=0}^{m,\ m-1}}
      {\prod_{r=0}^{m-1}r!}
 \det\M{V}_{\rm H}\,
 \epsilon^{m(m-1)/2-1},
 \label{eq:confluent-expansion}
\end{equation}
with a remainder of one higher order. The first \(K\) columns of \(\M{V}_{\rm H}\) are \((t_j^r)_{r=0}^{K}\), and its last column is \((r t_k^{r-1})_{r=0}^{K}\), whose first entry is zero. This Hermite evaluation matrix is nonsingular. Indeed, a polynomial of degree at most \(K\) that vanishes at all \(K\) distinct nodes and has zero derivative at \(t_k\) has at least \(K+1\) roots, and hence is identically zero. Both determinants in \eqref{eq:confluent-expansion} are therefore nonzero.

For sufficiently small nonzero \(\epsilon\), all chosen directions lie in \(\Set{I}\). The real-analytic squared modulus of each minor is consequently nontrivial, thus its zero set has Lebesgue measure zero~\cite{KrantzParks2002}. There are finitely many required minors, for which the union of their zero sets still has Lebesgue measure zero. Thus we are done.
\end{IEEEproof}

\subsection{\texorpdfstring{Proof of \cref{lem:channel-rank}}{Proof of Lemma \getrefnumber{lem:channel-rank}}}
\label{app:proof-lem-channel-rank}
\begin{IEEEproof}
Write
\[
 \M{H}=\M{A}_r\M{D}(\V{\alpha}) \M{A}_t^H .
\]
Suppose that a first-order variation satisfies $\delta\M{H}=\M{0}$. Write the AoA variation as
\[
 \delta \M{A}_r
 =\dot{\M{A}}_r{\rm diag}(\delta \V{u}_r).
\]
Then
\begin{align}
 \delta\M{H}
 &=\dot{\M{A}}_r{\rm diag}(\delta \V{u}_r)
 \M{D}(\V{\alpha}) \M{A}_t^H+\M{A}_r{\rm diag}(\delta\V{\alpha})\M{A}_t^H+\M{A}_r\M{D}(\V{\alpha})
 {\rm diag}(\delta \V{u}_t)\dot{\M{A}}_t^H\nonumber\\
 &=\M{0}.
 \label{eq:channel-differential}
\end{align}

Since $\M{A}_r$ has full column rank, define the orthogonal projector
\[
 \M{Q}_r=\M{I}_{N_r}-\M{A}_r\M{A}_r^\dagger.
\]
Left-multiplying \eqref{eq:channel-differential} by $\M{Q}_r$ gives
\[
 \M{Q}_r\dot{\M{A}}_r{\rm diag}(\delta \V{u}_r)
 \M{D}(\V{\alpha})\M{A}_t^H=\M{0},
\]
since $\M{Q}_r\M{A}_r=\M{0}$. The matrix $\M{D}(\V{\alpha})\M{A}_t^H$ has full row rank and hence a right inverse. Right-multiplication by that inverse yields $\M{Q}_r\dot{\M{A}}_r{\rm diag}(\delta \V{u}_r)=\M{0}$, with the $k$th column being $\delta u_{r,k}\M{Q}_r\dot{\V{a}}_{N_r}(u_{r,k})=\V{0}$. By the one-node conclusion of \cref{lem:confluent}, $\dot{\V{a}}_{N_r}(u_{r,k}) \notin\operatorname{col}\M{A}_r$, so $\M{Q}_r\dot{\V{a}}_{N_r}(u_{r,k})\ne\V{0}$ and therefore
\[
 \delta u_{r,k}=0,\qquad k=1,\ldots,K.
\]

Left-multiplying \eqref{eq:channel-differential} by $\M{A}_r^\dagger$ and substituting $\delta \V{u}_r=\V{0}$ yields
\[
 \M{0}=\M{A}_r^\dagger\delta\M{H}
 ={\rm diag}(\delta\V{\alpha})\M{A}_t^H
 +\M{D}(\V{\alpha}) {\rm diag}(\delta \V{u}_t)\dot{\M{A}}_t^H.
\]
The $k$th row is
\[
 \delta\alpha_k \V{a}_t(u_{t,k})^H
 +\alpha_k\delta u_{t,k}
 \dot{\V{a}}_{N_t}(u_{t,k})^H=\V{0}^H.
\]
The steering vector $\V{a}_t(u)$ and its derivative are linearly independent in the complex field. Otherwise, division of their nonzero entries would force all $\jmathu2\pi x_{t,n}$ to be equal, contradicting the distinct positions and $N_t\ge2$. Since $\alpha_k\ne0$, the row equation therefore gives $\delta\alpha_k=0$ and $\delta u_{t,k}=0$. We may now conclude that $\ker_\R D_{(\V{u}_r,\V{u}_t,\V{\alpha})}\M{H}=\{\V{0}\}$, which implies that
\[
 \dim_\R(\V{u}_r,\V{u}_t,\V{\alpha})
 =K+K+2K=4K,
\]
\[
 \rank_\R D_{(\V{u}_r,\V{u}_t,\V{\alpha})}\M{H}
 =4K-0=4K,
\]
which completes the proof.
\end{IEEEproof}

\subsection{\texorpdfstring{Proof of \cref{lem:joint-output-rank}}{Proof of Lemma \getrefnumber{lem:joint-output-rank}}}
\label{app:proof-lem-joint-output-rank}
\begin{IEEEproof}
Define the selected transmit steering matrix
\[
 \widetilde{\M{A}}_t
 =\M{B}_K^H\M{A}_t.
\]
It is invertible by the stated rank condition. This condition holds almost surely by \cref{lem:confluent}, applied to the first $K$ distinct transmit positions. Since the gains are nonzero, $\M{D}(\V{\alpha})\widetilde{\M{A}}_t^H$ is also invertible. Set
\[
 \M{C}=\M{D}(\V{\alpha})\widetilde{\M{A}}_t^H\M{S},
 \qquad
 \M{Z}=\M{A}_r(\V{u}_r)\M{C}.
\]
Then $\M{C}$ has full row rank. Holding $\V{u}_t$ and $\V{\alpha}$ fixed, every variation $\delta\M{C}$ is obtained by choosing
\[
 \delta\M{S}
 =\bigl[\M{D}(\V{\alpha})\widetilde{\M{A}}_t^H\bigr]^{-1}
 \delta\M{C}.
\]
Thus the map from the original parameters to $(\V{u}_r,\M{C})$ is a local submersion. The original noiseless output map has the same local image and Jacobian rank as $(\V{u}_r,\M{C})\mapsto\M{A}_r(\V{u}_r)\M{C}$.

Now suppose $\dot{\M{A}}_r{\rm diag}(\delta \V{u}_r)\M{C} +\M{A}_r\delta\M{C}=\M{0}$. Left-multiplication by $\M{Q}_r=\M{I}_{N_r}-\M{A}_r\M{A}_r^\dagger$ and right-multiplication by a right inverse of $\M{C}$ give
\[
 \M{Q}_r\dot{\M{A}}_r{\rm diag}(\delta \V{u}_r)=\M{0}.
\]
As in the proof of \cref{lem:channel-rank}, the one-node conclusion of
\cref{lem:confluent} gives
\[
 \delta u_{r,k}=0,\qquad k=1,\ldots,K.
\]
The original differential then reduces to $\M{A}_r\delta\M{C}=\M{0}$, and the full column rank of $\M{A}_r$ gives $\delta\M{C}=\M{0}$. Thus,
\[
 \ker_\R D_{(\V{u}_r,\M{C})}(\M{A}_r\M{C})=\{\V{0}\},
\]
and
\[
 \rank_\R D_{(\V{u}_r,\M{C})}(\M{A}_r\M{C})
 =K+2KT.
\]
The original joint map has the same local image, we thus have
\[
 \rank_\R
 D_{(\V{u}_r,\V{u}_t,\V{\alpha},\M{S})}
 [\M{H}(\V{u}_r,\V{u}_t,\V{\alpha})\M{B}_K\M{S}]
 =K+2KT,
\]
and we are done.
\end{IEEEproof}

\subsection{\texorpdfstring{Proof of \cref{lem:fixed-input-rank}}{Proof of Lemma \getrefnumber{lem:fixed-input-rank}}}
\label{app:proof-lem-fixed-input-rank}
\begin{IEEEproof}
Since $\M{S}$ has full row rank, there exists $\M{S}^\dagger\in\C^{T\times K}$ such that $\M{S}\M{S}^\dagger=\M{I}_K$. Suppose
\[
 \begin{aligned}
 D_{(\V{u}_r,\V{u}_t,\V{\alpha})}[\M{H}\M{B}_K\M{S}][\delta\V{u}_r,\delta\V{u}_t,\delta\V{\alpha}]=\delta\M{H}\,\M{B}_K\M{S}=\M{0}.
 \end{aligned}
\]
Right-multiplication by $\M{S}^\dagger$ gives $\delta\M{H}\M{B}_K=\M{0}$. The effective channel is
\[
 \M{H}\M{B}_K
 =\M{A}_r\M{D}(\V{\alpha})\widetilde{\M{A}}_t^H,
 \qquad
 \widetilde{\M{A}}_t=\M{B}_K^H\M{A}_t.
\]
The matrix $\widetilde{\M{A}}_t$ is invertible, and for every path $k$, $\widetilde{\V{a}}_t(u_{t,k})=\M{B}_K^H\V{a}_t(u_{t,k})$ and its derivative $\dot{\widetilde{\V{a}}}_t(u_{t,k})$ are linearly independent, because the selected positions are distinct and $K\ge2$. Repeating the projector proof of \cref{lem:channel-rank} with $\widetilde{\M{A}}_t$ in place of $\M{A}_t$ therefore gives $\delta\V{u}_r=\V{0}$, $\delta\V{u}_t=\V{0}$, and $\delta\V{\alpha}=\V{0}$. The conditional Jacobian has trivial real kernel on a $4K$-dimensional parameter domain and hence has rank $4K$.
\end{IEEEproof}

\subsection{Proof of \cref{cor:regular-chart-cover}}
\label{app:proof-cor-regular-chart-cover}
\begin{IEEEproof}
The sets $\Set{O}_{\rm ch}$ and $\Set{O}_{\rm joint}$ are open in their respective realified Euclidean parameter spaces. Every NULA steering entry is real analytic in its directional cosine. Therefore the joint output map
\[
 (\V{u}_r,\V{u}_t,\V{\alpha},\M{S})
 \longmapsto\M{H}(\V{u}_r,\V{u}_t,\V{\alpha})\M{B}_K\M{S}
\]
and its restriction to each fixed $\M{S}$ are real analytic, and hence $C^1$, on the corresponding regular sets.

By \cref{lem:joint-output-rank}, the joint Jacobian has rank $2KT+K$ at every point of $\Set{O}_{\rm joint}$.  By \cref{lem:fixed-input-rank}, for every fixed $\M{S}$ having full row rank, the conditional Jacobian has rank $4K$ at every point of $\Set{O}_{\rm ch}$.  Since both sets are open and the maps are real analytic there, condition (A2) follows.

It remains to verify that the regular sets have probability one. By \cref{lem:confluent}, each required receive confluent minor and the selected transmit minor is a nontrivial analytic function of the directions. The union of their zero sets has Lebesgue measure zero and hence probability zero under (U1). By the absolute continuity of the directional variables, the boundary events and the finitely many collision events
\[
 \rv{u}_{\nu,i}=\rv{u}_{\nu,j},
 \qquad \nu\in\{r,t\},\quad i\ne j,
\]
have probability zero. The nondegenerate complex Gaussian gains satisfy
\[
 \Prob{\rv{\alpha}_k=0}=0,
 \qquad k=1,\ldots,K.
\]
Consequently, we have $\Prob{(\RV{u}_r,\RV{u}_t,\RV{\alpha})\in\Set{O}_{\rm ch}}=1$.

Finally, because $T\ge K$, the set $\{\M{S}\in\C^{K\times T}:\rank\M{S}<K\}$ is a proper algebraic subset of $\C^{K\times T}$ and therefore has Lebesgue measure zero. Hence every absolutely continuous input distribution satisfies $\Prob{\rank\RM{S}=K}=1$. Combining the two probability-1 events gives $\Prob{(\RV{u}_r,\RV{u}_t,\RV{\alpha},\RM{S})\in\Set{O}_{\rm joint}}=1$. Thus the joint and conditional maps satisfy (A2) of \cref{ass:regular}.
\end{IEEEproof}

\subsection{\texorpdfstring{Proof of \cref{lem:small-ball-ui}}{Proof of Lemma \getrefnumber{lem:small-ball-ui}}}
\label{app:proof-lem-small-ball-ui}
\begin{IEEEproof}
The Gaussian input satisfies
\begin{equation}
 \E{\|\RM{X}\|_F^2}
 =\E{\|\RM{S}\|_F^2}=T.
 \label{eq:gaussian-input-power}
\end{equation}
The realified density of $\RM{S}$ is
\begin{equation}
 p_{\RM{S}}(\M{S})
 =\left(\frac K\pi\right)^{KT}
 \exp(-K\|\M{S}\|_F^2),
 \qquad\M{S}\in\C^{K\times T}.
 \label{eq:gaussian-input-density}
\end{equation}
Together with U1--U2 and the independence in \eqref{eq:gaussian-ach-input}, this proves (A1). By \cref{cor:regular-chart-cover},
\begin{equation}
 \Prob{\rank\RM{S}=K}=1,
 \label{eq:gaussian-full-row-rank}
\end{equation}
and the joint and conditional ranks satisfy (A2). Since steering vectors have unit norms, the semi-unitarity of $\M{B}_K$ give
\[
 \|\RM{H}\M{B}_K\RM{S}\|_F
 \le\|\RV{\alpha}\|_1\|\RM{S}\|_F
 \le\sqrt K\,\|\RV{\alpha}\|_2\|\RM{S}\|_F.
\]
Using independence of $\RV{\alpha}$ and $\RM{S}$ yields
\[
 \begin{aligned}
 \E{\|\RM{H}\M{B}_K\RM{S}\|_F^2}
 &\le K\E{\|\RV{\alpha}\|_2^2}
          \E{\|\RM{S}\|_F^2}\\
 &=K^2\sigma_\alpha^2T<\infty.
 \end{aligned}
\]
This verifies (A3) and completes the proof.
\end{IEEEproof}

\section{Achievability}
\subsection{\texorpdfstring{Proof of \cref{prop:full-ach}}{Proof of Proposition \getrefnumber{prop:full-ach}}}
\label{app:proof-prop-full-ach}
\begin{IEEEproof}
Set $\RM{Z} =\M{H}(\RV{u}_r,\RV{u}_t,\RV{\alpha})\M{B}_K\RM{S}$. To apply \cref{prop:quotient-rank}, substitute its data-bearing variable $\RV{\theta}$ by $\realify(\RM{S})$, its nuisance variable $\RV{u}$ by $\bigl(\RV{u}_r,\RV{u}_t,\realify(\RV{\alpha})\bigr)$ and its noiseless output $\RV{z}$ by $\realify(\RM{Z})$. The realification of the white complex Gaussian noise has a fixed positive-definite covariance, as required by \cref{prop:quotient-rank}.

The power identity \eqref{eq:gaussian-input-power} makes the input admissible. By \cref{lem:small-ball-ui}, conditions (A1)--(A3) hold for the preceding choice of data and nuisance variables. On the probability-1 regular set, \cref{lem:joint-output-rank} gives $r=2KT+K$, while \cref{lem:fixed-input-rank} gives $r_u=4K$ for every fixed full-row-rank $\M{S}$. Hence \cref{prop:quotient-rank} and invariance of mutual information under realification yield
\begin{align*}
 I(\RM{S};\RM{Y}_\rho)
 &=\frac{r-r_u}{2}\log\rho+o(\log\rho)\\
 &=\frac{(2KT+K)-4K}{2}\log\rho
 +o(\log\rho)\\
 &=\left(KT-\frac{3K}{2}\right)\log\rho
 +o(\log\rho).
\end{align*}
Finally, $\RM{X}=\M{B}_K\RM{S}$ is a deterministic one-to-one reparameterization because $\M{B}_K^H\M{B}_K=\M{I}_K$. Hence
\[
 I(\RM{X};\RM{Y}_\rho)=I(\RM{S};\RM{Y}_\rho),
\]
and the per-channel-use pre-log is
\begin{align*}
 d_{\rm Gauss}
 &=\lim_{\rho\to\infty}
 \frac{I(\RM{X};\RM{Y}_\rho)}{T\log\rho}=K\left(1-\frac{3}{2T}\right).
\end{align*}
Thus the proof is completed.
\end{IEEEproof}

\subsection{\texorpdfstring{Proof of \cref{prop:single-beam}}{Proof of Proposition \getrefnumber{prop:single-beam}}}
\label{app:proof-prop-single-beam}
\begin{IEEEproof}
Choose $\RM{X}=\V{e}_1\RV{s}^T$, where $\RV{s}\sim\CN(\V{0},\M{I}_T)$ is independent of the channel and noise. Then $\E{\|\RM{X}\|_F^2}=T$. Since $x_{r,1}=x_{t,1}=0$ in the chosen array coordinates, keeping only the first receive antenna gives the column observation
\begin{align*}
 \RV{y}_{1,\rho}&=\sqrt\rho\,\rv{g}\RV{s}+\RV{n}_1,\\
 \rv{g}&:=\frac1{\sqrt{N_rN_t}}\sum_{k=1}^K\rv{\alpha}_k,
 \qquad
 \rv{g}\sim\CN\!\left(0,\frac{K\sigma_\alpha^2}{N_rN_t}\right).
\end{align*}
This is a scalar Rayleigh block-fading channel. For completeness, the map $(g,\V{s})\mapsto g\V{s}$ has joint real rank $2T$ and, for fixed $\V{s}\ne\V{0}$, nuisance rank $2$ on the probability-1 open set $\{g\ne0,\V{s}\ne\V{0}\}$. The Gaussian parameters are absolutely continuous and independent, and
$\E{\|\rv{g}\RV{s}\|_2^2}=
TK\sigma_\alpha^2/(N_rN_t)<\infty$.
Thus \cref{prop:quotient-rank}, with data $\realify(\RV{s})$
and nuisance $\realify(\rv{g})$, gives
\[
 I(\RV{s};\RV{y}_{1,\rho})
 =(T-1)\log\rho+o(\log\rho).
\]
Data processing and the one-to-one mapping $\RV{s}\mapsto\V{e}_1\RV{s}^T$ imply
\[
 I(\RM{X};\RM{Y}_\rho)
 \ge I(\RV{s};\RV{y}_{1,\rho}).
\]
Dividing by $T\log\rho$ proves the bound for every $K\ge1$,
including $K=1$ and $T=1$.
\end{IEEEproof}

\section{Genie Bounds}
\subsection{\texorpdfstring{Proof of \cref{prop:pathwise-genie}}{Proof of Proposition \getrefnumber{prop:pathwise-genie}}}
\label{app:proof-prop-pathwise-genie}
\begin{IEEEproof}
Given $\RM{Z}_\rho=\M{Z}$, draw mutually independent random variables $\widetilde{\rv{u}}_{r,1},\ldots,\widetilde{\rv{u}}_{r,K}$, independently of $\RM{Z}_\rho$, where $\widetilde{\rv{u}}_{r,k}$ has density $f_{\rv{u}_{r,k}}$ for every $k$. Set $\widetilde{\RV{u}}_r = (\widetilde{\rv{u}}_{r,1},\ldots,\widetilde{\rv{u}}_{r,K})$ and $\widetilde{\RM{A}}_r=\M{A}_r(\widetilde{\RV{u}}_r)$. Every column of $\widetilde{\RM{A}}_r$ has unit norm, implying that
\[
 \|\widetilde{\RM{A}}_r\|_2^2
 \le
 \|\widetilde{\RM{A}}_r\|_F^2
 =K,
\]
and hence
\[
 \RM{\Sigma}_{\widetilde{\RM{N}}}
 :=
 \M{I}_{N_r}
 -K^{-1}\widetilde{\RM{A}}_r\widetilde{\RM{A}}_r^H
 \succeq0.
\]
Conditioned on $\widetilde{\RV{u}}_r$, draw the $T$ columns of $\widetilde{\RM{N}}$ independently from $\CN(\V{0},\RM{\Sigma}_{\widetilde{\RM{N}}})$, independently of $\M{Z}$, and define
\[
 \widetilde{\RM{Y}}
 =
 \widetilde{\RM{A}}_r\M{Z}
 +\widetilde{\RM{N}}.
\]
Substituting \eqref{eq:pathwise-genie},
\[
 \widetilde{\RM{Y}}
 =
 \sqrt\rho\,
 \widetilde{\RM{A}}_r\M{D}(\RV{\alpha})
 \M{A}_t(\RV{u}_t)^H\RM{X}
 +\RM{N}',
\]
where $\RM{N}' = \widetilde{\RM{A}}_r\RM{N}_0 +\widetilde{\RM{N}}$. Its conditional covariance is
\[
 K^{-1}\widetilde{\RM{A}}_r\widetilde{\RM{A}}_r^H
 +\RM{\Sigma}_{\widetilde{\RM{N}}}
 =\M{I}_{N_r},
\]
so $\RM{N}'$ is white complex Gaussian noise. Moreover, we see that $\widetilde{\RV{u}}_r$ is independent of $(\RV{u}_t,\RV{\alpha},\RM{X})$. By (U1)--(U2), $(\widetilde{\RV{u}}_r,\RV{u}_t,\RV{\alpha},\RM{X})$ is identically distributed as $(\RV{u}_r,\RV{u}_t,\RV{\alpha},\RM{X})$, under the indexed coordinates specified in \cref{ass:ula}. Therefore,
\[
 P_{\widetilde{\RM{Y}}\mid\RM{X}}
 =
 P_{\RM{Y}_\rho\mid\RM{X}}.
\]
The preceding randomization depends on $\RM{X}$ only through $\RM{Z}_\rho$, and therefore defines the kernel $Q_\rho$ and the Markov chain $\RM{X}\longrightarrow\RM{Z}_\rho\longrightarrow\widetilde{\RM{Y}}$. The equality of the conditional output laws gives
\[
 I(\RM{X};\RM{Y}_\rho)
 =I(\RM{X};\widetilde{\RM{Y}})
 \le I(\RM{X};\RM{Z}_\rho),
\]
which proves \eqref{eq:pathwise-data-processing}. 
\end{IEEEproof}

\subsection{\texorpdfstring{Proof of \cref{prop:geom-genie}}{Proof of Proposition \getrefnumber{prop:geom-genie}}}
\label{app:proof-thm-geom-genie}
\begin{IEEEproof}
For $K=1$, write $\RV{a}_r=\V{a}_r(\rv{u}_r)$ and $\RV{a}_t=\V{a}_t(\rv{u}_t)$. According to the geometric genie, we reveal $\RS{V}=(\rv{u}_r,\rv{u}_t)$ to the receiver. Since the transmitter does not know the geometry, $\RM{X}\perp\RS{V}$, and hence
\begin{align}
 I(\RM{X};\RM{Y}_\rho)
 &\le I(\RM{X};\RM{Y}_\rho,\RS{V})=I(\RM{X};\RM{Y}_\rho\mid\RS{V}).
 \label{eq:single-geometric-genie}
\end{align}
Define the effective input and output by
\begin{align*}
 \RV{s}_{\rm eff}&:=\RV{a}_t^H\RM{X},\\
 \RV{z}_{\rm eff}&:=\RV{a}_r^H\RM{Y}_\rho=\sqrt\rho\,\rv{\alpha}\RV{s}_{\rm eff}+\RV{w},
 \qquad \RV{w}\sim\CN(\V{0},\M{I}_T).
\end{align*}
Conditioned on $\RS{V}$, the component of $\RM{Y}_\rho$ orthogonal to $\RV{a}_r$ is independent white Gaussian noise. Thus $\RV{z}_{\rm eff}$ is a sufficient statistic. Since $\RV{s}_{\rm eff}$ is a function of $(\RM{X},\RS{V})$, the conditional Markov chain $\RM{X}\longrightarrow\RV{s}_{\rm eff}\longrightarrow\RV{z}_{\rm eff}$ gives
\[
 I(\RM{X};\RM{Y}_\rho\mid\RS{V})
 =I(\RV{s}_{\rm eff};\RV{z}_{\rm eff}\mid\RS{V}).
\]
Moreover, $\|\RV{a}_t\|_2=1$ and $\RM{X}\perp\RS{V}$ imply
\[
 \E{\|\RV{s}_{\rm eff}\|_2^2\mid\RS{V}}
 \le\E{\|\RM{X}\|_F^2}\le T.
\]
For each revealed geometry, this is the scalar constant-block noncoherent channel. Its capacity upper bound is uniform over all input distributions under average power constraints, including \ac{snr}-dependent ones, and gives \cite{ZhengTse2002,Morgenshtern2013}
\[
 I(\RV{s}_{\rm eff};\RV{z}_{\rm eff}\mid\RS{V})
 \le(T-1)\log\rho+o(\log\rho).
\]
Combining the preceding bounds yields
\begin{align*}
 I(\RM{X};\RM{Y}_\rho)
 &\le(T-1)\log\rho+o(\log\rho),\\
 d_{\rm blind}^{\rm cap}
 &\le1-\frac1T.
\end{align*}
Thus we are done.
\end{IEEEproof}

\section{AoD Resolvability and the Conditional-Output Entropy Bound}
\emph{Convention for the appendices regarding converse:}
Unless stated otherwise, constants may depend on the fixed array positions, dimensions, directional supports and density bounds, gain variance, positive-definite noise covariance, and fixed shell or radius bounds. They never depend on $\rho$, the input matrix or its distribution, $\sigma_2$, the variable tube radius $\delta$, a quantization cell, or a shell index. In \cref{lem:preimage-volume} and its sublemmas the constants depend only on $(\Set{\Lambda},\Set{I},T)$. The constants in the multiscale proof may additionally depend on fixed $(A,\Delta)$.

\subsection{\texorpdfstring{Proof of \cref{thm:wronskian}}{Proof of Proposition \getrefnumber{thm:wronskian}}}
\label{app:proof-thm-wronskian}
\begin{IEEEproof}
If $\rank\M{X}\le1$, then $\V{w}_{\M{X}}\equiv0$ and \eqref{eq:w-sigma} holds. We next consider the case where $\rank\M{X}\ge2$. Take a singular value decomposition
\[
 \M{X}=\sum_{i=1}^{\rank \M{X}}\sigma_i \V{u}_i\V{v}_i^H,
\]
and define the scalar exponential polynomials
\[
 p_i(u)=\V{a}_t(u)^H\V{u}_i.
\]
Then
\begin{align*}
 \V{v}_{\M{X}}(u)
 &=\sum_{i=1}^{\rank \M{X}}\sigma_ip_i(u)\V{v}_i^H,~~\dot{\V{v}}_{\M{X}}(u)
 =\sum_{i=1}^{\rank \M{X}}\sigma_i\dot p_i(u)\V{v}_i^H.
\end{align*}
By the bilinearity and antisymmetry of the exterior product, we obtain
\begin{equation}
 \begin{aligned}
 \V{w}_{\M{X}}(u)
 &=
 \sum_{1\le r<s\le\rank\M{X}}\sigma_r\sigma_s[p_r(u)\dot p_s(u)-p_s(u)\dot p_r(u)]
 (\V{v}_r^H\wedge\V{v}_s^H).
 \end{aligned}
 \label{eq:exterior-svd}
\end{equation}
Since $\{\V{v}_r^H\wedge \V{v}_s^H:r<s\}$ is an orthonormal family in $\bigwedge^2\C^T$, the identity holds separately at each distinct frequency. Summing these coefficient identities gives
\begin{equation}
 \|\V{w}_{\M{X}}\|_{\rm coef}^2
 =
 \sum_{r<s}\sigma_r^2\sigma_s^2
 \|p_r(u)\dot p_s(u)-p_s(u)\dot p_r(u)\|_{\rm coef}^2.
 \label{eq:exterior-orthogonal-sum}
\end{equation}

Consider the compact complex Stiefel manifold
\[
 \operatorname{St}(2,N_t)
 =
 \{\M{B}\in\mathbb{C}^{N_t\times 2}:\M{B}^H\M{B}=\M{I}\}.
\]
For $\M{B}=[\V{b}_1,\V{b}_2]\in\operatorname{St}(2,N_t)$, define
\[
 p_i(u)=\V{a}_t(u)^H\V{b}_i,~~
 \kappa(\V{b}_1,\V{b}_2)
 =
 \|p_1(u)\dot p_2(u)-p_2(u)\dot p_1(u)\|_{\rm coef}.
\]
If $\kappa(\V{b}_1,\V{b}_2)=0$, then on every open interval where $p_1\ne0$, we have $\frac{d}{du}\left(\frac{p_2(u)}{p_1(u)}\right)=0$. Analytic continuation gives $p_2(u)=cp_1(u)$. The distinct exponentials are linearly independent, because their derivatives of orders $0,\ldots,N_t-1$ give a nonsingular frequency Vandermonde matrix. Comparing their coefficients therefore yields $\V{b}_2=c\V{b}_1$, contradicting $\V{b}_1^H\V{b}_2=0$. Hence the continuous function $\kappa$ attains a positive minimum on the compact set $\operatorname{St}(2,N_t)$
\[
 c
 :=
 \min_{\operatorname{St}(2,N_t)}
 \kappa(\V{b}_1,\V{b}_2)>0.
\]
Finally, by keeping only the $(r,s)=(1,2)$ term in \eqref{eq:exterior-orthogonal-sum}, we obtain
\[
 \|\V{w}_{\M{X}}\|_{\rm coef}
 \ge
 c\sigma_1(\M{X})\sigma_2(\M{X}),
\]
which proves \eqref{eq:w-sigma}.
\end{IEEEproof}

\subsection{\texorpdfstring{Proof of \cref{lem:preimage-volume}}{Proof of Lemma \getrefnumber{lem:preimage-volume}}}
\label{app:proof-lem-preimage-volume}
For the fixed NULA frequencies, we first prove a uniform bound on the number of zeros, and obtain a sublevel-set estimate from the Tur\'an--Nazarov inequality.
\begin{sublemma}[Uniform control of zeros and sublevel sets]\label{lem:exp-uniform}
Let \(\Set{\Lambda}=\{\lambda_1,\ldots,\lambda_L\}\) be a fixed set of distinct real frequencies and let \(\Set{I}\) be a fixed compact interval of positive length. For $f(u)=\sum_{\ell=1}^{L}c_\ell e^{\jmathu\lambda_\ell u}$, there exist \(Z,C<\infty\), depending only on \((\Set{\Lambda},\Set{I})\), such that every nonzero \(f\) has at most \(Z\) distinct real zeros on \(\Set{I}\), and
\begin{equation}
 \bigl|\{u\in\Set{I}:|f(u)|\le\eta\|f\|_{\rm coef}\}\bigr|
 \le C\eta^{1/q},
 \label{eq:sublevel}
\end{equation}
where $0<\eta<1$ and $q=\max\{1,L-1\}$.
\end{sublemma}
\begin{IEEEproof}
Normalize \(\|f\|_{\rm coef}=1\). The case \(L=1\) is immediate, thus we
assume \(L\ge2\) hereafter. Define the fixed derivative matrix
\[
 \M{V}_\Set{\Lambda}=[(\jmathu\lambda_\ell)^r]_{r=0,\ \ell=1}^{L-1,\ L}.
\]
It is nonsingular by the Vandermonde determinant. At every real \(u\),
\begin{equation}
 \begin{bmatrix}f(u)\\ f'(u)\\ \vdots\\ f^{(L-1)}(u)\end{bmatrix}
 =\M{V}_\Set{\Lambda}
   \diag(e^{\jmathu\lambda_1u},\ldots,e^{\jmathu\lambda_Lu})\V{c}.
 \label{eq:jet}
\end{equation}
Since the diagonal matrix is unitary, at least one derivative has modulus at least \(\sigma_{\min}(\M{V}_\Set{\Lambda})/\sqrt L\). One of its real or imaginary parts consequently has absolute value at least \(a:=\sigma_{\min}(\M{V}_\Set{\Lambda})/\sqrt{2L}>0\). Every derivative of orders \(1,\ldots,L\) is bounded in modulus by
\[
 B:=1+\max_{1\le j\le L}
       \left(\sum_\ell|\lambda_\ell|^{2j}\right)^{1/2}.
\]
Partition \(\Set{I}\) into at most \(1+\lceil2B|\Set{I}|/a\rceil\) intervals of
length at most \(a/(2B)\). At the left endpoint of each interval choose \(g=\Real f\) or \(g=\Imag f\) and an order \(0\le j\le L-1\) that achieve the bound \(a\). The derivative bound implies \( |g^{(j)}|\ge a/2\) throughout that interval. In particular, its sign is fixed. If \(j=0\), there are no zeros of \(g\) on the interval. If \(j\ge1\), by repeating Rolle's theorem we see that \(g\) has at most \(j\) distinct zeros. Zeros of \(f\) are also zeros of \(g\), proving a bound uniform in all nonzero coefficients.

The sublevel estimate is a direct consequence of the Tur\'an--Nazarov inequality~\cite[Theorem~I]{Nazarov1993}. Since the frequencies are fixed and distinct, norm equivalence on their finite-dimensional span gives $\sup_{\Set{I}}|f|\ge c_0>0$ under the normalization $\|f\|_{\rm coef}=1$, where $c_0$ depends only on $(\Set{\Lambda},\Set{I})$. Let $\Set{E}_\eta=\{u\in\Set{I}:|f(u)|\le\eta\}$. If $|\Set{E}_\eta|>0$, the cited inequality gives
\[
 c_0\le\sup_{\Set{I}}|f|
 \le\left(\frac{A|\Set{I}|}{|\Set{E}_\eta|}\right)^{L-1}
       \sup_{\Set{E}_\eta}|f|
 \le\left(\frac{A|\Set{I}|}{|\Set{E}_\eta|}\right)^{L-1}\eta,
\]
where $A$ is an absolute constant. The exponential prefactor in the general inequality is one because the exponents $\jmathu\lambda_\ell$ are purely imaginary. Rearranging yields
\[
 |\Set{E}_\eta|
 \le A|\Set{I}|c_0^{-1/(L-1)}\eta^{1/(L-1)},
\]
which also holds when $|\Set{E}_\eta|=0$. This proves \eqref{eq:sublevel} uniformly over the coefficients.
\end{IEEEproof}

Next we show the main lemma. 
\begin{IEEEproof}
Define the set
$$
\Set{G}_\tau=\{t\in\Set{I}:\|\V{v}(t)\|_2\geq \tau,~\|\V{w}(t)\|_2\geq \delta\tau\},
$$
where $\V{w}(t)=\V{v}(t)\wedge \dot{\V{v}}(t)$. By applying \cref{lem:exp-uniform} to components of $\V{v}$ and $\V{w}$, whose corresponding frequency sets are $\Set{\Lambda}$ and $\Set{\Lambda}+\Set{\Lambda}$, respectively, we see that $|\Set{I}\backslash\Set{G}_\tau|\leq C\tau^{\kappa}$. We also have $\sup_{\Set{I}}\|\V{v}\|_2=O(1)$ since $\|\V{v}\|_{\rm coef}=1$. Now, fix $u\in\Set{G}_\tau$, $\alpha\neq 0$, and $\varepsilon>0$. Let $\Set{E}_u$ be the set of directional cosines $t\in\Set{G}_\tau$, for where there exists some $\alpha^\prime\in\mathbb{C}$ such that 
\begin{equation}\label{preimage_condition_9}
\|\alpha^\prime\V{v}(t)-\alpha\V{v}(u)\|_2\leq \varepsilon.
\end{equation}
For $\|\V{v}(t)\|_2^2> 0$, define the projector $\M{\Pi}(t)=\V{v}(t)\V{v}(t)^H\|\V{v}(t)\|_2^{-2}$. An orthogonal projection gives
\begin{subequations}
\begin{align}
\|(\M{I}-\M{\Pi}(t))\V{v}(u)\|_2&=\|(\M{I}-\M{\Pi}(\V{t}))((\alpha^\prime/\alpha)\V{v}(t)-\V{v}(u))\|_2\leq\frac{\varepsilon}{|\alpha|},\\
\|\M{\Pi}(u)-\M{\Pi}(t)\|_F&=\frac{\sqrt{2}\|(\M{I}-\M{\Pi}(t))\V{v}(u)\|_2}{\|\V{v}(u)\|_2}\leq \frac{\sqrt{2}\varepsilon}{|\alpha|\tau}.\label{proj_proximity}
\end{align}
\end{subequations}
Let us denote the $j$-th entry of the realified version of $\M{\Pi}(t)$ as $p_j(t)$. From \eqref{proj_proximity} we see that $p_j(t)\in[p_j(u)-\sqrt{2}\varepsilon/(|\alpha|\tau),p_j(u)+\sqrt{2}\varepsilon/(|\alpha|\tau)]$, for all $t\in\Set{E}_u$. Using the one-dimensinal area formula, we obtain
\begin{equation}\label{area1d}
\int_{\Set{E}_u} |\dot{p}_j(t)|{\rm d}t=\int_{\mathbb{R}}\#\{t\in\Set{E}_u:p_j(t)=y\}{\rm d}y=\int_{p_j(u)-\sqrt{2}\varepsilon/(|\alpha|\tau)}^{p_j(u)+\sqrt{2}\varepsilon/(|\alpha|\tau)}\#\{t\in\Set{E}_u:p_j(t)=y\}{\rm d}y.
\end{equation}
These derivatives can be computed as
\begin{equation}\label{dotpi}
\dot{\M{\Pi}}(t)=\frac{(\M{I}-\M{\Pi}(t))\dot{\V{v}}(t)\V{v}(t)^H+\V{v}(t)\dot{\V{v}}(t)^H(\M{I}-\M{\Pi}(t))}{\|\V{v}(t)\|_2^2},
\end{equation}
and hence the entries take the form of $p_j(t)=a_j(t)/\|\V{v}(t)\|_2^2$, where both the numerator and the denominator are real-valued exponential polynomials with frequencies in $\Set{\Lambda}-\Set{\Lambda}$. For a non-constant $p_j$, we see that $a_j(t)-a\|\V{v}(t)\|_2^2$ is never identically zero for any $a\in\mathbb{R}$. According to \cref{lem:exp-uniform}, we have a uniform bound for all non-constant $p_j$, namely $\#\{t\in\Set{E}_u:p_j(t)=y\}\leq C_0$, implying that $\int_{\Set{E}_u} |\dot{p}_j(t)|{\rm d}t\leq 2\sqrt{2}C_0\varepsilon/(|\alpha|\tau)$. For the entire $\M{\Pi}(t)$, we have
\begin{align}
\int_{\Set{E}_u}\|\dot{\M{\Pi}}(t)\|_F~{\rm d} t&\leq \sum_j\int_{\Set{E}_u}|\dot{p}_j(t)|{\rm d}t\leq \frac{C_1\varepsilon}{|\alpha|\tau},
\end{align}
since there are finite number of entries in $\M{\Pi}(t)$. Now, from \eqref{dotpi} we obtain
\begin{equation}
\|\dot{\M{\Pi}}(t)\|_F=\left[\frac{2\|(\M{I}-\M{\Pi}(t))\dot{\V{v}}(t)\|_2^2}{\|\V{v}(t)\|_2^2}\right]^{\frac{1}{2}}=\frac{\sqrt{2}\|\V{w}(t)\|_2}{\|\V{v}(t)\|_2^2}\geq c\tau \delta,~~\forall t\in\Set{G}_\tau,
\end{equation}
for some positive constant $c$. Therefore
\begin{equation}
|\Set{E}_u|\leq \frac{1}{c\tau\delta}\int_{\Set{E}_u}\|\dot{\M{\Pi}}(t)\|_F~{\rm d}t\leq \frac{C_2\varepsilon}{|\alpha|\tau^2\delta}.
\end{equation}
Since $|\Set{E}_u|\leq |\Set{I}|$, we further have 
\begin{equation}
|\Set{E}_u|\leq C\min\left\{1,\frac{\epsilon}{|\alpha|\tau^2\delta}\right\}.
\end{equation}
For the complex gain $\alpha^\prime$, from \eqref{preimage_condition_9} we obtain
\begin{align}
\left|\alpha^\prime-\frac{\alpha\V{v}(t)^H\V{v}(u)}{\|\V{v}(t)\|_2^2}\right|&=\frac{|\V{v}(t)^H(\alpha^\prime\V{v}(t)-\alpha\V{v}(u))|}{\|\V{v}(t)\|_2^2}\nonumber \\
&\leq \frac{\varepsilon}{\|\V{v}(t)\|_2}\leq \frac{\varepsilon}{\tau}.
\end{align}
Fubini’s theorem therefore gives
\begin{align}
\int_{\Set{P}_\varepsilon(\alpha,u)}{\rm d}^2\alpha^\prime {\rm d}t&\leq \pi\varepsilon^2\tau^{-2}|\Set{E}_u|\nonumber \\
&\leq C\varepsilon^2\tau^{-2}\min\left\{1,\frac{\epsilon}{|\alpha|\tau^2\delta}\right\}\nonumber \\
&\leq C\varepsilon^2\tau^{-4}\min\left\{1,\frac{\epsilon}{|\alpha|\delta}\right\},
\end{align}
which completes the proof.
\end{IEEEproof}

\subsection{Gaussian quantization lower bound}
\begin{lemma}[Gaussian quantization lower bound]
\label{lem:quantization-lower}
Let $\RV{x}\in\R^n$ be independent of $\RV{n}\sim\mathcal N(\V{0},\M{I}_n)$, and suppose $H(\lfloor\RV{x}\rfloor)<\infty$. There exists a constant $C_n<\infty$ that depends only on $n$ such that, for every real $s\ge1$,
\begin{equation}
 I(\RV{x};s\RV{x}+\RV{n})
 \ge H(\lfloor s\RV{x}\rfloor)-C_n.
 \label{eq:quantization-lower}
\end{equation}
The statement remains valid for Gaussian noise with any fixed positive-definite covariance, in which case the constant may also depend on that covariance.
This is essentially the lattice-quantization counterpart of Wu's nonasymptotic small-ball sandwich \cite[Lemma 32]{WuThesis2011}.
\end{lemma}

\begin{IEEEproof}
Let
\[
 \RV{q}_s=\lfloor s\RV{x}\rfloor,\qquad
 \RV{r}_s=s\RV{x}-\RV{q}_s\in[0,1)^n,\qquad
 \RV{y}_s=s\RV{x}+\RV{n}.
\]
Since each coordinate of $\lfloor s\RV{x}\rfloor$ has at
most $\lceil s\rceil+1$ possible values given
$\lfloor\RV{x}\rfloor$,
\[
 H(\RV{q}_s)
 \le H(\lfloor\RV{x}\rfloor)
     +n\log(\lceil s\rceil+1)<\infty.
\]
Use $\widehat{\RV{q}}_s=\lfloor\RV{y}_s\rfloor$ to
estimate $\RV{q}_s$. The integer error satisfies
\[
 \RV{d}_s
 :=\RV{q}_s-\widehat{\RV{q}}_s
 =-\lfloor\RV{r}_s+\RV{n}\rfloor
 =-\lfloor\RV{n}\rfloor-\RV{b}_s,
\]
where $\RV{b}_s\in\{0,1\}^n$, because
$\RV{r}_s\in[0,1)^n$. Consequently,
\[
 \begin{aligned}
 H(\RV{d}_s)
 &\le H(\lfloor\RV{n}\rfloor,\RV{b}_s)\\
 &\le H(\lfloor\RV{n}\rfloor)+n\log2
 =:C_{n,\M{\Sigma}}<\infty.
 \end{aligned}
\]
The constant is finite since the fixed Gaussian noise has
a finite second moment. Since $\widehat{\RV{q}}_s$ is a
function of $\RV{y}_s$,
\[
 \begin{aligned}
 I(\RV{x};\RV{y}_s)
 &\ge I(\RV{q}_s;\RV{y}_s)\\
 &=H(\RV{q}_s)-H(\RV{q}_s\mid\RV{y}_s)\\
 &=H(\RV{q}_s)-H(\RV{d}_s\mid\RV{y}_s)\\
 &\ge H(\lfloor s\RV{x}\rfloor)-C_{n,\M{\Sigma}}.
 \end{aligned}
\]
For $\M{\Sigma}=\M{I}_n$, the constant depends only on $n$. Thus the proof is completed.
\end{IEEEproof}

\subsection{\texorpdfstring{Proof of \cref{thm:one-path-entropy}}{Proof of Proposition \getrefnumber{thm:one-path-entropy}}}
\label{app:proof-thm-one-path-entropy}
\begin{IEEEproof}
All constants below are uniform over $1\le\sigma_1(\M{X}_\rho)\le C_1$. Write $\sigma_2=\sigma_2(\M{X}_\rho)$ and
\[
 \RV{z}_\rho=\rv{\alpha}\V{v}_{\M{X}_\rho}(\rv{u}),\qquad
 \RV{y}_\rho=\sqrt\rho\,\RV{z}_\rho+\RV{n}.
\]
First establish the complex-gain contribution for the entire shell. Note that there is a column $\V{x}_{\rho,t_\rho}$ of $\M{X}_\rho$ with norm no less than the column-average, namely $T^{-1/2}$. Thus $p_\rho(u)=\V{a}_t(u)^H\V{x}_{\rho,t_\rho}$ has frequencies in $\Set{\Lambda}_t$ and coefficient norm $\|\V{x}_{\rho,t_\rho}\|_2/\sqrt{N_t}$, bounded above and away from zero. \Cref{lem:exp-uniform} with $L=N_t$ and the directional-density upper
bound give, uniformly for $0<t<1$,
\[
 \Prob{|p_\rho(\rv{u})|<t}\le Ct^{\frac{1}{N_t-1}}.
\]
The constants are uniform over the entire coefficient family by \cref{lem:exp-uniform} and depend on the fixed positions and interval. In particular,
\[
 \E{\log^-|p_\rho(\rv{u})|}
 \!=\!\int_0^\infty\Prob{\log^-|p_\rho(\rv{u})|>t}{\rm d}t=\int_0^\infty\Prob{|p_\rho(\rv{u})|\!<\!e^{-t}}{\rm d}t
 \!\le\! C(N_t\!-\!1).
\]
The positive logarithmic part is bounded by the uniform coefficient upper bound. Given $\rv{u}=u$, the output is Gaussian with covariance $\M{\Sigma}_n+\rho\sigma_\alpha^2
 \V{v}_{\M{X}_\rho}(u)^H\V{v}_{\M{X}_\rho}(u)$. The determinant lemma~\cite[Sec.~0.8.5]{HornJohnson2013}, $\M{\Sigma}_n\succ0$, and $\|\V{v}_{\M{X}_\rho}(u)\|_2\ge|p_\rho(u)|$ imply
\begin{align}
 h(\RV{y}_\rho)
 &\ge h(\RV{y}_\rho\mid\rv{u})\nonumber\\
 &\ge h(\RV{n})+
       \E{\log(1+c\rho|p_\rho(\rv{u})|^2)}\nonumber\\
 &\ge h(\RV{n})+\log\rho+\log c
-2\mathbb E[\log^-|p_\rho(u)|]\nonumber \\
 &\ge\log\rho-C.
 \label{eq:one-path-gain-baseline}
\end{align}
If $\sigma_2=0$, this already proves \eqref{eq:one-path-entropy} with a uniform $O(1)$ loss. If $0<\rho\sigma_2^2\le1$, its corresponding term $\frac{1}{2}\log(1+\rho\sigma_2^2)$ is at most $\frac12\log2$, thus the same bound applies.

It remains to consider $\rho\sigma_2^2>1$. Set
\[
 \varepsilon=\rho^{-1/2},~~
 \mathcal Q_\varepsilon(\V{z})
 =\left\lfloor\frac{\realify(\V{z})}{\varepsilon}\right\rfloor,
 ~~ \RV{q}_\varepsilon=\mathcal Q_\varepsilon(\RV{z}_\rho).
\]
The shell bound and Gaussian gain give $\E{\|\RV{z}_\rho\|_2^2}<\infty$, hence $H(\lfloor\realify(\RV{z}_\rho)\rfloor)<\infty$. The fixed-covariance version of \cref{lem:quantization-lower} gives
\begin{equation}
 I(\RV{z}_\rho;\RV{y}_\rho)
 \ge H(\RV{q}_\varepsilon)-C.
 \label{eq:quantization-lower-onepath}
\end{equation}
Normalize $\widehat{\V{v}}_\rho=\V{v}_{\M{X}_\rho}/\|\V{v}_{\M{X}_\rho}\|_{\rm coef}$. By the shell bounds $1 \leq \sigma_1(\M{X}_{\rho}) \leq  C_1$ and \cref{thm:wronskian},
\[
 c\le \|\V{v}_{\M{X}_\rho}\|_{\rm coef}\le C,\qquad
 \|\widehat{\V{v}}_\rho\wedge\dot{\widehat{\V{v}}}_\rho\|_{\rm coef}
 \ge c_*\sigma_2.
\]
Choose such a constant $c_*>0$ and apply \cref{lem:preimage-volume} with $\Set{\Lambda}=\Set{\Lambda}_t$, $\Set{I}=\Set{U}_t$, $\delta=c_*\sigma_2>0$, and $\tau_\rho=e^{-\sqrt{\log\rho}}$. Let $\Set{G}_\rho^{\rm ang}$ be the corresponding good set, and define 
\[
 b_\rho(\alpha,u)
 =\mathbb I\{u\in \Set{G}_\rho^{\rm ang},\ |\alpha|\ge\tau_\rho\},
 \qquad \rv{b}_\rho=b_\rho(\rv{\alpha},\rv{u}).
\]
The angular estimate and the Gaussian small-ball bound yield
\begin{align}
 \Prob{\rv{b}_\rho=0}
 &\le\Prob{u\notin\Set{G}_\rho^{\rm ang}}+\Prob{|\alpha|<\tau_\rho} \nonumber \\
 &\le Ce^{-\kappa\sqrt{\log \rho}}+\frac{1}{\sigma_{\alpha}^2}e^{-2\sqrt{\log \rho}}\nonumber \\
 &\le\overline q(\rho):=C_qe^{-b\sqrt{\log\rho}},
 \label{eq:good-event-probability}
\end{align}
where $b=\min\{\kappa,2\}$ and $C_q=C+\sigma_{\alpha}^{-2}$. Although $\Prob{\rv{b}_\rho=0}$ may depend on $\M{X}_\rho$, its upper bound $\overline{q}(\rho)$ does not. For all sufficiently large $\rho$, $\Prob{\rv{b}_\rho=0}\le1/2$.

For every sample $(\alpha,u)$ satisfying $b_\rho(\alpha,u)=1$, define $\V{q}(\alpha,u)=\mathcal Q_\varepsilon (\alpha\V{v}_{\M{X}_\rho}(u))$. A lattice cell has diameter at most $\varepsilon\sqrt{2T}$. For the normalized curve, let $\beta=\alpha\|\V{v}_{\M{X}_\rho}\|_{\rm coef}$ and $\beta'=\alpha'\|\V{v}_{\M{X}_\rho}\|_{\rm coef}$. The cell event satisfies
\begin{align*}
 &\{(\alpha',u'):
   \mathcal Q_\varepsilon(\alpha'\V{v}_{\M{X}_\rho}(u'))
       =\V{q}(\alpha,u),\ b_\rho(\alpha',u')=1\}\subseteq
 \{(\alpha',u'):(\beta',u')
   \in\Set{P}_{\varepsilon\sqrt{2T}}(\beta,u)\}.
\end{align*}
Upon defining the joint mass function $p_{\RV{q}_\varepsilon,\rv{b}_\rho}(\V{q},1) =\Prob{\RV{q}_\varepsilon=\V{q},\rv{b}_\rho=1}$, Lemma~\ref{lem:preimage-volume} gives the following pointwise statement for every sample in the resolution cell:
\begin{equation}
  p_{\RV{q}_\varepsilon,\rv{b}_\rho}(\V{q}(\alpha,u),1)\le C\tau_\rho^{-4}\varepsilon^2
 \min\left\{1,\frac{\varepsilon}{|\alpha|\sigma_2}\right\}.
 \label{eq:cell-mass-onepath}
\end{equation}
The conditional probability is given by
\[
 p_{\RV{q}_\varepsilon\mid\rv{b}_\rho=1}(\V{q})
 =\frac{p_{\RV{q}_\varepsilon,\rv{b}_\rho}(\V{q},1)}{\Prob{\rv{b}_\rho=1}}.
\]
Substitute the random sample in \eqref{eq:cell-mass-onepath}, take negative logarithms, and average conditionally obtain
\begin{align*}
 H(\RV{q}_\varepsilon\mid\rv{b}_\rho=1)
 &\ge\log\rho-4\sqrt{\log\rho}+\log\Prob{\rv{b}_\rho=1}-C+\E{\log^+\frac{|\rv{\alpha}|\sigma_2}{\varepsilon}
 \,\middle|\,\rv{b}_\rho=1},
\end{align*}
where we have used $\varepsilon=\rho^{-1/2}$. Since $\log^+(ab)\ge\log^+b-|\log a|$ and
\[
 \E{|\log|\rv{\alpha}||\mid\rv{b}_\rho=1}
 \le\frac{\E{|\log|\rv{\alpha}||}}{\Prob{\rv{b}_\rho=1}}\le C,
\]
the expectation $\E{\log^+\frac{|\rv{\alpha}|\sigma_2}{\varepsilon} \,\middle|\,\rv{b}_\rho=1}$ is at least $\log^+(\sigma_2/\varepsilon)-C$. Discrete conditional entropy is nonnegative, therefore
\[
 H(\RV{q}_\varepsilon)
 \ge\Prob{\rv{b}_\rho=1}H(\RV{q}_\varepsilon\mid\rv{b}_\rho=1).
\]
The shell bound $1\le\sigma_1(\M{X}_\rho)\le C_1$ implies $\log\rho+\log^+(\sqrt\rho\sigma_2)\le\frac32\log\rho+C$. Also $\log\Prob{\rv{b}_\rho=1}\ge-2\Prob{\rv{b}_\rho=0}$. Thus
\begin{equation}
 H(\RV{q}_\varepsilon)
 \ge \log\rho+\log^+(\sqrt\rho\sigma_2)-4\sqrt{\log\rho}
       -C\overline q(\rho)\log\rho-C.
 \label{eq:quantized-entropy-onepath}
\end{equation}
Using $0\le\log(1+x)-\log^+x\le\log2$ for $x>0$ and $h(\RV{y}_\rho)=h(\RV{n})+I(\RV{z}_\rho;\RV{y}_\rho)$ gives
\[
 h(\RV{y}_\rho)
 \ge\log\rho+\tfrac12\log(1+\rho\sigma_2^2)
      -\mathcal R_{\rm sp}(\rho),
\]
where, with constants $B_1$, $B_2$ and $B_3$, we have
\[
 \mathcal R_{\rm sp}(\rho)
 =B_1\sqrt{\log\rho}
 +B_2(\log\rho)e^{-b\sqrt{\log\rho}}+B_3.
\]
This single deterministic remainder obeys $\mathcal R_{\rm sp}(\rho)/\log\rho\to0$ and is independent of the entire matrix sequence. This proves \eqref{eq:one-path-entropy}.
\end{IEEEproof}

\section{Rank-1-Tube Output Entropy Bounds}
In what follows, we denote $m=2KT$, $d_1=2(K+T-1)$.
\subsection{\texorpdfstring{Proof of \cref{lem:standard-rankone-cover}}{Proof of Lemma \getrefnumber{lem:standard-rankone-cover}}}
\label{app:proof-lem-standard-rankone-cover}
\begin{IEEEproof}
Let us first denote the set of rank-1 matrices with unit norm as
\begin{equation}
\Set{L} = \{\M{Q}\in\mathbb{C}^{K\times T}:~\rank(\M{Q})=1,~\|\M{Q}\|_F=1\},
\end{equation}
which is a compact smooth manifold. We may then cover the manifold with finitely many local coordinate charts. On each of the charts, we choose a smooth parameterization
\begin{equation}
\M{Q}_\ell(\V{t})=\V{a}_\ell(\V{t})\V{b}_\ell(\V{t})^H,~~\|\V{a}_\ell\|=\|\V{b}_\ell\|=1.
\end{equation}
The dimensionality of $\Set{L}$ can then be computed as
\begin{equation}
\dim_{\mathbb{R}}\Set{L}=(2K-1)+(2T-1)-1=d_1-1,
\end{equation}
where $2K-1$ is the real dimensionality of $\V{a}_\ell$ while $2T-1$ is the real dimensionality of $\V{b}_\ell$. The additional rank loss of 1 originates from the common phase ambiguity $(e^{\jmathu \theta}\V{a}_\ell)(e^{\jmathu \theta}\V{b}_\ell)^H=\V{a}_\ell\V{b}_\ell^H$. Therefore, $\V{t}$ has real dimensionality of $d_1-1$. We resolve the common phase ambiguity of $\V{a}_\ell$ and $\V{b}_\ell$ by letting one of the nonzero elements in $\V{b}_\ell$, say $b_j$, be real and positive. The specific choice of this element can be distinct for different local coordinate charts. On each coordinate neighborhood, choose smooth orthonormal complements
\begin{equation}
\M{A}_{\ell,\perp}\in\mathbb{C}^{K\times (K-1)},~~\M{B}_{\ell,\perp}\in\mathbb{C}^{T\times(T-1)},
\end{equation}
on each of the charts, which are semi-unitary matrices and satisfy $\M{A}_{\ell,\perp}^H\V{a}_\ell=\V{0}$ and $\M{B}_{\ell,\perp}^H\V{b}_\ell=\V{0}$ respectively, and restrict the parameterizations to bounded coordinate boxes whose closures are contained in their coordinate domains. Within these bounded boxes, the mappings $\V{a}_\ell(\V{t})$ and $\V{b}_\ell(\V{t})$, as well as their orthogonal complements, have uniform Lipschitz constants across the finitely many charts.

Next, let us define the tubular parameterization as
\begin{equation}\label{tubular_coordinate}
\M{\Phi}_\ell(s,\V{t},\M{W}) = s\V{a}_\ell(\V{t})\V{b}_\ell(\V{t})^H+\M{A}_{\ell,\perp}(\V{t})\M{W}\M{B}_{\ell,\perp}^H(\V{t}),
\end{equation}
where $0\leq s\leq 2R$, and $\M{W}\in\mathbb{C}^{(K-1)\times(T-1)}$ satisfies $\|\M{W}\|_F\leq r$. Next we show that $\M{\Phi}_\ell(s,\V{t},\M{W})$ covers the entire $\Set{T}_1(R,r)$. For a nonzero $\M{Z}\in\Set{T}_1(R,r)$, consider its best approximation in the set of matrices with rank at most 1, given by $\M{Z}_1=s\V{a}\V{b}^H$. We then have $\|\M{Z}-\M{Z}_1\|_F\leq r$, and $s\leq \|\M{Z}\|_F\leq R+r\leq 2R$. Denote the residual as
\begin{equation}
\M{E}=\M{Z}-\M{Z}_1=\M{Z}-s\V{a}\V{b}^H,
\end{equation}
for which we have $\M{E}\V{b}=\V{0}$ and $\V{a}^H\V{E}=\V{0}^H$. Thus all columns of $\M{E}$ are within the orthogonal complement of $\V{a}$, implying that $\M{E}=\M{A}_\perp\M{W}\M{B}_\perp^H$ and $\M{W}=\M{A}_\perp^H\M{E}\M{B}_\perp$, and $\|\M{E}\|_F=\|\M{W}\|_F\leq r$. The zero matrix is represents by $s=0$, $\M{W}=\M{0}$. Therefore, the tubular coordinate \eqref{tubular_coordinate} indeed covers the entire $\Set{T}_1(R,r)$. Now, for any pair of points within the tube, we have
\begin{align}
\|\M{\Phi}_\ell(s,\V{t},\M{W})-\M{\Phi}_\ell(s^\prime,\V{t}^\prime,\M{W}^\prime)\|_F& =\|s\M{Q}_\ell-s^\prime\M{Q}_\ell^\prime+\M{A}_{\ell,\perp}\M{W}\M{B}_{\ell,\perp}^H-\M{A}_{\ell,\perp}^\prime\M{W}^\prime(\M{B}_{\ell,\perp}^\prime)^H\|_F\nonumber \\
&\leq |s-s^\prime|\|\M{Q}_\ell\|_F+s^\prime \|\M{Q}_\ell-\M{Q}_\ell^\prime\|_F+\|\M{A}_{\ell,\perp}\|_2\|\M{W}-\M{W}^\prime\|_F\|\M{B}_{\ell,\perp}\|_2\nonumber \\
&~~+ \|\M{A}_{\ell,\perp}-\M{A}_{\ell,\perp}^\prime\|_2\|\M{W}^\prime\|_F\|\M{B}_{\ell,\perp}\|_2+\|\M{A}_{\ell,\perp}^\prime\|_2\|\M{W}^\prime\|_F\|\M{B}_{\ell,\perp}-\M{B}_{\ell,\perp}^\prime\|_2\nonumber \\
&\leq |s-s^\prime|+(2RL_Q+rL_A+rL_B)\|\V{t}-\V{t}^\prime\|+\|\M{W}-\M{W}^\prime\|_F\nonumber \\
&\leq C(|s-s^\prime|+R\|\V{t}-\V{t}^\prime\|+\|\M{W}-\M{W}^\prime\|_F),
\end{align}
where $L_Q$, $L_A$ and $L_B$ are the Lipschitz constants of $\M{Q}_{\ell}$, $\M{A}_{\ell,\perp}$ and $\M{B}_{\ell,\perp}$, respectively. By choosing grid spacings scaling as $\Theta(\epsilon)$, $\Theta(\epsilon/R)$ and $\Theta(\epsilon)$ for $s$, $\V{t}$ and $\M{W}$, respectively, we see that their $\epsilon$-cover would require $O(R/\epsilon)$, $O[(R/\epsilon)^{d_1-1}]$, and $O[(1+r/\epsilon)^{m-d_1}]$ points, respectively. Multiplying them and combining the finitely many local coordinate charts, we obtain
\begin{equation}
N_\epsilon(\Set{T}_1(R,r))\leq C\left(\frac{cR}{\epsilon}\right)^{d_1}\left(1+\frac{cr}{\epsilon}\right)^{m-d_1}.
\end{equation}
This proves the claimed covering bound, with constants depending only on $(K,T)$.
\end{IEEEproof}

\subsection{\texorpdfstring{Proof of \cref{lem:tube-entropy}}{Proof of Lemma \getrefnumber{lem:tube-entropy}}}
\label{app:proof-lem-tube-entropy}
\begin{IEEEproof}
Let $\RM{X}_1$ be a measurable best approximation to $\RM{X}$ with rank at most 1.\footnote{Choose a fixed measurable singular vector convention if the non-zero singular value has multiplicity larger than 1.} Then
\[
 \|\RM{X}-\RM{X}_1\|_F\le\delta,\qquad
 \|\RM{X}_1\|_F\le\|\RM{X}\|_F\le C_1.
\]
With $\RM{G}_{\rm path}$ from \eqref{random_pathwise}, write $\RM{Q}=\RM{G}_{\rm path}\RM{X}$ and $\RM{Q}_1=\RM{G}_{\rm path}\RM{X}_1$. Since the steering vectors have unit norms, $\|\RM{G}_{\rm path}\|_2\le\|\RM{G}_{\rm path}\|_F=\|\RV{\alpha}\|_2$, and therefore
\[
 \rank\{\RM{G}_{\rm path}\RM{X}_1\}\le1,~
 \|\RM{G}_{\rm path}\RM{X}_1\|_F\le C_1\|\RV{\alpha}\|_2,~
 \|\RM{G}_{\rm path}\RM{X}-\RM{G}_{\rm path}\RM{X}_1\|_F\le\delta\|\RV{\alpha}\|_2.
\]

\emph{Gain tails.}
Fix $M>0$, and only then choose a fixed $c>0$ sufficiently large. Define $\Set{E}_\rho=\{\|\RV{\alpha}\|_2^2\le c\log\rho\}$. Since $\rv{\gamma}=\|\RV{\alpha}\|_2^2/\sigma_\alpha^2\sim\operatorname{Gamma}(K,1)$, for $t\ge0$, we have
\begin{align*}
 \Prob{\rv{\gamma}>t}&=e^{-t}\sum_{j=0}^{K-1}\frac{t^j}{j!},\\
 \E{\rv{\gamma}\mid \rv{\gamma}>t}
 &=K+\frac{t^K}{(K-1)!\sum_{j=0}^{K-1}t^j/j!}
 \le K+t.
\end{align*}
Taking $t=c\log\rho/\sigma_\alpha^2$ with $c/\sigma_\alpha^2>M+1$ therefore gives
\begin{equation}\label{moment_gain}
 \Prob{\Set{E}_\rho^c}\le C_M\rho^{-M},~~
 \E{\|\RV{\alpha}\|_2^2\mid \Set{E}_\rho^c}
 \le K\sigma_\alpha^2+c\log\rho.
\end{equation}

\emph{Support and covering index on the good event.} 
Choose fixed $C_r=\sqrt c$ and $C_R\ge\max\{1,\sqrt c C_1,\sqrt c\delta_{\max}\}$. On $\Set{E}_\rho$, we have
\begin{equation}\label{constants_cover_event}
 \RM{G}_{\rm path}\RM{X}\in\Set{T}_1(R_\rho,r_\rho),~~ R_\rho=C_R\sqrt{\log\rho},~~r_\rho=C_r\delta\sqrt{\log\rho}\le R_\rho.
\end{equation}
Take $\epsilon=\rho^{-1/2}$ and a finite $\epsilon$-cover with matrix centers $\M{C}_1,\ldots,\M{C}_{N_\epsilon}$. Let $\rv{J}$ be the least index of a center within $\epsilon$ of $\RM{G}_{\rm path}\RM{X}$ on $\Set{E}_\rho$. The noise $\RM{N}_0$ in $\RM{Z}_\rho=\sqrt\rho\,\RM{G}_{\rm path}\RM{X}+\RM{N}_0$ is independent of $(\RM{G}_{\rm path}\RM{X},\mathbb I_{\Set{E}_\rho},\rv{J})$. The entropy chain rule gives
\begin{align*}
 h(\RM{Z}_\rho\mid \Set{E}_\rho)
 &\le H(\rv{J}\mid \Set{E}_\rho)
      +h(\RM{Z}_\rho\mid\rv{J},\Set{E}_\rho),\\
 H(\rv{J}\mid \Set{E}_\rho)&\le\log N_\epsilon.
\end{align*}
For every index $j$ with positive probability, translation by $\sqrt\rho\,\M{C}_j$ leaves entropy unchanged. Since we are considering a $\epsilon$-cover, the scaled residual $\sqrt{\rho}(\RM{G}_{\rm path}\RM{X}-\M{C}_j)$ has norm at most 1. Thus
\[
 \E{\|\RM{Z}_\rho-\sqrt\rho\,\M{C}_j\|_F^2 \mid\rv{J}=j,\Set{E}_\rho} \le1+\E{\|\RM{N}_0\|_F^2}.
\]
Gaussian maximum entropy after realification implies $h(\RM{Z}_\rho\mid\rv{J}=j,\Set{E}_\rho)\le C_{K,T}$ uniformly in the index, the input distribution, and $\delta$. Hence
\begin{equation}\label{entropy_ncover_bound}
 h(\RM{Z}_\rho\mid \Set{E}_\rho)\le\log N_\epsilon+C.
\end{equation}

Lemma~\ref{lem:standard-rankone-cover} gives
\begin{equation}\label{ncover_upper}
 \log N_\epsilon
 \le C+d_1\log(c'R_\rho/\epsilon)+(m-d_1)\log(1+c'r_\rho/\epsilon).
\end{equation}
Using \eqref{constants_cover_event}, we have
\begin{subequations}\label{separate_terms}
\begin{align}
 \log(c'R_\rho/\epsilon)&=\tfrac12\log\rho+\tfrac12\log \log\rho+O(1),\\
 2\log(1+c'r_\rho/\epsilon)&=\log[(1+c'C_r \delta\sqrt{\rho\log\rho})^2]\leq  \log(1+\rho\delta^2)+O(\log\log \rho).
\end{align}
\end{subequations}
The last equation follows from $(1+y)^2\le2(1+y^2)$ and is valid also at $\rho\delta^2=0$. Substituting \eqref{separate_terms} and \eqref{ncover_upper} into \eqref{entropy_ncover_bound} yields
\begin{equation}\label{entropy_on_erho}
 h(\RM{Z}_\rho\mid \Set{E}_\rho)
 \le\tfrac{d_1}{2}\log\rho
    +\tfrac{m-d_1}{2}\log(1+\rho\delta^2)+C(1+\log\log\rho).
\end{equation}

\emph{The weighted bad-event entropy.}
The bound \eqref{moment_gain} on the moment of gain, the peak power bound on $\RM{X}$, and the fixed noise second moment imply
\[
 \E{\|\RM{Z}_\rho\|_F^2\mid \Set{E}_\rho^c}
 \le C[1+\rho(1+\log\rho)].
\]
Gaussian maximum entropy consequently gives
\begin{equation}\label{coh_bad_event}
 h(\RM{Z}_\rho\mid \Set{E}_\rho^c)
 \le KT\log\bigl(C[1+\rho(1+\log\rho)]\bigr).
\end{equation}
For a sufficiently large $C_1$, \eqref{coh_bad_event} simplifies into
\begin{equation}
h(\RM{Z}_\rho\mid \Set{E}_\rho^c)\leq C_1(1+\log\rho).
\end{equation}
Using again \eqref{moment_gain}, we have
\begin{subequations}
 \begin{align}
 \Prob{\Set{E}_\rho^c}h(\RM{Z}_\rho\mid \Set{E}_\rho^c)
 &\le C_M\rho^{-M}(1+\log\rho),\\
 H_2(\Prob{\Set{E}_\rho^c})&\le C_M\rho^{-M}(1+\log\rho).
 \end{align}
 \label{eq:alpha-tail-tube}
\end{subequations}
Finally, since conditioning on an event meaning that the event occurs, we have
\begin{equation}\label{entropy_total_uncond}
 h(\RM{Z}_\rho)
 \le H_2(\Prob{\Set{E}_\rho^c})+(1-\Prob{\Set{E}_\rho^c})h(\RM{Z}_\rho\mid \Set{E}_\rho)+\Prob{\Set{E}_\rho^c} h(\RM{Z}_\rho\mid \Set{E}_\rho^c).
\end{equation}
Now, insert the bounds \eqref{entropy_on_erho} and \eqref{eq:alpha-tail-tube} into \eqref{entropy_total_uncond}, and fix a value of $M$. We see that
\begin{equation}
C_M\rho^{-M}(1+\log \rho)\leq C_M\Big(1+\frac{1}{Me}\Big):=C_2,
\end{equation}
where $C_2$ is a constant. By adjusting the constants, we obtain \eqref{eq:tube-entropy} with $\mathcal R_{\rm tube}(\rho)=C(1+\log\log\rho)$ for $\rho\ge e^2$. This remainder is uniform over all input distributions and $0\le\delta\le\delta_{\max}$, and thus we are done.
\end{IEEEproof}

\section{Multiscale Converse and Final Assembly}
The three proofs below progress through successive levels of uniformity. The multiscale theorem applies the two component bounds directly inside finitely many two-dimensional cells and recombines them uniformly over all input distributions with peak power constraints. The lifting lemma removes the peak power constraint, and the final subsection assembles the converse with the achievability results.

For the $K\ge2$ arguments in this appendix, set
\[
 r_\star:=\min\{N_t,T\},
 \qquad
 c_{\rm rk}:=\max\{1,\sqrt{r_\star-1}\}.
\]
By \cref{lem:sigma2-rankone-distance}, $c_{\rm rk}$ is a fixed
dimension-dependent factor converting a $\sigma_2$ bound into a rank-1 tube
radius.
\subsection{\texorpdfstring{Proof of \cref{thm:peak}}{Proof of Theorem \getrefnumber{thm:peak}}}
\label{app:proof-thm-peak}
\begin{IEEEproof}
Fix $0<\Delta<1$ and define $q(\M{X}):=\log\!\left(1+\rho\sigma_2(\M{X})^2\right)$. Assign the label $\rv{\kappa}_\Delta=\dagger$ to the low-SNR event $\rho\sigma_1(\RM{X})^2\le A^2\rho^\Delta$, which includes $\RM{X}=\M{0}$. On its complement, define $s(\M{X}):=\log(\rho\sigma_1(\M{X})^2)$. The peak power constraint gives
\begin{align*}
 s(\M{X})&\in\bigl(\log(A^2\rho^\Delta),\log(A^2\rho)\bigr],~~q(\M{X})\in[0,\log(1+A^2\rho)].
\end{align*}
Partition the $s$-interval into left-open, right-closed intervals of length at most $\Delta\log\rho$. Partition the $q$-interval into intervals of the same maximum length, taking the first as $[0,\Delta\log\rho]$ and the remaining intervals left-open and right-closed. The label $\rv{\kappa}_\Delta$ records the resulting rectangular cell
outside the low-SNR event. Throughout the proof, $\Delta$ is fixed while $\rho\to\infty$.

On the low-SNR event $\rv{\kappa}_\Delta=\dagger$, we have
$\rho\|\RM{X}\|_F^2\le r_\star A^2\rho^\Delta$.
Reveal the pathwise channel $\RM{G}_{\rm path}$ defined in
\eqref{random_pathwise}. Since
$\E{\|\RM{G}_{\rm path}\|_2^2}
\le\E{\|\RM{G}_{\rm path}\|_F^2}=K\sigma_\alpha^2<\infty$,
the coherent channel capacity bound gives, for some fixed $c>0$,
\begin{align}
I(\RM{X};\RM{Z}_\rho\!\mid\! \rv{\kappa}_\Delta\!=\!\dagger)
&\!\le\! I(\RM{X};\RM{Z}_\rho\mid \RM{G}_{\rm path},\rv{\kappa}_\Delta=\dagger)\nonumber\\
&\le KT\log(1+c r_\star A^2\rho^\Delta)\nonumber\\
&=O_{A,K,T}(1+\Delta\log\rho).
\end{align}
Thus the contribution of this low-SNR cell can be neglected as $\Delta\rightarrow 0$. In addition, for fixed $(A,\Delta)$ and all sufficiently large $\rho$, the number of labels is $O_{A,\Delta}(1)$, and hence $H(\rv{\kappa}_\Delta)=O_{A,\Delta}(1)=o(\log\rho)$, so the cell label contributes only negligible entropy.

Let us now consider the ``ordinary'' grid. Fix a positive-probability cell $\rv{\kappa}_{\Delta}=\Set{C}$ and denote its endpoints by
\[
 \underline{s}<s(\RM{X})\le \overline{s},
 \qquad
 \underline{q}\le q(\RM{X})\le \overline{q}.
\]
They satisfy
\[
 0\le \overline{s}-\underline{s}\le \Delta\log \rho,
 \qquad
 0\le \overline{q}-\underline{q}\le \Delta\log \rho.
\]

For the output entropy bound, define $\RM{X}' :=\sqrt{\rho}\,e^{-\overline{s}/2}\RM{X}$. The same observation is obtained by operating the pathwise channel at SNR $e^{\overline{s}}$ with input $\RM{X}'$. On $\Set{C}$, we have
\[
 \sigma_1(\RM{X}')\le1,~~
 \|\RM{X}'\|_F\le\sqrt{r_\star},~~
 e^{\overline{s}}\sigma_2(\RM{X}')^2
 \le e^{\overline{q}}-1.
\]
Thus \cref{lem:sigma2-rankone-distance} places the conditional input support
in a rank-1 tube of radius
\[
 \delta_C
 :=c_{\rm rk}\min\!\left\{
 1,\,
 e^{-\overline{s}/2}\sqrt{e^{\overline{q}}-1}
 \right\}
 \le c_{\rm rk}.
\]
The truncation by $1$ ensures the fixed radius bound required by \cref{lem:tube-entropy}, while
\[
 \log(1+e^{\overline{s}}\delta_C^2)
 \le \overline{q}+2\log c_{\rm rk}.
\]
Applying \cref{lem:tube-entropy} at SNR $e^{\overline{s}}$ yields
\begin{equation}\label{entropy_c}
 h(\RM{Z}_\rho\mid\rv{\kappa}_{\Delta}=\Set{C})
 \le (K+T-1)\overline{s} +(KT-K-T+1)\overline{q}
 +o_{A,\Delta}(\log\rho).
\end{equation}

For the conditional output entropy, normalize each realization $\M{X}$ by $\sigma_1(\M{X})$. The actual SNR is $e^{s(\M{X})}$, and the observation with ${\rm SNR}=e^{\underline{s}}$ is a stochastic degradation of the actual one. For any fixed signal distribution and independent Gaussian noise of fixed covariance, we have
\[
 h(\sqrt t\,\RV{v}+\RV{n})
 =h(\RV{n})+I(\RV{v};\sqrt t\,\RV{v}+\RV{n}).
\]
Data processing therefore makes this entropy nondecreasing in $t$.
Thus the lower-SNR entropy is a valid lower bound on the actual
conditional-output entropy. Its second-mode term obeys
\begin{align*}
 \log\!\left(
 1+e^{\underline{s}}
 \sigma_2\!\left(\frac{\M{X}}{\sigma_1(\M{X})}\right)^2
 \right)
 &=
 \log\!\left(
 1+e^{\underline{s}-s(\M{X})}(e^{q(\M{X})}-1)
 \right)\\
 &\ge q(\M{X})-\bigl(s(\M{X})-\underline{s}\bigr)
 \ge \underline{q} -\Delta\log \rho.
\end{align*}
The first inequality uses
$1+\theta(e^q-1)\ge \theta e^q$ for $0<\theta\le1$. Applying
\cref{thm:one-path-entropy} at SNR $e^{\underline{s}}$, summing the $K$ conditionally
independent rows, and averaging within the cell gives
\begin{equation}\label{entropy_condition_c}
h(\RM{Z}_\rho\!\mid\!\RM{X},\rv{\kappa}_{\Delta}\!=\!\Set{C})\ge K\underline{s}\!+\!\frac K2(\underline{q}-\Delta\log\rho)
 -o_{A,\Delta}(\log\rho).
\end{equation}
Both effective SNRs $e^{\underline{s}}$ and $e^{\overline{s}}$ lie in the interval
$\Set{I}_\rho:=[A^2\rho^\Delta,A^2\rho]$. Hence the sum of the remainder terms in the two entropy bounds is at most
\[
\mathcal R_{A,\Delta}(\rho)
:=
K\sup_{t\in\Set{I}_\rho}\mathcal R_{\rm sp}(t)
+\sup_{t\in\Set{I}_\rho}\mathcal R_{\rm tube}(t)
=o_{A,\Delta}(\log\rho).
\]
The last relation follows from the explicit remainder bounds, since $\Delta>0$ is fixed. The normalization fixes the input shell bounds, and $\delta_C\le c_{\rm rk}$ fixes the maximum tube radius. Thus this remainder is independent of the cell and the input distribution.

Now, subtracting the two entropy bounds \eqref{entropy_c} and \eqref{entropy_condition_c}, and using the cell widths gives
\[
 \begin{aligned}
 I(\RM{X};\RM{Z}_\rho\mid\rv{\kappa}_{\Delta}=\Set{C})
 &\le
 (T-1)\overline{s}
 +\bigl[D_\star-(T-1)\bigr]\overline{q}+O_{K,T}(\Delta\log\rho)
 +o_{A,\Delta}(\log\rho).
 \end{aligned}
\]
Here, $KT-K-T+1-\frac K2 =D_\star-(T-1)\ge0$. Because the cell has positive probability, it contains some valid realizations of $\M{X}$, and
\begin{align*}
 \overline{q}
 \le q(\M{X})+ \Delta\log\rho
 &\le\log(1+e^{s(\M{X})})+\Delta\log\rho\\
 &\le \overline{s}+\Delta\log\rho+\log2.
\end{align*}
Together with $\overline{s}\le\log(A^2\rho)$, this yields, uniformly over all active
cells,
\[
 I(\RM{X};\RM{Z}_\rho\!\mid\!\rv{\kappa}_{\Delta}\!=\!\Set{C})
 \!\le\!
 D_\star\log\rho
 \!+\!O_{A,K,T}(1\!+\!\Delta\log\rho)
 \!+\!o_{A,\Delta}(\log\rho).
\]

Finally, informing the receiver of the cell label gives
\begin{align*}
 I(\RM{X};\RM{Z}_\rho)
 &\le H(\rv{\kappa}_\Delta)
 +\Prob{\rv{\kappa}_\Delta=\dagger}
 I(\RM{X};\RM{Z}_\rho\mid\rv{\kappa}_\Delta=\dagger)+
 \sum_{\Set{C}}
 \Prob{\rv{\kappa}_{\Delta}=\Set{C}}\,
 I(\RM{X};\RM{Z}_\rho\mid\rv{\kappa}_{\Delta}=\Set{C})\\
 &\le
 D_\star\log\rho
 +O_{A,K,T}(\Delta\log\rho+1)
 +o_{A,\Delta}(\log\rho).
\end{align*}
More explicitly, we have
\begin{align*}
 J_{\rm pk}^{Z}(\rho;A)
 &\le D_\star\log\rho
 +C_{A,K,T}(1+\Delta\log\rho)+C_{A,\Delta}+\mathcal R_{A,\Delta}(\rho).
\end{align*}
The right-hand side is independent of the input distribution. For every fixed $\Delta$, we have
\[
 \limsup_{\rho\to\infty}
 \frac{J_{\rm pk}^{Z}(\rho;A)}{\log\rho}
 \le D_\star+O_{A,K,T}(\Delta).
\]
Letting $\Delta\downarrow0$ after this limit proves
\eqref{eq:peak-envelope-converse}.
\end{IEEEproof}

\subsection{\texorpdfstring{Proof of \cref{lem:power-tail}}{Proof of Lemma \getrefnumber{lem:power-tail}}}
\label{app:proof-lem-power-tail}
\begin{IEEEproof}
Fix arbitrary $\epsilon>0$. Consider any possibly SNR-dependent input distribution satisfying $\E{\|\RM{X}\|_F^2}\le T$, and define $\Set{E}_\epsilon:=\{\|\RM{X}\|_F^2\le\rho^\epsilon\}$. Markov's inequality gives
\begin{equation}
 \Prob{\Set{E}_\epsilon^c}\le T\rho^{-\epsilon}.
 \label{eq:markov-power-tail}
\end{equation}
Since whether $\Set{E}_\epsilon$ occurs is determined by $\RM{X}$, revealing this event label gives
\begin{align}
 I(\RM{X};\RM{Z}_\rho)
 &\le H_2(\Prob{\Set{E}_\epsilon^c})
 +\Prob{\Set{E}_\epsilon}I(\RM{X};\RM{Z}_\rho\mid \Set{E}_\epsilon)+\Prob{\Set{E}_\epsilon^c} I(\RM{X};\RM{Z}_\rho\mid \Set{E}_\epsilon^c).
 \label{eq:power-event-split}
\end{align}

On the high-power tail, reveal the random pathwise channel $\RM{G}_{\rm path}$ as defined in \eqref{random_pathwise}. Because $\Set{E}_\epsilon^c$ depends only on $\RM{X}$, conditioning preserves the independence between $\RM{X}$ and $\RM{G}_{\rm path}$. When $\Prob{\Set{E}_\epsilon^c}>0$, we have
\[
 \E{\|\RM{X}\|_F^2\mid\Set{E}_\epsilon^c}
 \le\frac{T}{\Prob{\Set{E}_\epsilon^c}}.
\]
This implies that
\begin{align}
\E{\|\RM{Z}_\rho\|_F^2\mid\Set{E}_\epsilon^c,\RM{G}_{\rm path}=\M{G}}&=\rho\E{\|\M{G}\RM{X}\|_F^2|\Set{E}_\epsilon^c}+\E{\|\RM{N}_0\|_F^2}\nonumber \\
&\leq\rho\|\M{G}\|_2^2\E{\|\RM{X}\|_F^2|\Set{E}_\epsilon^c}+T\nonumber \\
&\leq \rho\|\M{G}\|_2^2\frac{T}{\Prob{\Set{E}_\epsilon^c}}+T.
\end{align}
The Gaussian maximum-entropy bound yields
$$
h(\RM{Z}_\rho\mid\Set{E}_\epsilon^c,\RM{G}_{\rm path}=\M{G})\leq KT\log\left(\frac{\pi e}{K}\left(1+\frac{\rho\|\M{G}\|_2^2}{\Prob{\Set{E}_\epsilon^c}}\right)\right).
$$
After revealing $\RM{G}_{\rm path}$, Jensen's inequality gives
\begin{align}
 I(\RM{X};\RM{Z}_\rho\mid\Set{E}_\epsilon^c)
 &\le I(\RM{X};\RM{Z}_\rho\mid\Set{E}_\epsilon^c,\RM{G}_{\rm path})\nonumber\\
 &\le\!KT\log\!\left(1\!+\!\frac{\E{\|\RM{G}_{\rm path}\|_2^2}\rho}{\Prob{\Set{E}_\epsilon^c}}\right).
 \label{eq:coherent-tail-bound}
\end{align}
For all sufficiently large $\rho$, $T\rho^{-\epsilon}\le1$. Since $p\mapsto p\log(1+a/p)$ is increasing on $(0,1]$, combining
\eqref{eq:markov-power-tail} and \eqref{eq:coherent-tail-bound} yields
\begin{align}
 &\Prob{\Set{E}_\epsilon^c} I(\RM{X};\RM{Z}_\rho\mid \Set{E}_\epsilon^c)\le\! KT^2\rho^{-\epsilon}
 \log\!\left(1\!+\!\frac{\E{\|\RM{G}_{\rm path}\|_2^2}}{T}\rho^{1+\epsilon}\right)
 \!=\!o(\log\rho).
 \label{eq:weighted-tail-zero}
\end{align}
The weighted tail term $\Prob{\Set{E}_\epsilon^c}I(\RM{X};\RM{Z}_\rho\mid\Set{E}_\epsilon^c)$ is understood as zero when $\Prob{\Set{E}_\epsilon^c}=0$. The event-label cost satisfies the uniform bound
\begin{equation}
 H_2(\Prob{\Set{E}_\epsilon^c})\le\log2=o(\log\rho).
 \label{eq:event-entropy-zero}
\end{equation}

On $\Set{E}_\epsilon$, choose $\widetilde{\RM{X}}:=\rho^{-\epsilon/2}\RM{X}$, $\gamma:=\rho^{1+\epsilon}$. The conditional input is supported on the unit Frobenius ball. Since $\Set{E}_\epsilon$ depends only on the input, conditioning preserves its
independence from the pathwise channel and noise. Moreover, $\sqrt\rho\,\RM{G}_{\rm path}\RM{X} =\sqrt\gamma\,\RM{G}_{\rm path}\widetilde{\RM{X}}$. Thus the conditional distribution is admissible in the definition of $J_{\rm pk}^{Z}(\gamma;1)$. 
Now, fix $\eta>0$. By \eqref{eq:peak-envelope-converse}, there exists $\gamma_0(\eta)<\infty$ such that
\begin{equation}
 \gamma\ge\gamma_0(\eta)\Longrightarrow
 J_{\rm pk}^{Z}(\gamma;1)
 \le(D_\star+\eta)\log\gamma.
 \label{eq:peak-eventual-bound}
\end{equation}
For all sufficiently large $\rho$, \eqref{eq:peak-eventual-bound} gives
\begin{align}
 I(\RM{X};\RM{Z}_\rho\mid \Set{E}_\epsilon)
 &=I(\widetilde{\RM{X}};\RM{Z}_\gamma\mid \Set{E}_\epsilon)\nonumber\\
 &\le J_{\rm pk}^{Z}(\gamma;1)\nonumber\\
 &\le(D_\star+\eta)(1+\epsilon)\log\rho.
 \label{eq:typical-power-bound}
\end{align}
Here $\RM{Z}_\gamma$ denotes the same pathwise channel operated at SNR $\gamma$ and driven by the conditional input $\widetilde{\RM{X}}$.

Combining \eqref{eq:power-event-split}, \eqref{eq:weighted-tail-zero},
\eqref{eq:event-entropy-zero}, and \eqref{eq:typical-power-bound} yields the input-independent, finite-SNR bound
\begin{align*}
 I(\RM{X};\RM{Z}_\rho)
 &\le(D_\star+\eta)(1+\epsilon)\log\rho
 +\log2+KT^2\rho^{-\epsilon}
 \log\!\left(1+\frac{\E{\|\RM{G}_{\rm path}\|_2^2}}{T}\rho^{1+\epsilon}\right),
\end{align*}
implying that
\[
 \limsup_{\rho\to\infty}
 \frac{J_{\rm av}^{Z}(\rho)}{\log\rho}
 \le(D_\star+\eta)(1+\epsilon).
\]
Finally, letting $\eta\downarrow0$ and then $\epsilon\downarrow0$ yields
\eqref{eq:average-envelope-converse}.
\end{IEEEproof}

\Needspace{6\baselineskip}
\subsection{\texorpdfstring{Proof of \cref{thm:capacity}}{Proof of Theorem \getrefnumber{thm:capacity}}}
\label{app:proof-thm-capacity}
\begin{IEEEproof}
For $K=1$, the fixed-input lower bound in \cref{prop:single-beam}
and the capacity upper bound in \cref{prop:geom-genie} give
\begin{align*}
 1-\frac1T
 &\le\liminf_{\rho\to\infty}\frac{C_{\rm blind}(\rho)}{\log\rho}\\
 &\le\limsup_{\rho\to\infty}\frac{C_{\rm blind}(\rho)}{\log\rho}\\
 &\le1-\frac1T.
\end{align*}
Thus the limit exists and has the asserted value. Now let $K\ge2$.
By \cref{prop:pathwise-genie}, every admissible input distribution satisfies
\[
 I(\RM{X};\RM{Y}_\rho)
 \le
 I(\RM{X};\RM{Z}_\rho).
\]
Taking supremum over the input distributions under the average power constraint and applying
\cref{lem:power-tail} gives
\begin{align*}
 \limsup_{\rho\to\infty}
 \frac{C_{\rm blind}(\rho)}{\log\rho}
 &\le\frac1T\limsup_{\rho\to\infty}
 \frac{J_{\rm av}^{Z}(\rho)}{\log\rho}\\
 &\le\frac{D_\star}{T}
 =K\left(1-\frac{3}{2T}\right).
\end{align*}

Conversely, \cref{prop:full-ach} supplies the admissible Gaussian input in
\eqref{eq:gaussian-ach-input} for which
\[
 \lim_{\rho\to\infty}
 \frac{I(\RM{X};\RM{Y}_\rho)}{T\log\rho}
 =K\left(1-\frac{3}{2T}\right).
\]
Since capacity is at least the mutual information of this fixed admissible
input divided by $T$,
\[
 \liminf_{\rho\to\infty}
 \frac{C_{\rm blind}(\rho)}{\log\rho}
 \ge K\left(1-\frac{3}{2T}\right).
\]
The liminf lower bound and limsup upper bound coincide. Hence the
capacity pre-log limit exists also for $K\ge2$, proving
\eqref{eq:capacity} in both cases.
\end{IEEEproof}

\section{Proofs for Other Array Structures}
\label{app:other-array-proofs}
\subsection{\texorpdfstring{Proof of \cref{prop:fixed-precoder-conditional}}{Proof of Proposition \getrefnumber{prop:fixed-precoder-conditional}}}
\label{app:proof-fixed-precoder-conditional}
\begin{IEEEproof}
We work on the probability-1 open state set in (G2), and fix a realization $\M{S}$ with full row rank. Write
\[
 \M{F}=\diag(\V{\alpha})\M{A}_t^H\M{B},\qquad
 \M{C}=\M{F}\M{S},\qquad
 \M{Z}=\M{A}_r\M{C}.
\]
By (G2), $\rank\M{F}=K$. Since $\M{S}\M{S}^\dagger=\M{I}_L$, we have $\M{C}\M{S}^\dagger=\M{F}$ and hence $\rank\M{C}=K$.

First compute the joint rank. For any prescribed $\delta\M{C}\in\C^{K\times T}$, the data
variation $\delta\M{S}=\M{F}^\dagger\delta\M{C}$ gives
\[
 \M{F}\delta\M{S}
 =\M{F}\M{F}^\dagger\delta\M{C}
 =\delta\M{C}.
\]
Thus data variations generate every output variation $\M{A}_r\delta\M{C}$, a real space of dimension $2KT$. The complete output differential is
\[
 \delta\M{Z}
 =\delta\M{A}_r\M{C}
   +\M{A}_r(\delta\M{F}\M{S}+\M{F}\delta\M{S}).
\]
In particular, all gain and \ac{aod} variations are already contained in the space generated by the data.  

To identify the remaining receive contribution, let
$\M{Q}_r=\M{I}_{N_r}-\M{A}_r\M{A}_r^\dagger$. Projection gives
\[
 \M{Q}_r\delta\M{Z}=\M{Q}_r\delta\M{A}_r\M{C}.
\]
Since $\M{C}$ has full row rank, multiplication by
$\M{C}$ is injective on these projected variations:
\[
 (\M{Q}_r\delta\M{A}_r\M{C})\M{C}^\dagger
 =\M{Q}_r\delta\M{A}_r.
\]
Moreover, $\M{Q}_r\delta\M{A}_r=\M{0}$ precisely when the steering variations remain inside $\operatorname{col}\M{A}_r$, thereby the receive subspace has zero first-order variation. The maps $\delta\V{\xi}\mapsto\M{Q}_r\delta\M{A}_r$ and the differential defining $d_R$ therefore have the same kernel and the same real rank $d_R$. The kernel of $\delta\M{Z}\mapsto\M{Q}_r\delta\M{Z}$ on the joint differential image consists exactly of $\M{A}_r\delta\M{C}$. Thus we have $r=2KT+d_R$.

Next hold $\M{S}$ fixed. The differential of the conditional map is $\delta\M{Z}=\delta(\M{H}\M{B})\M{S}$, and
\[
 \delta\M{Z}=\M{0}
 \quad\Longleftrightarrow\quad
 \delta(\M{H}\M{B})=\M{0},
\]
because right multiplication by $\M{S}^\dagger$ recovers $\delta(\M{H}\M{B})$. Thus the fixed-input conditional map and the effective channel map have the same real kernel, proving $r_u=d_H(\M{B})$. Both $r$ and $r_u$ are constant on the required probability-1 open sets by (G2).

The independent Gaussian data and (G1) give (A1) of \cref{ass:regular}, the preceding calculations give (A2), and (G3) gives (A3). Applying \cref{lem:smooth-image,prop:quotient-rank} therefore gives
\[
 \begin{aligned}
 I(\RM{S};\RM{Y}_\rho)
 &=\frac{r-r_u}{2}\log\rho+o(\log\rho)\\
 &=\frac{2KT+d_R-d_H(\M{B})}{2}\log\rho
   +o(\log\rho).
 \end{aligned}
\]
Dividing by $T\log\rho$ proves
\eqref{eq:fixed-precoder-prelog}.
Finally, full column rank of $\M{B}$ gives $\RM{S}=\M{B}^\dagger\RM{X}$, so $\RM{S}$ and $\RM{X}=\M{B}\RM{S}$ determine one another and have the same mutual information with $\RM{Y}_\rho$.
\end{IEEEproof}

\subsection{\texorpdfstring{Proof of \cref{prop:explicit-precoder-conditions}}{Proof of Proposition \getrefnumber{prop:explicit-precoder-conditions}}}
\label{app:proof-explicit-precoder-conditions}
\begin{IEEEproof}
First consider the receive subspace map $\V{\xi}\mapsto\operatorname{col}\M{A}_r$.
Since $\M{A}_r$ has full column rank, its subspace has zero first-order variation precisely when $\M{Q}_r\,\delta\M{A}_r=\M{0}$. Variations within the same column space have the form $\delta\M{A}_r=\M{A}_r\M{F}$ for a complex $K\times K$ matrix $\M{F}$. Taking columns gives
\begin{equation}\label{per_column}
 \M{Q}_rD_{\V{\xi}_k}\V{a}_r(\V{\xi}_k)
                  \delta\V{\xi}_k=\V{0},\qquad k=1,\ldots,K.
\end{equation}
By (R2), every $\delta\V{\xi}_k$ is zero. The receive subspace differential therefore has trivial kernel on its $Kp_r$-dimensional real parameter domain, proving $d_R=Kp_r$.

Next examine the kernel of the effective channel differential. The factorization $\M{H}\M{B}
 =\M{A}_r\diag(\V{\alpha})\M{A}_t^H\M{B}$ gives
\begin{align*}
 \delta(\M{H}\M{B})
 &=\delta\M{A}_r\diag(\V{\alpha})\M{A}_t^H\M{B}+\M{A}_r\,
       \delta\!\left(\diag(\V{\alpha})\M{A}_t^H\M{B}\right).
\end{align*}
Suppose $\delta(\M{H}\M{B})=\M{0}$. Left multiplication by $\M{Q}_r$ removes the second term. By (R1), $\diag(\V{\alpha})\M{A}_t^H\M{B}$ has full row rank. Right multiplication by one of its right inverses yields $\M{Q}_r\delta\M{A}_r=\M{0}$. The preceding argument for \eqref{per_column} then gives $\delta\V{\xi}_k=\V{0}$ for every path. Left multiplication of the remaining equation $\M{A}_r\,\delta\!\left(\diag(\V{\alpha})\M{A}_t^H\M{B}\right)=\V{0}$ by $\M{A}_r^\dagger$ separates the path rows:
\begin{equation}\label{rows_separated}
 \delta\alpha_k\,\V{b}(\V{\eta}_k)^H
 +\alpha_k
   \bigl(D_{\V{\eta}_k}\V{b}(\V{\eta}_k)
                    \delta\V{\eta}_k\bigr)^H
 =\V{0}^H.
\end{equation}
Here $\V{b}(\V{\eta}_k)\ne\V{0}$ and $\alpha_k\ne0$ by (R1). Such an equation admits a gain variation $\delta\alpha_k$ if and only if
\[
 \left(\M{I}_L-
   \frac{\V{b}(\V{\eta}_k)\V{b}(\V{\eta}_k)^H}
        {\|\V{b}(\V{\eta}_k)\|_2^2}\right)
 D_{\V{\eta}_k}\V{b}(\V{\eta}_k)\delta\V{\eta}_k
 =\V{0}.
\]
For every admissible $\delta\V{\eta}_k$, the compensating complex gain variation $\delta\alpha_k$ is uniquely determined by \eqref{rows_separated}. By (R3), these variations form a real kernel of dimension $p_t-q_{\M{B}}$. The path equations are separate, thus the complete effective channel kernel has real dimension $K(p_t-q_{\M{B}})$. Thereby we obtain
\begin{align*}
 d_H(\M{B})
 &=K(p_r+p_t+2)-K(p_t-q_{\M{B}})\\
 &=K(p_r+q_{\M{B}}+2).
\end{align*}
In particular, both $d_R$ and $d_H(\M{B})$ are constant on the required open set. Thus (G2) of \cref{prop:fixed-precoder-conditional} follows from (R1)--(R3). Together with (G1) and (G3), that proposition gives \eqref{eq:projected-geometry-prelog}.
\end{IEEEproof}

\subsection{Verification of the Array Examples}
\label{app:proof-array-examples}
We verify (R1)--(R2) for each example, using the Gaussian input of Sec.~\ref{sec:projected-geometry} with $K\le L\le T$. Condition (R3) then follows from the projective rank results in Sec.~\ref{ssec:examples_arrays}.

\subsubsection{UPA with antenna selection}

\emph{(R1): Separation of the path responses.}
Select $L$ distinct transmit elements $(m_\ell,n_\ell)$. Their response functions $e^{\jmathu(m_\ell u+n_\ell v)}$ are linearly independent on every nonempty open direction set. Hence the response vectors span $\C^L$, one can then choose $K$ directions with linearly independent responses. Therefore, there exists at least one $K\times K$ minor of $\M{A}_t^H\M{B}$ that is a nontrivial real-analytic function. Its zero set is null, proving the transmit rank in (R1)~\cite{KrantzParks2002}. Now, choose a receive UPA containing $(m,0)$ for $m=0,\ldots,K$ and the additional element $(0,1)$. Writing $z_k=e^{\jmathu u_k}$, the first $K+1$ receive rows have entries $z_k^m$. The $z_k$ are distinct almost surely, so the ordinary Vandermonde submatrix has column rank $K$. This proves the receive rank in (R1).

\emph{(R2): Local identifiability of \acp{aoa}.}
On the first $K+1$ rows of $\M{A}_r$, append the partial derivative in $u_k$ to the $K$ steering columns. The resulting square matrix is nonsingular, since a polynomial of degree at most $K$ annihilating its columns would vanish at the $K$ distinct nodes and have zero derivative at $z_k\ne0$, forcing the polynomial to be identically zero. The partial derivative in $v_k$ vanishes on these rows and has the nonzero entry $\jmathu e^{\jmathu v_k}$ on element $(0,1)$. Consequently,
\[
 \rank_{\C}
 [\M{A}_r,\partial_{u_k}\V{a}_{r,k},
              \partial_{v_k}\V{a}_{r,k}]=K+2.
\]
Projection by $\M{Q}_r$ leaves the two partial derivative columns that are linearly independent in $\mathbb{C}$, which are also linearly independent in $\mathbb{R}$. Thus (R2) holds with $p_r=2$. Additional receive elements preserve these ranks.

\subsubsection{Rectangular aperture with Fourier-mode selection}
Use the modal responses in \eqref{eq:aperture-mode-response}.

\emph{(R1): Separation of the path responses.}
Choose $L$ distinct transmit modes. Their response functions are linearly independent, because their values at the selected integer pairs form the identity matrix. Real analyticity extends this independence to every nonempty open direction set. The response vectors therefore span $\C^L$, and there exists a subset containing $K$ of them that are mutually linearly independent. Thus $\rank_{\C}(\M{A}_t^H\M{B})=K$ holds almost surely.

For reception, choose the modes $(m,0)$, $m=0,\ldots,K$, and $(0,1)$. Away from the zero set of $b_{0,0}$, divide each receive steering column by this common response. The first $K+1$ components become
\[
 f_m(u)=(-1)^m\frac{u}{u-m},\qquad f_0(u)=1,
\]
and the last component becomes $v/(1-v)$. For distinct noninteger $u_1,\ldots,u_K$, the first $K$ rows form a nonsingular scaled Cauchy matrix with entries $1/(u_k-m)$, $m=0,\ldots,K-1$. Indeed, a linear combination of these rational functions has a numerator of degree at most $K-1$, so vanishing at all $K$ nodes forces it to be identically zero. Thus the receive part of (R1) also holds almost surely.

\emph{(R2): Local identifiability of \acp{aoa}.}
Consider all $K+1$ functions $f_0,\ldots,f_K$ and append the partial derivative in $u_k$ to their $K$ response columns. If a linear combination vanished on these columns, it would have zeros at all $u_j$ and a double zero at $u_k$. After division by $u$, the linear combination becomes
\[
g(u)=\sum_{m=0}^K \frac{c_m}{u} f_m(u)=\sum_{m=0}^K\frac{c_m(-1)^m}{u-m}.
\]
Its numerator has degree at most $K$, but has at least $K+1$ zeros, so it is identically zero. Multiplying the identity $g(u)=0$ by $u-m$ and letting $u\to m$ gives $c_m=0$ for every $m=0,\ldots,K$, proving the augmented matrix nonsingular. The $v_k$ derivative is zero on these $K+1$ rows and equals $(1-v_k)^{-2}\ne0$ on the final row. Hence we have $\rank_{\C}
 [\M{A}_r,\partial_{u_k}\V{a}_{r,k},\partial_{v_k}\V{a}_{r,k}]=K+2$, and (R2) holds with $p_r=2$. The column normalization preserves these ranks.

\bibliographystyle{IEEEtran}
\bibliography{IEEEabrv,sparse_mimo_dof}

\end{document}

%% file: nula_channel_tikz.tex
\begin{tikzpicture}[x=1cm,y=1cm,
  font=\fontsize{9}{11}\selectfont,
  >=Stealth,
  element/.style={circle,draw=black!80,fill=white,line width=.75pt,minimum size=4.2pt,inner sep=0pt},
  scatter/.style={diamond,draw=black!60,fill=black!15,minimum size=5pt,inner sep=0pt},
  ray/.style={line width=.75pt,postaction={decorate},
    decoration={markings,mark=at position .24 with {\arrow{Stealth}},
    mark=at position .77 with {\arrow{Stealth}}}},
  labelbg/.style={fill=white,inner sep=1.6pt},
  small/.style={font=\fontsize{8}{10}\selectfont}]
\definecolor{rayblue}{RGB}{30,95,143}
\definecolor{panelgray}{RGB}{246,247,249}
\path[use as bounding box] (-.05,-2.22) rectangle (17.25,7.12);

\node[font=\fontsize{10}{12}\selectfont\bfseries] at (6.8,6.83)
  {NULA-SV channel within one block};
\coordinate (tx) at (2.3,2.5);
\coordinate (rx) at (11.4,2.5);

\draw[black!40,densely dotted,->] (2.3,2.25)--(2.3,5.85);
\draw[black!40,densely dotted,->] (11.4,2.25)--(11.4,5.85);
\node[small,anchor=south] at (2.3,5.9) {Tx array axis};
\node[small,anchor=south] at (11.4,5.9) {Rx array axis};
\foreach \yy in {2.5,3.0,4.0,5.25} \node[element] at (2.3,\yy) {};
\foreach \yy in {2.5,3.25,3.65,5.25} \node[element] at (11.4,\yy) {};
\node[fill=white,inner sep=0pt] at (2.3,4.63) {$\vdots$};
\node[fill=white,inner sep=0pt] at (11.4,4.47) {$\vdots$};
\node[anchor=east] at (2.08,2.27) {$x_{t,1}=0$};
\node[anchor=east] at (2.08,3.0) {$x_{t,2}$};
\node[anchor=east] at (2.08,4.0) {$x_{t,3}$};
\node[anchor=east] at (2.08,5.25) {$x_{t,N_t}$};
\node[anchor=west] at (11.62,2.27) {$x_{r,1}=0$};
\node[anchor=west] at (11.62,3.25) {$x_{r,2}$};
\node[anchor=west] at (11.62,3.65) {$x_{r,3}$};
\node[anchor=west] at (11.62,5.25) {$x_{r,N_r}$};
\node[align=center] at (2.3,1.35) {Tx NULA\\$N_t$ elements};
\node[align=center] at (11.4,1.35) {Rx NULA\\$N_r$ elements};

\coordinate (sone) at (6.15,6.0);
\coordinate (sk) at (7.25,4.4);
\coordinate (slast) at (6.35,1.55);
\draw[ray,black!48,dashed] (tx)--(sone)--(rx);
\draw[ray,black!48,dashed] (tx)--(slast)--(rx);
\draw[ray,rayblue,line width=1.15pt] (tx)--(sk)--(rx);
\node[scatter] at (sone) {};
\node[scatter] at (slast) {};
\node[scatter,draw=rayblue,fill=rayblue!15] at (sk) {};
\node[small,anchor=south] at (6.15,6.16) {Path $1$};
\node[small,anchor=north] at (6.35,1.38) {Path $K$};
\node[text=rayblue,anchor=south] at (7.25,4.56) {Path $k$};
\node[labelbg,text=black!65] at (4.65,4.76) {$\rv{\alpha}_1$};
\node[labelbg,text=black!65] at (8.0,1.77) {$\rv{\alpha}_K$};
\node[labelbg,text=rayblue] at (5.88,3.88) {$\rv{\alpha}_k$};
\node[labelbg,text=rayblue,small] at (4.25,2.94) {AoD cosine $\rv{u}_{t,k}$};
\node[labelbg,text=rayblue,small] at (9.22,2.95) {AoA cosine $\rv{u}_{r,k}$};
\draw[rayblue!75,line width=.45pt] (4.25,3.13)--(4.55,3.36);
\draw[rayblue!75,line width=.45pt] (9.22,3.14)--(9.12,3.44);

\draw[->,line width=.8pt] (.15,2.5)--(2.16,2.5);
\node[anchor=south] at (.78,2.63) {$\RM{X}_b$};
\draw[->,line width=.8pt] (11.54,2.5)--(13.50,2.5);
\node[small,anchor=south] at (12.5,2.64) {$\sqrt{\rho}\,\RM{H}_b\RM{X}_b$};
\node[circle,draw=black!75,line width=.7pt,minimum size=4.5mm,inner sep=0pt] (sum) at (13.73,2.5) {$+$};
\draw[->,line width=.8pt] (13.73,3.24)--(sum.north);
\node[anchor=south] at (13.73,3.29) {$\RM{N}_b$};
\draw[->,line width=.8pt] (sum.east)--(16.75,2.5);
\node[anchor=south] at (15.75,2.63) {$\RM{Y}_{\rho,b}$};

\draw[rounded corners=3pt,draw=black!25,fill=panelgray,line width=.55pt]
 (13.3,3.90) rectangle (17.1,6.45);
\node[font=\fontsize{9}{11}\selectfont\bfseries] at (15.2,6.11) {ULA special case};
\node[small] at (15.2,5.69) {Equal element spacing};
\draw[black!45,line width=.55pt] (13.9,5.12)--(16.5,5.12);
\foreach \xx in {14.0,14.8,15.6,16.4} \node[element] at (\xx,5.12) {};
\draw[<->,line width=.45pt] (14.0,4.78)--(14.8,4.78);
\node[small,anchor=north] at (14.4,4.72) {$d_\nu$};
\node[small] at (15.2,4.10) {$x_{\nu,n}=(n-1)d_\nu,~\nu\in\{t,r\}$};

\draw[black!18,line width=.5pt] (.15,.22)--(17.1,.22);
\node[anchor=west,font=\fontsize{9}{11}\selectfont\bfseries] at (.25,-.14)
 {Blockwise memoryless state};
\foreach \left/\right/\lab in {1.0/5.55/{b-1},5.9/10.45/{b},10.8/15.35/{b+1}} {
 \draw[rounded corners=2pt,draw=black!35,fill=panelgray,line width=.6pt]
  (\left,-1.46) rectangle (\right,-.47);
 \node at ({(\left+\right)/2},-.74) {Block $\lab$ };
 \node[small] at ({(\left+\right)/2},-1.16) {State $\RS{S}_{\lab}$ fixed};
}
\draw[->,black!65] (5.57,-.96)--(5.88,-.96);
\draw[->,black!65] (10.47,-.96)--(10.78,-.96);
\draw[->,black!65] (15.48,-.96)--(16.8,-.96);
\node[small,anchor=south] at (16.2,-.84) {time};
\end{tikzpicture}

%% file: preimage_entropy_tikz.tex
\begin{tikzpicture}[x=1cm,y=1cm,>=Stealth,font=\fontsize{10}{12}\selectfont,
 small/.style={font=\fontsize{9}{11}\selectfont},
 flow/.style={->,line width=.8pt,black!55}]
\definecolor{amber}{RGB}{183,110,25}
\definecolor{teal}{RGB}{20,114,106}
\definecolor{blue}{RGB}{42,107,158}
\definecolor{violet}{RGB}{122,93,150}
\path[use as bounding box] (.15,2.10) rectangle (19.05,8.9);
\draw[rounded corners=5pt,black!16,fill=black!1] (.3,5.75) rectangle (18.9,8.75);
\draw[rounded corners=5pt,teal!22,fill=teal!1] (.3,2.25) rectangle (18.9,5.4);

\path[fill=amber!12,draw=amber!60] (1.2,6.5)--(6.1,6.5)--(6.65,7.85)--(1.75,7.85)--cycle;
\fill[amber] (1.4729,6.6688) circle (1.1pt);
\fill[amber] (1.8813,6.6688) circle (1.1pt);
\fill[amber] (2.2896,6.6688) circle (1.1pt);
\fill[amber] (2.6979,6.6688) circle (1.1pt);
\fill[amber] (3.1063,6.6688) circle (1.1pt);
\fill[amber] (3.5146,6.6688) circle (1.1pt);
\fill[amber] (3.9229,6.6688) circle (1.1pt);
\fill[amber] (4.3312,6.6688) circle (1.1pt);
\fill[amber] (4.7396,6.6688) circle (1.1pt);
\fill[amber] (5.1479,6.6688) circle (1.1pt);
\fill[amber] (5.5563,6.6688) circle (1.1pt);
\fill[amber] (5.9646,6.6688) circle (1.1pt);
\fill[amber] (1.6104,7.0062) circle (1.1pt);
\fill[amber] (2.0187,7.0062) circle (1.1pt);
\fill[amber] (2.4271,7.0062) circle (1.1pt);
\fill[amber] (2.8354,7.0062) circle (1.1pt);
\fill[amber] (3.2437,7.0062) circle (1.1pt);
\fill[amber] (3.6521,7.0062) circle (1.1pt);
\fill[amber] (4.0604,7.0062) circle (1.1pt);
\fill[amber] (4.4688,7.0062) circle (1.1pt);
\fill[amber] (4.8771,7.0062) circle (1.1pt);
\fill[amber] (5.2854,7.0062) circle (1.1pt);
\fill[amber] (5.6938,7.0062) circle (1.1pt);
\fill[amber] (6.1021,7.0062) circle (1.1pt);
\fill[amber] (1.7479,7.3438) circle (1.1pt);
\fill[amber] (2.1562,7.3438) circle (1.1pt);
\fill[amber] (2.5646,7.3438) circle (1.1pt);
\fill[amber] (2.9729,7.3438) circle (1.1pt);
\fill[amber] (3.3813,7.3438) circle (1.1pt);
\fill[amber] (3.7896,7.3438) circle (1.1pt);
\fill[amber] (4.1979,7.3438) circle (1.1pt);
\fill[amber] (4.6063,7.3438) circle (1.1pt);
\fill[amber] (5.0146,7.3438) circle (1.1pt);
\fill[amber] (5.4229,7.3438) circle (1.1pt);
\fill[amber] (5.8313,7.3438) circle (1.1pt);
\fill[amber] (6.2396,7.3438) circle (1.1pt);
\fill[amber] (1.8854,7.6813) circle (1.1pt);
\fill[amber] (2.2938,7.6813) circle (1.1pt);
\fill[amber] (2.7021,7.6813) circle (1.1pt);
\fill[amber] (3.1104,7.6813) circle (1.1pt);
\fill[amber] (3.5188,7.6813) circle (1.1pt);
\fill[amber] (3.9271,7.6813) circle (1.1pt);
\fill[amber] (4.3354,7.6813) circle (1.1pt);
\fill[amber] (4.7438,7.6813) circle (1.1pt);
\fill[amber] (5.1521,7.6813) circle (1.1pt);
\fill[amber] (5.5604,7.6813) circle (1.1pt);
\fill[amber] (5.9688,7.6813) circle (1.1pt);
\fill[amber] (6.3771,7.6813) circle (1.1pt);
\node[small,text=amber] at (3.95,8.35) {A large parameter resolution cell};
\node[small] at (3.95,6.13) {Many parameters fall into the same output ball};
\draw[flow] (6.95,7.22)--(8.8,7.22);
\node[small] at (7.9,7.54) {Channel map};
\path[fill=amber!12,draw=amber!55] (1.2000,3.0000)--(2.4250,3.0000)--(2.9750,4.3500)--(1.7500,4.3500)--cycle;
\path[fill=teal!12,draw=teal!55] (2.4250,3.0000)--(3.6500,3.0000)--(4.2000,4.3500)--(2.9750,4.3500)--cycle;
\path[fill=blue!12,draw=blue!55] (3.6500,3.0000)--(4.8750,3.0000)--(5.4250,4.3500)--(4.2000,4.3500)--cycle;
\path[fill=violet!12,draw=violet!55] (4.8750,3.0000)--(6.1000,3.0000)--(6.6500,4.3500)--(5.4250,4.3500)--cycle;
\fill[amber] (1.4729,3.1688) circle (1.1pt);
\fill[amber] (1.8813,3.1688) circle (1.1pt);
\fill[amber] (2.2896,3.1688) circle (1.1pt);
\fill[teal] (2.6979,3.1688) circle (1.1pt);
\fill[teal] (3.1063,3.1688) circle (1.1pt);
\fill[teal] (3.5146,3.1688) circle (1.1pt);
\fill[blue] (3.9229,3.1688) circle (1.1pt);
\fill[blue] (4.3312,3.1688) circle (1.1pt);
\fill[blue] (4.7396,3.1688) circle (1.1pt);
\fill[violet] (5.1479,3.1688) circle (1.1pt);
\fill[violet] (5.5563,3.1688) circle (1.1pt);
\fill[violet] (5.9646,3.1688) circle (1.1pt);
\fill[amber] (1.6104,3.5063) circle (1.1pt);
\fill[amber] (2.0187,3.5063) circle (1.1pt);
\fill[amber] (2.4271,3.5063) circle (1.1pt);
\fill[teal] (2.8354,3.5063) circle (1.1pt);
\fill[teal] (3.2437,3.5063) circle (1.1pt);
\fill[teal] (3.6521,3.5063) circle (1.1pt);
\fill[blue] (4.0604,3.5063) circle (1.1pt);
\fill[blue] (4.4688,3.5063) circle (1.1pt);
\fill[blue] (4.8771,3.5063) circle (1.1pt);
\fill[violet] (5.2854,3.5063) circle (1.1pt);
\fill[violet] (5.6938,3.5063) circle (1.1pt);
\fill[violet] (6.1021,3.5063) circle (1.1pt);
\fill[amber] (1.7479,3.8438) circle (1.1pt);
\fill[amber] (2.1562,3.8438) circle (1.1pt);
\fill[amber] (2.5646,3.8438) circle (1.1pt);
\fill[teal] (2.9729,3.8438) circle (1.1pt);
\fill[teal] (3.3813,3.8438) circle (1.1pt);
\fill[teal] (3.7896,3.8438) circle (1.1pt);
\fill[blue] (4.1979,3.8438) circle (1.1pt);
\fill[blue] (4.6063,3.8438) circle (1.1pt);
\fill[blue] (5.0146,3.8438) circle (1.1pt);
\fill[violet] (5.4229,3.8438) circle (1.1pt);
\fill[violet] (5.8313,3.8438) circle (1.1pt);
\fill[violet] (6.2396,3.8438) circle (1.1pt);
\fill[amber] (1.8854,4.1813) circle (1.1pt);
\fill[amber] (2.2938,4.1813) circle (1.1pt);
\fill[amber] (2.7021,4.1813) circle (1.1pt);
\fill[teal] (3.1104,4.1813) circle (1.1pt);
\fill[teal] (3.5188,4.1813) circle (1.1pt);
\fill[teal] (3.9271,4.1813) circle (1.1pt);
\fill[blue] (4.3354,4.1813) circle (1.1pt);
\fill[blue] (4.7438,4.1813) circle (1.1pt);
\fill[blue] (5.1521,4.1813) circle (1.1pt);
\fill[violet] (5.5604,4.1813) circle (1.1pt);
\fill[violet] (5.9688,4.1813) circle (1.1pt);
\fill[violet] (6.3771,4.1813) circle (1.1pt);
\node[small,text=teal] at (3.95,4.99) {Smaller parameter resolution cells};
\node[small] at (3.95,2.61) {Each output ball receives only part of the probability};
\draw[flow] (6.95,3.7199999999999998)--(8.8,3.7199999999999998);
\node[small] at (7.9,4.04) {Channel map};
\path[fill=teal!4,draw=teal!25] (9.1920,6.9436)--(9.2899,6.8949)--(9.3879,6.8478)--(9.4858,6.8023)--(9.5837,6.7585)--(9.6817,6.7164)--(9.7796,6.6759)--(9.8775,6.6371)--(9.9755,6.5999)--(10.0734,6.5644)--(10.1713,6.5305)--(10.2693,6.4983)--(10.3672,6.4677)--(10.4651,6.4388)--(10.5631,6.4115)--(10.6610,6.3859)--(10.7589,6.3619)--(10.8569,6.3396)--(10.9548,6.3190)--(11.0527,6.3000)--(11.1507,6.2826)--(11.2486,6.2669)--(11.3465,6.2529)--(11.4445,6.2405)--(11.5424,6.2297)--(11.6403,6.2207)--(11.7383,6.2132)--(11.8362,6.2074)--(11.9341,6.2033)--(12.0321,6.2008)--(12.1300,6.2000)--(12.2279,6.2008)--(12.3259,6.2033)--(12.4238,6.2074)--(12.5217,6.2132)--(12.6197,6.2207)--(12.7176,6.2297)--(12.8155,6.2405)--(12.9135,6.2529)--(13.0114,6.2669)--(13.1093,6.2826)--(13.2073,6.3000)--(13.3052,6.3190)--(13.4031,6.3396)--(13.5011,6.3619)--(13.5990,6.3859)--(13.6969,6.4115)--(13.7949,6.4388)--(13.8928,6.4677)--(13.9907,6.4983)--(14.0887,6.5305)--(14.1866,6.5644)--(14.2845,6.5999)--(14.3825,6.6371)--(14.4804,6.6759)--(14.5783,6.7164)--(14.6763,6.7585)--(14.7742,6.8023)--(14.8721,6.8478)--(14.9701,6.8949)--(15.0680,6.9436)--(15.7080,7.5836)--(15.6101,7.5349)--(15.5121,7.4878)--(15.4142,7.4423)--(15.3163,7.3985)--(15.2183,7.3564)--(15.1204,7.3159)--(15.0225,7.2771)--(14.9245,7.2399)--(14.8266,7.2044)--(14.7287,7.1705)--(14.6307,7.1383)--(14.5328,7.1077)--(14.4349,7.0788)--(14.3369,7.0515)--(14.2390,7.0259)--(14.1411,7.0019)--(14.0431,6.9796)--(13.9452,6.9590)--(13.8473,6.9400)--(13.7493,6.9226)--(13.6514,6.9069)--(13.5535,6.8929)--(13.4555,6.8805)--(13.3576,6.8697)--(13.2597,6.8607)--(13.1617,6.8532)--(13.0638,6.8474)--(12.9659,6.8433)--(12.8679,6.8408)--(12.7700,6.8400)--(12.6721,6.8408)--(12.5741,6.8433)--(12.4762,6.8474)--(12.3783,6.8532)--(12.2803,6.8607)--(12.1824,6.8697)--(12.0845,6.8805)--(11.9865,6.8929)--(11.8886,6.9069)--(11.7907,6.9226)--(11.6927,6.9400)--(11.5948,6.9590)--(11.4969,6.9796)--(11.3989,7.0019)--(11.3010,7.0259)--(11.2031,7.0515)--(11.1051,7.0788)--(11.0072,7.1077)--(10.9093,7.1383)--(10.8113,7.1705)--(10.7134,7.2044)--(10.6155,7.2399)--(10.5175,7.2771)--(10.4196,7.3159)--(10.3217,7.3564)--(10.2237,7.3985)--(10.1258,7.4423)--(10.0279,7.4878)--(9.9299,7.5349)--(9.8320,7.5836)--cycle;
\draw[teal!12,line width=.35pt] (9.3840,7.1356)--(9.4819,7.0869)--(9.5799,7.0398)--(9.6778,6.9943)--(9.7757,6.9505)--(9.8737,6.9084)--(9.9716,6.8679)--(10.0695,6.8291)--(10.1675,6.7919)--(10.2654,6.7564)--(10.3633,6.7225)--(10.4613,6.6903)--(10.5592,6.6597)--(10.6571,6.6308)--(10.7551,6.6035)--(10.8530,6.5779)--(10.9509,6.5539)--(11.0489,6.5316)--(11.1468,6.5110)--(11.2447,6.4920)--(11.3427,6.4746)--(11.4406,6.4589)--(11.5385,6.4449)--(11.6365,6.4325)--(11.7344,6.4217)--(11.8323,6.4127)--(11.9303,6.4052)--(12.0282,6.3994)--(12.1261,6.3953)--(12.2241,6.3928)--(12.3220,6.3920)--(12.4199,6.3928)--(12.5179,6.3953)--(12.6158,6.3994)--(12.7137,6.4052)--(12.8117,6.4127)--(12.9096,6.4217)--(13.0075,6.4325)--(13.1055,6.4449)--(13.2034,6.4589)--(13.3013,6.4746)--(13.3993,6.4920)--(13.4972,6.5110)--(13.5951,6.5316)--(13.6931,6.5539)--(13.7910,6.5779)--(13.8889,6.6035)--(13.9869,6.6308)--(14.0848,6.6597)--(14.1827,6.6903)--(14.2807,6.7225)--(14.3786,6.7564)--(14.4765,6.7919)--(14.5745,6.8291)--(14.6724,6.8679)--(14.7703,6.9084)--(14.8683,6.9505)--(14.9662,6.9943)--(15.0641,7.0398)--(15.1621,7.0869)--(15.2600,7.1356);
\draw[teal!12,line width=.35pt] (9.6400,7.3916)--(9.7379,7.3429)--(9.8359,7.2958)--(9.9338,7.2503)--(10.0317,7.2065)--(10.1297,7.1644)--(10.2276,7.1239)--(10.3255,7.0851)--(10.4235,7.0479)--(10.5214,7.0124)--(10.6193,6.9785)--(10.7173,6.9463)--(10.8152,6.9157)--(10.9131,6.8868)--(11.0111,6.8595)--(11.1090,6.8339)--(11.2069,6.8099)--(11.3049,6.7876)--(11.4028,6.7670)--(11.5007,6.7480)--(11.5987,6.7306)--(11.6966,6.7149)--(11.7945,6.7009)--(11.8925,6.6885)--(11.9904,6.6777)--(12.0883,6.6687)--(12.1863,6.6612)--(12.2842,6.6554)--(12.3821,6.6513)--(12.4801,6.6488)--(12.5780,6.6480)--(12.6759,6.6488)--(12.7739,6.6513)--(12.8718,6.6554)--(12.9697,6.6612)--(13.0677,6.6687)--(13.1656,6.6777)--(13.2635,6.6885)--(13.3615,6.7009)--(13.4594,6.7149)--(13.5573,6.7306)--(13.6553,6.7480)--(13.7532,6.7670)--(13.8511,6.7876)--(13.9491,6.8099)--(14.0470,6.8339)--(14.1449,6.8595)--(14.2429,6.8868)--(14.3408,6.9157)--(14.4387,6.9463)--(14.5367,6.9785)--(14.6346,7.0124)--(14.7325,7.0479)--(14.8305,7.0851)--(14.9284,7.1239)--(15.0263,7.1644)--(15.1243,7.2065)--(15.2222,7.2503)--(15.3201,7.2958)--(15.4181,7.3429)--(15.5160,7.3916);
\fill[amber!8,opacity=.65] (12.4500,6.5200) circle (.48);
\fill[amber] (12.1787,6.3133) circle (1.1pt);
\fill[amber] (12.2562,6.3084) circle (1.1pt);
\fill[amber] (12.3337,6.3076) circle (1.1pt);
\fill[amber] (12.4112,6.3113) circle (1.1pt);
\fill[amber] (12.4887,6.3189) circle (1.1pt);
\fill[amber] (12.5663,6.3310) circle (1.1pt);
\fill[amber] (12.6437,6.3470) circle (1.1pt);
\fill[amber] (12.7212,6.3675) circle (1.1pt);
\fill[amber] (12.1787,6.3953) circle (1.1pt);
\fill[amber] (12.2562,6.3906) circle (1.1pt);
\fill[amber] (12.3337,6.3898) circle (1.1pt);
\fill[amber] (12.4112,6.3935) circle (1.1pt);
\fill[amber] (12.4887,6.4012) circle (1.1pt);
\fill[amber] (12.5663,6.4132) circle (1.1pt);
\fill[amber] (12.6437,6.4292) circle (1.1pt);
\fill[amber] (12.7212,6.4497) circle (1.1pt);
\fill[amber] (12.1787,6.4776) circle (1.1pt);
\fill[amber] (12.2562,6.4726) circle (1.1pt);
\fill[amber] (12.3337,6.4721) circle (1.1pt);
\fill[amber] (12.4112,6.4755) circle (1.1pt);
\fill[amber] (12.4887,6.4834) circle (1.1pt);
\fill[amber] (12.5663,6.4952) circle (1.1pt);
\fill[amber] (12.6437,6.5115) circle (1.1pt);
\fill[amber] (12.7212,6.5317) circle (1.1pt);
\fill[amber] (12.1787,6.5598) circle (1.1pt);
\fill[amber] (12.2562,6.5549) circle (1.1pt);
\fill[amber] (12.3337,6.5541) circle (1.1pt);
\fill[amber] (12.4112,6.5578) circle (1.1pt);
\fill[amber] (12.4887,6.5654) circle (1.1pt);
\fill[amber] (12.5663,6.5775) circle (1.1pt);
\fill[amber] (12.6437,6.5935) circle (1.1pt);
\fill[amber] (12.7212,6.6140) circle (1.1pt);
\fill[amber] (12.1787,6.6418) circle (1.1pt);
\fill[amber] (12.2562,6.6371) circle (1.1pt);
\fill[amber] (12.3337,6.6363) circle (1.1pt);
\fill[amber] (12.4112,6.6400) circle (1.1pt);
\fill[amber] (12.4887,6.6477) circle (1.1pt);
\fill[amber] (12.5663,6.6597) circle (1.1pt);
\fill[amber] (12.6437,6.6757) circle (1.1pt);
\fill[amber] (12.7212,6.6962) circle (1.1pt);
\fill[amber] (12.1787,6.7241) circle (1.1pt);
\fill[amber] (12.2562,6.7191) circle (1.1pt);
\fill[amber] (12.3337,6.7186) circle (1.1pt);
\fill[amber] (12.4112,6.7220) circle (1.1pt);
\fill[amber] (12.4887,6.7299) circle (1.1pt);
\fill[amber] (12.5663,6.7417) circle (1.1pt);
\fill[amber] (12.6437,6.7580) circle (1.1pt);
\fill[amber] (12.7212,6.7782) circle (1.1pt);
\draw[amber,line width=.75pt] (12.4500,6.5200) circle (.48);
\draw[amber!35,dashed,line width=.3pt] (11.9700,6.5200) arc[start angle=180,end angle=0,x radius=.48,y radius=.15];
\draw[amber!50,line width=.3pt] (11.9700,6.5200) arc[start angle=180,end angle=360,x radius=.48,y radius=.15];
\draw[amber!25,line width=.3pt] (12.4500,6.5200) ellipse (.17 and .48);
\node[small,text=amber] at (12.45,8.35) {One ball collects nearly all the probability};
\draw[->,amber!70] (12.45,8.05)--(12.45,7.08);
\node[align=center,small] at (17.30,7.18) {Hard to identify\\the parameters\\[5pt]\textbf{Few distinguishable}\\\textbf{outputs}};
\path[fill=teal!4,draw=teal!25] (9.1920,3.4436)--(9.2899,3.3949)--(9.3879,3.3478)--(9.4858,3.3023)--(9.5837,3.2585)--(9.6817,3.2164)--(9.7796,3.1759)--(9.8775,3.1371)--(9.9755,3.0999)--(10.0734,3.0644)--(10.1713,3.0305)--(10.2693,2.9983)--(10.3672,2.9677)--(10.4651,2.9388)--(10.5631,2.9115)--(10.6610,2.8859)--(10.7589,2.8619)--(10.8569,2.8396)--(10.9548,2.8190)--(11.0527,2.8000)--(11.1507,2.7826)--(11.2486,2.7669)--(11.3465,2.7529)--(11.4445,2.7405)--(11.5424,2.7297)--(11.6403,2.7207)--(11.7383,2.7132)--(11.8362,2.7074)--(11.9341,2.7033)--(12.0321,2.7008)--(12.1300,2.7000)--(12.2279,2.7008)--(12.3259,2.7033)--(12.4238,2.7074)--(12.5217,2.7132)--(12.6197,2.7207)--(12.7176,2.7297)--(12.8155,2.7405)--(12.9135,2.7529)--(13.0114,2.7669)--(13.1093,2.7826)--(13.2073,2.8000)--(13.3052,2.8190)--(13.4031,2.8396)--(13.5011,2.8619)--(13.5990,2.8859)--(13.6969,2.9115)--(13.7949,2.9388)--(13.8928,2.9677)--(13.9907,2.9983)--(14.0887,3.0305)--(14.1866,3.0644)--(14.2845,3.0999)--(14.3825,3.1371)--(14.4804,3.1759)--(14.5783,3.2164)--(14.6763,3.2585)--(14.7742,3.3023)--(14.8721,3.3478)--(14.9701,3.3949)--(15.0680,3.4436)--(15.7080,4.0836)--(15.6101,4.0349)--(15.5121,3.9878)--(15.4142,3.9423)--(15.3163,3.8985)--(15.2183,3.8564)--(15.1204,3.8159)--(15.0225,3.7771)--(14.9245,3.7399)--(14.8266,3.7044)--(14.7287,3.6705)--(14.6307,3.6383)--(14.5328,3.6077)--(14.4349,3.5788)--(14.3369,3.5515)--(14.2390,3.5259)--(14.1411,3.5019)--(14.0431,3.4796)--(13.9452,3.4590)--(13.8473,3.4400)--(13.7493,3.4226)--(13.6514,3.4069)--(13.5535,3.3929)--(13.4555,3.3805)--(13.3576,3.3697)--(13.2597,3.3607)--(13.1617,3.3532)--(13.0638,3.3474)--(12.9659,3.3433)--(12.8679,3.3408)--(12.7700,3.3400)--(12.6721,3.3408)--(12.5741,3.3433)--(12.4762,3.3474)--(12.3783,3.3532)--(12.2803,3.3607)--(12.1824,3.3697)--(12.0845,3.3805)--(11.9865,3.3929)--(11.8886,3.4069)--(11.7907,3.4226)--(11.6927,3.4400)--(11.5948,3.4590)--(11.4969,3.4796)--(11.3989,3.5019)--(11.3010,3.5259)--(11.2031,3.5515)--(11.1051,3.5788)--(11.0072,3.6077)--(10.9093,3.6383)--(10.8113,3.6705)--(10.7134,3.7044)--(10.6155,3.7399)--(10.5175,3.7771)--(10.4196,3.8159)--(10.3217,3.8564)--(10.2237,3.8985)--(10.1258,3.9423)--(10.0279,3.9878)--(9.9299,4.0349)--(9.8320,4.0836)--cycle;
\draw[teal!12,line width=.35pt] (9.3840,3.6356)--(9.4819,3.5869)--(9.5799,3.5398)--(9.6778,3.4943)--(9.7757,3.4505)--(9.8737,3.4084)--(9.9716,3.3679)--(10.0695,3.3291)--(10.1675,3.2919)--(10.2654,3.2564)--(10.3633,3.2225)--(10.4613,3.1903)--(10.5592,3.1597)--(10.6571,3.1308)--(10.7551,3.1035)--(10.8530,3.0779)--(10.9509,3.0539)--(11.0489,3.0316)--(11.1468,3.0110)--(11.2447,2.9920)--(11.3427,2.9746)--(11.4406,2.9589)--(11.5385,2.9449)--(11.6365,2.9325)--(11.7344,2.9217)--(11.8323,2.9127)--(11.9303,2.9052)--(12.0282,2.8994)--(12.1261,2.8953)--(12.2241,2.8928)--(12.3220,2.8920)--(12.4199,2.8928)--(12.5179,2.8953)--(12.6158,2.8994)--(12.7137,2.9052)--(12.8117,2.9127)--(12.9096,2.9217)--(13.0075,2.9325)--(13.1055,2.9449)--(13.2034,2.9589)--(13.3013,2.9746)--(13.3993,2.9920)--(13.4972,3.0110)--(13.5951,3.0316)--(13.6931,3.0539)--(13.7910,3.0779)--(13.8889,3.1035)--(13.9869,3.1308)--(14.0848,3.1597)--(14.1827,3.1903)--(14.2807,3.2225)--(14.3786,3.2564)--(14.4765,3.2919)--(14.5745,3.3291)--(14.6724,3.3679)--(14.7703,3.4084)--(14.8683,3.4505)--(14.9662,3.4943)--(15.0641,3.5398)--(15.1621,3.5869)--(15.2600,3.6356);
\draw[teal!12,line width=.35pt] (9.6400,3.8916)--(9.7379,3.8429)--(9.8359,3.7958)--(9.9338,3.7503)--(10.0317,3.7065)--(10.1297,3.6644)--(10.2276,3.6239)--(10.3255,3.5851)--(10.4235,3.5479)--(10.5214,3.5124)--(10.6193,3.4785)--(10.7173,3.4463)--(10.8152,3.4157)--(10.9131,3.3868)--(11.0111,3.3595)--(11.1090,3.3339)--(11.2069,3.3099)--(11.3049,3.2876)--(11.4028,3.2670)--(11.5007,3.2480)--(11.5987,3.2306)--(11.6966,3.2149)--(11.7945,3.2009)--(11.8925,3.1885)--(11.9904,3.1777)--(12.0883,3.1687)--(12.1863,3.1612)--(12.2842,3.1554)--(12.3821,3.1513)--(12.4801,3.1488)--(12.5780,3.1480)--(12.6759,3.1488)--(12.7739,3.1513)--(12.8718,3.1554)--(12.9697,3.1612)--(13.0677,3.1687)--(13.1656,3.1777)--(13.2635,3.1885)--(13.3615,3.2009)--(13.4594,3.2149)--(13.5573,3.2306)--(13.6553,3.2480)--(13.7532,3.2670)--(13.8511,3.2876)--(13.9491,3.3099)--(14.0470,3.3339)--(14.1449,3.3595)--(14.2429,3.3868)--(14.3408,3.4157)--(14.4387,3.4463)--(14.5367,3.4785)--(14.6346,3.5124)--(14.7325,3.5479)--(14.8305,3.5851)--(14.9284,3.6239)--(15.0263,3.6644)--(15.1243,3.7065)--(15.2222,3.7503)--(15.3201,3.7958)--(15.4181,3.8429)--(15.5160,3.8916);
\fill[amber!8,opacity=.65] (10.0770,3.5051) circle (.48);
\fill[amber] (9.8445,3.3874) circle (1.1pt);
\fill[amber] (9.9995,3.3861) circle (1.1pt);
\fill[amber] (10.1545,3.4016) circle (1.1pt);
\fill[amber] (10.3095,3.4339) circle (1.1pt);
\fill[amber] (9.8445,3.5008) circle (1.1pt);
\fill[amber] (9.9995,3.4995) circle (1.1pt);
\fill[amber] (10.1545,3.5150) circle (1.1pt);
\fill[amber] (10.3095,3.5473) circle (1.1pt);
\fill[amber] (9.8445,3.6141) circle (1.1pt);
\fill[amber] (9.9995,3.6128) circle (1.1pt);
\fill[amber] (10.1545,3.6283) circle (1.1pt);
\fill[amber] (10.3095,3.6606) circle (1.1pt);
\draw[amber,line width=.75pt] (10.0770,3.5051) circle (.48);
\draw[amber!35,dashed,line width=.3pt] (9.5970,3.5051) arc[start angle=180,end angle=0,x radius=.48,y radius=.15];
\draw[amber!50,line width=.3pt] (9.5970,3.5051) arc[start angle=180,end angle=360,x radius=.48,y radius=.15];
\draw[amber!25,line width=.3pt] (10.0770,3.5051) ellipse (.17 and .48);
\fill[teal!8,opacity=.65] (11.6590,3.0739) circle (.48);
\fill[teal] (11.4265,2.9562) circle (1.1pt);
\fill[teal] (11.5815,2.9549) circle (1.1pt);
\fill[teal] (11.7365,2.9704) circle (1.1pt);
\fill[teal] (11.8915,3.0027) circle (1.1pt);
\fill[teal] (11.4265,3.0696) circle (1.1pt);
\fill[teal] (11.5815,3.0683) circle (1.1pt);
\fill[teal] (11.7365,3.0838) circle (1.1pt);
\fill[teal] (11.8915,3.1161) circle (1.1pt);
\fill[teal] (11.4265,3.1829) circle (1.1pt);
\fill[teal] (11.5815,3.1816) circle (1.1pt);
\fill[teal] (11.7365,3.1971) circle (1.1pt);
\fill[teal] (11.8915,3.2294) circle (1.1pt);
\draw[teal,line width=.75pt] (11.6590,3.0739) circle (.48);
\draw[teal!35,dashed,line width=.3pt] (11.1790,3.0739) arc[start angle=180,end angle=0,x radius=.48,y radius=.15];
\draw[teal!50,line width=.3pt] (11.1790,3.0739) arc[start angle=180,end angle=360,x radius=.48,y radius=.15];
\draw[teal!25,line width=.3pt] (11.6590,3.0739) ellipse (.17 and .48);
\fill[blue!8,opacity=.65] (13.2410,3.0739) circle (.48);
\fill[blue] (13.0085,2.9562) circle (1.1pt);
\fill[blue] (13.1635,2.9549) circle (1.1pt);
\fill[blue] (13.3185,2.9704) circle (1.1pt);
\fill[blue] (13.4735,3.0027) circle (1.1pt);
\fill[blue] (13.0085,3.0696) circle (1.1pt);
\fill[blue] (13.1635,3.0683) circle (1.1pt);
\fill[blue] (13.3185,3.0838) circle (1.1pt);
\fill[blue] (13.4735,3.1161) circle (1.1pt);
\fill[blue] (13.0085,3.1829) circle (1.1pt);
\fill[blue] (13.1635,3.1816) circle (1.1pt);
\fill[blue] (13.3185,3.1971) circle (1.1pt);
\fill[blue] (13.4735,3.2294) circle (1.1pt);
\draw[blue,line width=.75pt] (13.2410,3.0739) circle (.48);
\draw[blue!35,dashed,line width=.3pt] (12.7610,3.0739) arc[start angle=180,end angle=0,x radius=.48,y radius=.15];
\draw[blue!50,line width=.3pt] (12.7610,3.0739) arc[start angle=180,end angle=360,x radius=.48,y radius=.15];
\draw[blue!25,line width=.3pt] (13.2410,3.0739) ellipse (.17 and .48);
\fill[violet!8,opacity=.65] (14.8230,3.5051) circle (.48);
\fill[violet] (14.5905,3.3874) circle (1.1pt);
\fill[violet] (14.7455,3.3861) circle (1.1pt);
\fill[violet] (14.9005,3.4016) circle (1.1pt);
\fill[violet] (15.0555,3.4339) circle (1.1pt);
\fill[violet] (14.5905,3.5008) circle (1.1pt);
\fill[violet] (14.7455,3.4995) circle (1.1pt);
\fill[violet] (14.9005,3.5150) circle (1.1pt);
\fill[violet] (15.0555,3.5473) circle (1.1pt);
\fill[violet] (14.5905,3.6141) circle (1.1pt);
\fill[violet] (14.7455,3.6128) circle (1.1pt);
\fill[violet] (14.9005,3.6283) circle (1.1pt);
\fill[violet] (15.0555,3.6606) circle (1.1pt);
\draw[violet,line width=.75pt] (14.8230,3.5051) circle (.48);
\draw[violet!35,dashed,line width=.3pt] (14.3430,3.5051) arc[start angle=180,end angle=0,x radius=.48,y radius=.15];
\draw[violet!50,line width=.3pt] (14.3430,3.5051) arc[start angle=180,end angle=360,x radius=.48,y radius=.15];
\draw[violet!25,line width=.3pt] (14.8230,3.5051) ellipse (.17 and .48);
\node[small,text=teal] at (12.45,4.99) {The probability is distributed across several balls};
\node[align=center,small] at (17.30,3.71) {Easier to identify\\the parameters\\[5pt]\textbf{More distinguishable}\\\textbf{outputs}};

\end{tikzpicture}

%% file: theorem2_grid_tikz.tex
\begin{tikzpicture}[x=1cm,y=1cm,>=Stealth,font=\fontsize{10}{12}\selectfont,
 small/.style={font=\fontsize{9}{11}\selectfont},
 flow/.style={->,black!50,line width=.75pt}]
\definecolor{gridblue}{RGB}{35,101,152}
\definecolor{lowerteal}{RGB}{20,114,104}
\definecolor{lowgold}{RGB}{182,112,25}
\path[use as bounding box] (.1,2.05) rectangle (19.6,8.0);
\path[fill=lowgold!13] (0.8000,3.0000)--(3.2000,3.0000)--(3.2000,3.9193)--(3.1760,3.9040)--(3.1520,3.8888)--(3.1280,3.8738)--(3.1040,3.8589)--(3.0800,3.8441)--(3.0560,3.8295)--(3.0320,3.8149)--(3.0080,3.8006)--(2.9840,3.7863)--(2.9600,3.7722)--(2.9360,3.7583)--(2.9120,3.7444)--(2.8880,3.7308)--(2.8640,3.7172)--(2.8400,3.7038)--(2.8160,3.6906)--(2.7920,3.6775)--(2.7680,3.6646)--(2.7440,3.6518)--(2.7200,3.6391)--(2.6960,3.6266)--(2.6720,3.6143)--(2.6480,3.6021)--(2.6240,3.5900)--(2.6000,3.5782)--(2.5760,3.5664)--(2.5520,3.5549)--(2.5280,3.5434)--(2.5040,3.5322)--(2.4800,3.5211)--(2.4560,3.5101)--(2.4320,3.4993)--(2.4080,3.4887)--(2.3840,3.4782)--(2.3600,3.4679)--(2.3360,3.4578)--(2.3120,3.4478)--(2.2880,3.4379)--(2.2640,3.4282)--(2.2400,3.4187)--(2.2160,3.4093)--(2.1920,3.4001)--(2.1680,3.3910)--(2.1440,3.3821)--(2.1200,3.3734)--(2.0960,3.3648)--(2.0720,3.3563)--(2.0480,3.3480)--(2.0240,3.3399)--(2.0000,3.3319)--(1.9760,3.3240)--(1.9520,3.3163)--(1.9280,3.3087)--(1.9040,3.3013)--(1.8800,3.2940)--(1.8560,3.2869)--(1.8320,3.2799)--(1.8080,3.2731)--(1.7840,3.2663)--(1.7600,3.2598)--(1.7360,3.2533)--(1.7120,3.2470)--(1.6880,3.2408)--(1.6640,3.2348)--(1.6400,3.2289)--(1.6160,3.2231)--(1.5920,3.2174)--(1.5680,3.2119)--(1.5440,3.2064)--(1.5200,3.2011)--(1.4960,3.1959)--(1.4720,3.1909)--(1.4480,3.1859)--(1.4240,3.1811)--(1.4000,3.1764)--(1.3760,3.1717)--(1.3520,3.1672)--(1.3280,3.1628)--(1.3040,3.1585)--(1.2800,3.1543)--(1.2560,3.1502)--(1.2320,3.1462)--(1.2080,3.1423)--(1.1840,3.1385)--(1.1600,3.1347)--(1.1360,3.1311)--(1.1120,3.1276)--(1.0880,3.1241)--(1.0640,3.1207)--(1.0400,3.1175)--(1.0160,3.1142)--(0.9920,3.1111)--(0.9680,3.1081)--(0.9440,3.1051)--(0.9200,3.1022)--(0.8960,3.0994)--(0.8720,3.0967)--(0.8480,3.0940)--(0.8240,3.0914)--(0.8000,3.0888)--cycle;
\path[fill=gridblue!5] (3.2000,3.0000)--(7.2000,3.0000)--(7.2000,7.2017)--(7.1600,7.1668)--(7.1200,7.1319)--(7.0800,7.0970)--(7.0400,7.0621)--(7.0000,7.0272)--(6.9600,6.9923)--(6.9200,6.9575)--(6.8800,6.9226)--(6.8400,6.8877)--(6.8000,6.8529)--(6.7600,6.8180)--(6.7200,6.7832)--(6.6800,6.7483)--(6.6400,6.7135)--(6.6000,6.6787)--(6.5600,6.6439)--(6.5200,6.6090)--(6.4800,6.5743)--(6.4400,6.5395)--(6.4000,6.5047)--(6.3600,6.4699)--(6.3200,6.4352)--(6.2800,6.4005)--(6.2400,6.3657)--(6.2000,6.3310)--(6.1600,6.2963)--(6.1200,6.2617)--(6.0800,6.2270)--(6.0400,6.1924)--(6.0000,6.1577)--(5.9600,6.1231)--(5.9200,6.0885)--(5.8800,6.0540)--(5.8400,6.0194)--(5.8000,5.9849)--(5.7600,5.9504)--(5.7200,5.9159)--(5.6800,5.8815)--(5.6400,5.8471)--(5.6000,5.8127)--(5.5600,5.7784)--(5.5200,5.7440)--(5.4800,5.7097)--(5.4400,5.6755)--(5.4000,5.6413)--(5.3600,5.6071)--(5.3200,5.5730)--(5.2800,5.5389)--(5.2400,5.5048)--(5.2000,5.4708)--(5.1600,5.4369)--(5.1200,5.4030)--(5.0800,5.3691)--(5.0400,5.3354)--(5.0000,5.3016)--(4.9600,5.2680)--(4.9200,5.2344)--(4.8800,5.2008)--(4.8400,5.1674)--(4.8000,5.1340)--(4.7600,5.1007)--(4.7200,5.0675)--(4.6800,5.0344)--(4.6400,5.0013)--(4.6000,4.9684)--(4.5600,4.9355)--(4.5200,4.9028)--(4.4800,4.8702)--(4.4400,4.8376)--(4.4000,4.8052)--(4.3600,4.7729)--(4.3200,4.7408)--(4.2800,4.7088)--(4.2400,4.6769)--(4.2000,4.6451)--(4.1600,4.6136)--(4.1200,4.5821)--(4.0800,4.5509)--(4.0400,4.5198)--(4.0000,4.4888)--(3.9600,4.4581)--(3.9200,4.4276)--(3.8800,4.3972)--(3.8400,4.3671)--(3.8000,4.3372)--(3.7600,4.3075)--(3.7200,4.2780)--(3.6800,4.2487)--(3.6400,4.2197)--(3.6000,4.1910)--(3.5600,4.1625)--(3.5200,4.1343)--(3.4800,4.1064)--(3.4400,4.0787)--(3.4000,4.0514)--(3.3600,4.0243)--(3.3200,3.9976)--(3.2800,3.9711)--(3.2400,3.9450)--(3.2000,3.9193)--cycle;
\begin{scope}
\clip (3.2000,3.0000)--(7.2000,3.0000)--(7.2000,7.2017)--(7.1600,7.1668)--(7.1200,7.1319)--(7.0800,7.0970)--(7.0400,7.0621)--(7.0000,7.0272)--(6.9600,6.9923)--(6.9200,6.9575)--(6.8800,6.9226)--(6.8400,6.8877)--(6.8000,6.8529)--(6.7600,6.8180)--(6.7200,6.7832)--(6.6800,6.7483)--(6.6400,6.7135)--(6.6000,6.6787)--(6.5600,6.6439)--(6.5200,6.6090)--(6.4800,6.5743)--(6.4400,6.5395)--(6.4000,6.5047)--(6.3600,6.4699)--(6.3200,6.4352)--(6.2800,6.4005)--(6.2400,6.3657)--(6.2000,6.3310)--(6.1600,6.2963)--(6.1200,6.2617)--(6.0800,6.2270)--(6.0400,6.1924)--(6.0000,6.1577)--(5.9600,6.1231)--(5.9200,6.0885)--(5.8800,6.0540)--(5.8400,6.0194)--(5.8000,5.9849)--(5.7600,5.9504)--(5.7200,5.9159)--(5.6800,5.8815)--(5.6400,5.8471)--(5.6000,5.8127)--(5.5600,5.7784)--(5.5200,5.7440)--(5.4800,5.7097)--(5.4400,5.6755)--(5.4000,5.6413)--(5.3600,5.6071)--(5.3200,5.5730)--(5.2800,5.5389)--(5.2400,5.5048)--(5.2000,5.4708)--(5.1600,5.4369)--(5.1200,5.4030)--(5.0800,5.3691)--(5.0400,5.3354)--(5.0000,5.3016)--(4.9600,5.2680)--(4.9200,5.2344)--(4.8800,5.2008)--(4.8400,5.1674)--(4.8000,5.1340)--(4.7600,5.1007)--(4.7200,5.0675)--(4.6800,5.0344)--(4.6400,5.0013)--(4.6000,4.9684)--(4.5600,4.9355)--(4.5200,4.9028)--(4.4800,4.8702)--(4.4400,4.8376)--(4.4000,4.8052)--(4.3600,4.7729)--(4.3200,4.7408)--(4.2800,4.7088)--(4.2400,4.6769)--(4.2000,4.6451)--(4.1600,4.6136)--(4.1200,4.5821)--(4.0800,4.5509)--(4.0400,4.5198)--(4.0000,4.4888)--(3.9600,4.4581)--(3.9200,4.4276)--(3.8800,4.3972)--(3.8400,4.3671)--(3.8000,4.3372)--(3.7600,4.3075)--(3.7200,4.2780)--(3.6800,4.2487)--(3.6400,4.2197)--(3.6000,4.1910)--(3.5600,4.1625)--(3.5200,4.1343)--(3.4800,4.1064)--(3.4400,4.0787)--(3.4000,4.0514)--(3.3600,4.0243)--(3.3200,3.9976)--(3.2800,3.9711)--(3.2400,3.9450)--(3.2000,3.9193)--cycle;
\draw[gridblue!35,line width=.4pt] (3.2000,3.0000)--(3.2000,7.2700);
\draw[gridblue!35,line width=.4pt] (4.0000,3.0000)--(4.0000,7.2700);
\draw[gridblue!35,line width=.4pt] (4.8000,3.0000)--(4.8000,7.2700);
\draw[gridblue!35,line width=.4pt] (5.6000,3.0000)--(5.6000,7.2700);
\draw[gridblue!35,line width=.4pt] (6.4000,3.0000)--(6.4000,7.2700);
\draw[gridblue!35,line width=.4pt] (7.2000,3.0000)--(7.2000,7.2700);
\draw[gridblue!35,line width=.4pt] (3.2000,3.0000)--(7.2000,3.0000);
\draw[gridblue!35,line width=.4pt] (3.2000,3.7000)--(7.2000,3.7000);
\draw[gridblue!35,line width=.4pt] (3.2000,4.4000)--(7.2000,4.4000);
\draw[gridblue!35,line width=.4pt] (3.2000,5.1000)--(7.2000,5.1000);
\draw[gridblue!35,line width=.4pt] (3.2000,5.8000)--(7.2000,5.8000);
\draw[gridblue!35,line width=.4pt] (3.2000,6.5000)--(7.2000,6.5000);
\draw[gridblue!35,line width=.4pt] (3.2000,7.2000)--(7.2000,7.2000);
\draw[gridblue!35,line width=.4pt] (3.2000,7.9000)--(7.2000,7.9000);
\end{scope}
\draw[gridblue!65,line width=.7pt] (0.8000,3.0888)--(0.8640,3.0958)--(0.9280,3.1032)--(0.9920,3.1111)--(1.0560,3.1196)--(1.1200,3.1287)--(1.1840,3.1385)--(1.2480,3.1488)--(1.3120,3.1599)--(1.3760,3.1717)--(1.4400,3.1843)--(1.5040,3.1977)--(1.5680,3.2119)--(1.6320,3.2269)--(1.6960,3.2429)--(1.7600,3.2598)--(1.8240,3.2776)--(1.8880,3.2964)--(1.9520,3.3163)--(2.0160,3.3372)--(2.0800,3.3591)--(2.1440,3.3821)--(2.2080,3.4062)--(2.2720,3.4314)--(2.3360,3.4578)--(2.4000,3.4852)--(2.4640,3.5138)--(2.5280,3.5434)--(2.5920,3.5742)--(2.6560,3.6061)--(2.7200,3.6391)--(2.7840,3.6732)--(2.8480,3.7083)--(2.9120,3.7444)--(2.9760,3.7816)--(3.0400,3.8198)--(3.1040,3.8589)--(3.1680,3.8989)--(3.2320,3.9399)--(3.2960,3.9817)--(3.3600,4.0243)--(3.4240,4.0677)--(3.4880,4.1119)--(3.5520,4.1568)--(3.6160,4.2025)--(3.6800,4.2487)--(3.7440,4.2956)--(3.8080,4.3431)--(3.8720,4.3912)--(3.9360,4.4398)--(4.0000,4.4888)--(4.0640,4.5384)--(4.1280,4.5884)--(4.1920,4.6388)--(4.2560,4.6896)--(4.3200,4.7408)--(4.3840,4.7923)--(4.4480,4.8441)--(4.5120,4.8963)--(4.5760,4.9487)--(4.6400,5.0013)--(4.7040,5.0542)--(4.7680,5.1074)--(4.8320,5.1607)--(4.8960,5.2142)--(4.9600,5.2680)--(5.0240,5.3219)--(5.0880,5.3759)--(5.1520,5.4301)--(5.2160,5.4844)--(5.2800,5.5389)--(5.3440,5.5934)--(5.4080,5.6481)--(5.4720,5.7029)--(5.5360,5.7578)--(5.6000,5.8127)--(5.6640,5.8677)--(5.7280,5.9228)--(5.7920,5.9780)--(5.8560,6.0332)--(5.9200,6.0885)--(5.9840,6.1439)--(6.0480,6.1993)--(6.1120,6.2547)--(6.1760,6.3102)--(6.2400,6.3657)--(6.3040,6.4213)--(6.3680,6.4769)--(6.4320,6.5325)--(6.4960,6.5882)--(6.5600,6.6439)--(6.6240,6.6996)--(6.6880,6.7553)--(6.7520,6.8110)--(6.8160,6.8668)--(6.8800,6.9226)--(6.9440,6.9784)--(7.0080,7.0342)--(7.0720,7.0900)--(7.1360,7.1459)--(7.2000,7.2017);
\draw[lowgold!65,dashed,line width=.7pt] (3.2000,3.0000)--(3.2000,4.2600);
\draw[gridblue!55,line width=.6pt] (7.2000,3.0000)--(7.2000,7.2017);
\draw[->,black!60] (.7,3)--(7.8,3);
\draw[->,black!60] (.8,2.9)--(.8,7.8);
\node[anchor=west] at (.94,7.72) {$q$};
\node[anchor=west] at (7.88,3.0) {$s$};
\node[small,anchor=east] at (.66,3) {$0$};
\node[small,align=center] at (3.2,2.55) {$\log(A^2\rho^\Delta)$};
\node[small,align=center] at (7.15,2.55) {$\log(A^2\rho)$};
\draw[black!60] (3.2,2.95)--(3.2,3.05);
\draw[black!60] (7.2,2.95)--(7.2,3.05);
\node[small,text=gridblue] at (4.45,7.12) {$q\le\log(1+e^s)$};
\draw[->,gridblue!65] (5.15,6.94)--(6.42,6.52);
\node[small,align=center,text=lowgold] at (1.93,5.2)
 {Low-SNR cell\\$O(\Delta\log\rho)$};
\draw[->,lowgold!70] (2.15,4.75)--(2.45,3.42);
\path[fill=gridblue!17,draw=gridblue,line width=.9pt] (4.8,4.4) rectangle (5.6,5.1);
\node[text=gridblue] at (5.2,4.75) {$\Set C$};
\fill[lowerteal] (4.8,4.4) circle (1.8pt);
\fill[gridblue] (5.6,5.1) circle (1.8pt);
\draw[flow] (5.8,4.8)--(9.55,4.8);
\node[small,text=black!60] at (8.05,5.12) {Zoom};
\path[fill=gridblue!4,draw=gridblue!65,line width=.8pt] (9.9,3.65) rectangle (13.1,6.85);
\draw[<->,black!50,line width=.55pt] (9.9,7.3)--(13.1,7.3);
\node[small] at (11.5,7.62) {$\overline{s}-\underline{s}\le w$};
\draw[<->,black!50,line width=.55pt] (9.5,3.65)--(9.5,6.85);
\node[small,rotate=90] at (9.15,5.8) {$\overline{q}-\underline{q}\le w$};
\fill[gridblue] (13.1,6.85) circle (2.6pt);
\fill[lowerteal] (9.9,3.65) circle (2.6pt);
\node[text=gridblue,anchor=north east] at (12.95,6.63) {$(\overline{s},\overline{q})$};
\node[text=lowerteal,anchor=north west] at (9.9,3.47) {$(\underline{s},\underline{q})$};
\node[text=gridblue,font=\fontsize{14}{16}\selectfont] at (11.5,5.3) {$\Set C$};
\draw[->,gridblue,line width=.75pt] (13.25,6.85)--(14.0,6.85);
\node[anchor=west,align=left,text=gridblue] at (14.12,6.75)
 {Upper endpoints\\\textbf{Output entropy upper bound}};
\draw[->,lowerteal,line width=.75pt] (10.05,3.65) to[bend left=12] (14.0,4.08);
\node[anchor=west,align=left,text=lowerteal] at (14.12,4.10)
 {Lower endpoints + SNR correction\\\textbf{Conditional entropy lower bound}};
\node[small] at (11.7,2.7) {$w=\Delta\log\rho$};
\end{tikzpicture}